\documentclass[sigconf, nonacm]{acmart}

\usepackage{amsmath}
\usepackage{graphicx}
\usepackage{subfigure}
\usepackage{multicol}
\usepackage{multirow}
\newtheorem{definition}{Definition}

\newtheorem{proposition}{Proposition}
\usepackage{float}
\usepackage{ulem}
\usepackage{wrapfig}
\usepackage[font={small}]{caption}

\usepackage{enumitem}
\setlist[itemize]{leftmargin=2.2em}

\usepackage[ruled, linesnumbered, noresetcount, fillcomment, noend]{algorithm2e}

\SetCommentSty{algocommentfont}
\newcommand{\rtcp}[1]{\tcp*[r]{#1}}
\makeatletter
\patchcmd{\@algocf@start}
{-1.5em}
{0pt}
{}{}
\let\oldnl\nl
\newcommand{\nolinum}{\renewcommand{\nl}{\let\nl\oldnl}}
\makeatother

\definecolor{auburn}{rgb}{0.43, 0.21, 0.1}
\definecolor{pastelblue}{rgb}{0.0, 0.35, 0.9}

\newcommand{\flac}{\ensuremath{\textsf{FLAC}}}
\newcommand{\flacFF}{\ensuremath{\textsf{FLAC}_{\mathsf{FF}}}}
\newcommand{\flacCF}{\ensuremath{\textsf{FLAC}_{\mathsf{CF}}}}
\newcommand{\flacNF}{\ensuremath{\textsf{FLAC}_{\mathsf{NF}}}}
\newcommand{\rlsm}{\ensuremath{\mathrm{RLSM}}}
\newcommand{\result}[1]{\ensuremath{\langle #1 \rangle}}

\def\pracWebsite{https://github.com/nusdbsystem/prac}

\def\levelFF{\ensuremath{L_{\mathsf{FF}}}}
\def\levelCF{\ensuremath{L_{\mathsf{CF}}}}
\def\levelNF{\ensuremath{L_{\mathsf{NF}}}}
\def\numRunCF{\ensuremath{\alpha_{\mathsf{CF}}}}
\def\numRunNF{\ensuremath{\alpha_{\mathsf{NF}}}}
\def\levelSet{\ensuremath{\mathcal{L}}}
\def\eventSet{\ensuremath{\mathcal{E}}}
\def\resultsT{\ensuremath{\mathcal{R}_T}}
\def\nodesT{\ensuremath{\mathcal{C}_T}}
\def\S{\ensuremath{\mathcal{S}}}

\def\no{\ensuremath{No}}
\def\yes{\ensuremath{Yes}}
\def\abort{\ensuremath{Abort}}
\def\commit{\ensuremath{Commit}}
\def\undecided{\ensuremath{Undecide}}

\def\pracWebsite{https://github.com/nusdbsystem/FLAC}

\begin{document}
\title[\flac: A Robust Failure-Aware Atomic Commit Protocol for Distributed
  Transactions (Extended Version)]{\flac: A Robust Failure-Aware Atomic Commit Protocol\\ for
  Distributed Transactions (Extended Version)}

\settopmatter{authorsperrow=3,printfolios=false}

\author{Hexiang Pan}
\affiliation{%
  \institution{National University of Singapore}
}
\email{panh@u.nus.edu}

\author{Quang-Trung Ta}
\affiliation{%
  \institution{National University of Singapore}
}\email{taqt@comp.nus.edu.sg}

\author{Meihui Zhang}
\affiliation{%
  \institution{Beijing Institute of Technology}
}
\email{meihui\_zhang@bit.edu.cn}

\author{Yeow Meng Chee}
\affiliation{%
  \institution{National University of Singapore}
}
\email{ymchee@nus.edu.sg}

\author{Gang Chen}
\affiliation{
 \institution{Zhejiang University}
 }
 \email{cg@zju.edu.cn}

\author{Beng Chin Ooi}
\affiliation{%
  \institution{National University of Singapore}
}\email{ooibc@comp.nus.edu.sg}

\begin{abstract}
In distributed transaction processing, atomic commit protocol (ACP) is used to
ensure database consistency.
With the use of commodity compute nodes and networks, failures such as system
crashes and network partitioning are common.
It is therefore important for ACP to dynamically adapt to the operating
condition for efficiency while ensuring the consistency of the database.
Existing ACPs often assume stable operating conditions,
hence, they are either non-generalizable to different environments or slow in
practice.

In this paper, we propose a novel and practical ACP, called
Failure-Aware Atomic Commit ({\flac}).
In essence, {\flac} includes three protocols, which are specifically
designed for three different environments: (i) no failure occurs, (ii)
participant nodes might crash but there is no delayed connection, or (iii) both
crashed nodes and delayed connection can occur.
It models these environments as the failure-free, crash-failure, and
network-failure robustness levels.
During its operation, {\flac} can monitor if any failure occurs and dynamically
switch to operate the most suitable protocol, using a robustness level state
machine, whose parameters are fine-tuned by reinforcement learning.
Consequently, it improves both the response time and throughput, and effectively
handles nodes distributed across the Internet where crash and network failures
might occur.
%
We implement {\flac} in a distributed transactional key-value storage system
based on Google Percolator and evaluate its performance with both a micro
benchmark and a macro benchmark of real workload.
The results show that {\flac} achieves up to 2.22x throughput improvement and
2.82x latency speedup, compared to existing ACPs for high-contention workloads.
\end{abstract}


\maketitle

\section{Introduction}



Database systems that support critical operations require all transactions
to guarantee the ACID properties (atomicity, consistency, isolation, durability)
\cite{gray1992transaction}.
Among them, the atomicity property, which mandates a transaction
must happen in its entirety, is costly to guarantee, especially in sharded or
distributed databases.
Sharded databases partition states into disjoint computing nodes for horizontal
scalability.
Consequently, when a transaction's accesses span across nodes, a new challenge
arises on how nodes agree on a common status of the transaction.
To address this problem, atomic commit protocols (ACPs) such as
  two-phase commit (2PC) \cite{2PC} and three-phase commit (3PC) \cite{3PC} have been
proposed.
%
%
However, one fundamental drawback of ACPs is that they work with a fixed assumption
on the node behavior and network connectivity, i.e., the operating conditions.
For example, 2PC only ensures termination in failure-free environments,
rendering services unavailable even in the event of two-node crashes \cite{2PC,
  gupta2018easycommit}.
3PC overcomes this problem with an additional message round trip \cite{3PC,
  gupta2020efficient}.
%
%
However, the overhead of this additional message round trip
often overshadows the benefit of the termination property, especially when the
failure-free run constitutes the bulk of system operating time.

To alleviate the above problems, we design {\flac}, a
\underline{\textbf{f}}ai\underline{\textbf{l}}ure-aware
\underline{\textbf{a}}tomic \underline{\textbf{c}}ommit protocol, which can
monitor and adapt its execution to failures happening in the operating
environments to improve its performance.
Similar to existing ACPs~\cite{ML12PC, presumedeither2PC, 2PC, 3PC,
  gupta2018easycommit}, our system model consists of a coordinator node and
different participant nodes.
The coordinator drives the protocol while participants vote on the transaction
status.
To make {\flac} operate efficiently in environments where different types of
failures can occur, we design it to handle transactions from clients in four
steps.
First, the coordinator sends \textit{Propose} messages to participants with its
current assumption on operating conditions, so that participants can select the
corresponding protocol according to the hinted conditions.
Second, the participants execute the selected protocol and send the execution
results back to the coordinator.
Third, the coordinator performs a validation to detect whether any unexpected
failure happens during the execution so that it can adjust {\flac} to adapt to
the new environment.
This step is conducted by a component called the robustness-level state machine
({\rlsm}) resided inside the coordinator.
Finally, the coordinator makes the final decision and sends the results to all
undecided participants so that they can commit or abort the transaction
accordingly.
In our work, we follow state-of-the-art approaches \cite{2PC, 3PC,
  gupta2018easycommit, gupta2020efficient, HFC} to consider two types of
failures that can happen in the operating environment of ACPs.
Consequently, three types of environments can occur during the execution of
{\flac}, where (i) no failure occurs; (ii) participant nodes might crash but
there is no delayed network connection, or (iii) both crashed nodes and delayed
connections occur.
We have implemented in {\flac} a family of three protocols that can operate
effectively in the three aforementioned environments.
These protocols correspond to the three robustness levels: failure-free,
crash failure, and network failure, based on the expected failure set defined
in~\cite{HFC}.
{\flac} can switch between the three protocols based on closely monitoring
its operating environment.
In any setting, {\flac} can always ensure the safety~\cite{gupta2018easycommit}
that no two participants will decide differently, and a {\commit} decision is
made only if all participants have voted {\yes}.
Furthermore, it can guarantee the liveness property when the correct
protocol is selected for a given operating condition.
Finally, when there is a change in the operating condition (e.g., a node crashes or recovers),
the {\rlsm} component in the coordinator will adjust each node's robustness level accordingly.
%
%
%
More specifically, when the operating condition changes, the previously selected protocol may not be appropriate for the new environment.
This incorrect selection of protocol can cause potential blocking of transactions.
To alleviate the blocking problem,
we equip \rlsm~ with RL-tuned parameters to capture failure's degree of recurrence, and let {\rlsm} cautiously adjust nodes' robustness levels to lenient ones.
As a result, {\flac} can consistently provide high performance in all environments.

In summary, we make the following contributions:

\begin{itemize}[nolistsep, left=0.5em]
  \item We propose {\flac}, a robust failure-aware commit
  protocol for the transaction atomicity under different operating conditions.

  \item We design a robustness-level state machine ({\rlsm}), a provable state
  transition model to infer operating conditions for {\flac}.

  \item We apply reinforcement learning methods to fine-tune parameters of {\rlsm}, which improves the
  performance of {\flac} in unstable environments.

  \item We implement a key-value store system with transaction support based on Google Percolator \cite{GOO} as a test bed to evaluate {\flac} and other atomic commit protocols (ACPs).

  \item We conduct a comprehensive evaluation of {\flac} with state-of-the-art ACPs using the above system with
  YCSB-like micro-benchmark and TPC-C workloads. Our experiments show
  that {\flac}'s ability to adapt to different operating conditions results in a significant performance gain compared to other ACPs.
\end{itemize}

The rest of this paper is organized as follows.
Section \ref{background} introduces the background of ACPs.
%
Section \ref{sec:systemModel}, \ref{sec:subProtocols} present the design of our
{\flac} protocol: its system model, protocols, and transaction processing
system.
Section \ref{rlsm} introduces the failure detector {\rlsm}, and shows how
{\flac} combines three protocols with it. 
Section \ref{exp} exhibits our performance evaluation of {\flac} and existing ACPs.
We then review related works in Section~\ref{related}, discuss the limitations
of {\flac} and future work in Section \ref{sec:short-discussion}, and finally
conclude in Section~\ref{conclusion}.


\section{Background}
\label{background}



\subsection{Non-Blocking Atomic Commit}
\label{nbac}
The atomic commit problem is for a set of nodes in a distributed database system
to decide to commit or abort a transaction~\cite{DataI}.
The decision is always based on the ``votes'' from nodes indicating their local
checks for the transaction.
A node votes {\no} if it wants to abort a transaction due to local invalid
executions that violate ACID properties, process crashes, etc. Otherwise, the
node votes {\yes} to inform other nodes it is ready to commit. Once collecting
enough information, each node will decide to commit or abort locally.

Existing research defines the correctness and liveness of
  atomic commit protocols (ACPs) as three properties: \textit{agreement},
  \textit{validity}, and \textit{termination} \cite{NBAF, WFD, Skeen, HFC}.
  We formally describe these properties in Definition \ref{def1}.
  An atomic commit protocol (ACP) is always expected to satisfy the
  \textit{agreement} and \textit{validity} properties for protocol correctness,
  while a non-blocking ACP also needs to ensure the \textit{termination} property
  for liveness (Definition \ref{def1}).

\begin{definition}[Non-Blocking Atomic Commit \cite{NBAF, WFD, Skeen, HFC}]
  \label{def1}
  Let $\S$ be a distributed system of $n$ nodes $C_1, ..., C_n$ and $\Pi$
  be an atomic commit protocol (ACP) defined by two events:

  \begin{itemize}[leftmargin=1.5em]
    \item Propose: $C_1, ..., C_n$ proposes values $v_1, ..., v_n$
    representing the vote ($v_i$ is either {\no} or {\yes}).

    \item Decide: $C_1, ..., C_n$ outputs the decided values $d_1, ...,
    d_n$ ($d_i$ is either {\abort} or {\commit}).
  \end{itemize}

  \noindent An execution of $\Pi$ is said to solve the non-blocking atomic
  commit problem iff it satisfies three properties:

  \begin{itemize}[leftmargin=1.5em]
    \item Agreement: no two nodes decide differently ($\forall i, j.~ d_i =
    d_j$).

    \item Validity: a node decides {\abort} only if some nodes propose {\no} or
    a failure occurs; it decides {\commit} only if all nodes propose {\yes}.

    \item Termination:
    every correct node will eventually decide regardless of failures.
    \end{itemize}

    \noindent Finally, $\Pi$ is said to be a non-blocking ACP in $\S$ if every
    execution of $\Pi$ in $\S$ solves the non-blocking atomic problem.
\end{definition}

\subsection{The Failure Model and Safe Assumptions}

\label{failures}

Real systems often face different kinds of failures, such as process crashes,
message losses, and Byzantine faults~\cite{demillo1982cryptographic, IMPOS}.
It is therefore not easy to consider all possible failures when designing systems.
Instead, the designer often proclaims a set of expected failures as basis.
In this work, we adopt the failure model used by existing ACPs \cite{2PC, 3PC,
  12PC, gupta2018easycommit, Skeen, HFC} that crash failure is considered and
point-to-point network communications never fail but could be
delayed.
We also employ their assumption that the receiver node can reliably detect the
crash failure of the sender node by a timeout, which makes it possible
to differentiate between crashed nodes and slow nodes~\cite{IMPOS,
  DBLP:phd/us/Chandra93}.
We call this timeout as \textit{crash timeout}, i.e., the timeout
  which is used to detect crash failures.
%
%
Moreover, we also follow \cite{yan2020domino} to use \textit{network timeout} to
accommodate non-delayed messages that expose low variability of delays.
Note that \textit{network timeout} is smaller than \textit{crash
    timeout} because the former only accommodates non-delayed messages while the
  latter accommodates all messages.

In this work, we follow \cite{HFC} to classify the operating conditions into
three categories:

\begin{itemize}[nolistsep]
  \item[(i)] Failure-free environments: the environments which do not contain any
  crashed nodes or delayed messages.

  \item[(ii)] Crash-failure environments: the environments in which nodes might crash
  but messages are never delayed.

  \item[(iii)] Network-failure environments: the environments where the crashed nodes and
  delayed messages can both occur.
\end{itemize}

The above classification enables us to implement two optimizations in
ACPs for the failure-free and crash-failure environments.
%
First, in a failure-free environment, an ACP does not need to consider crash
failures, thus can be optimized further for higher efficiency.
Second, in both failure-free and crash-failure environments where all
connections are non-delayed, a node can try to synchronize with others by
broadcasting a message and waiting for a short network
  timeout~\cite{yan2020domino} (one network delay), instead of sending
round-trip messages to synchronize (two network delays) like other
ACPs \cite{12PC, 2PC, 3PC, gupta2018easycommit}.
Moreover, since all messages
arrive within a network timeout,
the message receiver can quickly assert
the crash failure of the sender.
It gives optimization opportunities for the ACP executions when some nodes
crash.


There are also other failure models considered in distributed systems, such as
timing failures \cite{frolund1998quality, ouyang2016reducing}, gray failures
\cite{huang2017gray}, Byzantine failures \cite{castro2002practical,
  ruan2019fine, maiyya2020fides}.
A timing failure occurs when servers return correct but untimely responses
~\cite{cristian1991understanding}.
ACPs can treat timing failures occurring within the crash timeout
  as network failures and other timing failures as crash failures.
In contrast, gray failures happen when a system's unhealthy behavior is
observed by users but not by the system itself, and a Byzantine failure occurs
when nodes behave arbitrarily or send misleading messages.
These failures cannot be tolerated by both {\flac} and existing ACPs.
~\cite{2PC, 3PC, 12PC, gupta2018easycommit, Skeen, HFC}


\section{System Overview}
\label{sec:systemModel}

This section presents an overview of our {\flac} protocol: its
  architecture and transaction processing system.

%

\subsection{Architecture}

We design {\flac} using a typical transactional system architecture,
where one node is designated as the coordinator, while the other nodes serve as
participants.
Its system model is depicted in Figure~\ref{fig:STA}, where $C^*$ is
the coordinator and $\S =\{C_1, ..., C_n\}$ are the participants.
Each participant manages a local data store, while the coordinator serves as a
transaction coordinator which accepts and manages transactions from clients.
A transaction $T$ may read or update data stored in multiple nodes
which directly participate in $T$'s execution.
We denote these participant nodes as $\nodesT$ $(\nodesT \subseteq \S)$.

\begin{figure}[htbp]
  \centering
  \includegraphics[width=0.9\linewidth]{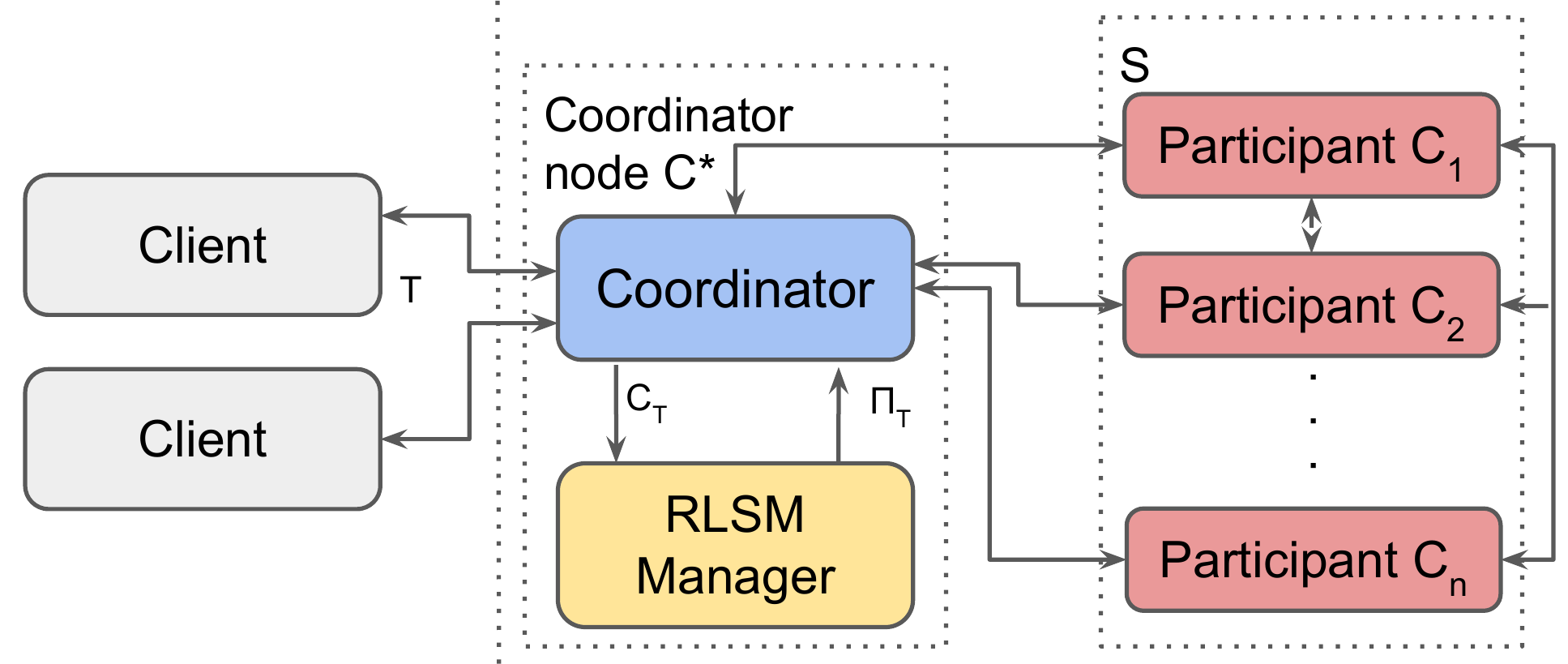}
  \caption{The {\flac} system model.}
  \label{fig:STA}
\end{figure}%


Furthermore, {\flac} consists of three different protocols, namely
  {\flacFF}, {\flacCF}, {\flacNF}, which are dedicated respectively to a
failure-free, a crash-failure, and a network-failure environment.
In each execution, the coordinator $C^*$ first selects the most suitable
protocol, denoted as $\Pi_T$, based upon operating conditions of the
participants $\nodesT$.
Then, all nodes in $\nodesT$ will jointly execute $\Pi_T$ and try to reach their
local decisions for $T$.
Afterwards, they reply with their execution results to $C^*$.
These results, denoted as $\resultsT$, include the nodes' initial votes and
decisions.
If a node fails to make a decision ({\commit}/{\abort}), it is
called the \textit{undecided} participant.
Finally, the coordinator $C^*$ will decide to commit/abort the transaction based
on the collected results and send that decision to all undecided participants.

Since we consider both crash and network failures (Section \ref{background}),
some participants in $\nodesT$ might crash, or their replied results might not
reach the coordinator within the network timeout due to delay.
All other participants whose results can be well received by the coordinator are
called the responsive participants, and their results contribute to the result
set $\resultsT$.
Furthermore, each of the responsive participants can either (i) vote {\yes} and
decide {\commit}, or (ii) vote {\yes} and decide {\abort}, or (iii) vote {\yes}
and do not decide, or (iv) vote {\no} and decide {\abort}.
We denote these cases as $\result{\yes, \commit}$, $\result{\yes, \abort}$,
$\result{\yes, \undecided}$, and $\result{\no, \abort}$
\footnote{The two cases $\result{\no, \commit}$ and $\result{\no, \undecided}$
  will never occur since participants which vote {\no} will always decide
  {\abort} immediately (by the validity property in Definition~\ref{def1}).}.
The necessity of collecting both votes and decisions to keep protocol safe and alive is justified in Section~\ref{upward} and Section~\ref{downwd}.

We also equip the coordinator node $C^*$ with a component called robustness-level
state machine (RLSM) manager, which is responsible for determining the suitable
protocol $\Pi_T$ to execute $T$.
More specifically, we first represent three operating environments in
Section~\ref{failures} as three robustness levels: failure-free ($\levelFF$),
crash failure ($\levelCF$), and network failure ($\levelNF$).
%
During a transaction execution, {\flac} analyzes the results reported by
participant nodes to detect if any failure occurs.
If there is an unexpected failure occurring in a participant node
(e.g., network failure or crash failure), it will upgrade this
node's robustness level to a more stringent level.
Otherwise, {\flac} will record the total number of consecutive transactions
which have been executed successfully in this node.
If this number reaches a threshold representing the environment's stability,
{\flac} will downgrade the node's robustness level to a more lenient level.
Later in Section \ref{rlsm}, we shall describe how to use reinforcement learning
to fine-tune this threshold so that {\flac} can adapt well to changes of the
operating environment.








In our system model, we follow existing works~\cite{gupta2018easycommit,
  skeen1982quorum, 3PC, 2PC} to assume \textit{reliable} but \textit{delayable}
connections between nodes.
%
%
%
Furthermore, although we currently separate the coordinator from the
participants, {\flac} can be easily extended to the setting where the
coordinator also serves as a participant node, like the system model proposed in
\cite{lamport2006fast}.
This setting is simply an optimized version of our current system model, in
which the connection between that participant and the coordinator happens in
the same node.
%
%

We summarize the notations used in this paper in Table~\ref{tab:notation}.

\begin{table}[ht]
  \centering
  \caption{Symbols and descriptions}
  \vspace{-1em}
  \label{tab:notation}
  \begin{tabular}{llll}
    \toprule
    $\S$ & the system of $n$ participant nodes $C_1, ..., C_n$ \\

    $C^*$ &  the coordinator node \\

    $T$ &  the transaction sent to $C^*$ \\

    $\nodesT$ & the participant nodes contributing to the commit of $T$ \\

    $\resultsT$ & the tentative execution results of $\nodesT$ for $T$ \\

    $\Pi_T$ &  the protocol selected for $T$ \\

    \bottomrule
  \end{tabular}
\end{table}

\subsection{Transaction Processing}
\label{steps}

%
{\flac} requires two phases to process a transaction $T$: the \textit{propose}
and the \textit{validate} phase.
They are described as follows.
%

\subsubsection{The propose phase.}
\label{propose}

The coordinator $C^*$ first selects a protocol $\Pi_T$ based on the most
stringent robustness level of participant nodes $\nodesT$ of $T$.
Then, it instructs all these nodes to execute $\Pi_T$ to process $T$ tentatively
or persistently.
Next, relevant participant nodes $\nodesT$ execute $\Pi_T$ trying to make local
decisions for $T$.
There are two scenarios for each participant node as follows.
Firstly, if the participant can make an {\abort} or {\commit} decision, it will
complete its local execution for $T$, add the decision in its log, and  roll back or persist the local
changes respectively.
Secondly, if the participant cannot make a decision (i.e., stays {\undecided}),
it will block its local transaction execution and wait for the message from the
coordinator or other participants to continue.
In any cases, all participant nodes in $\nodesT$ will eventually
reply their final states ({\commit}/{\abort}/{\undecided}) and initial votes
({\yes}/{\no}) to $C^*$.
The collection of decisions and votes constitutes the set of results
$\resultsT$, which is processed by {\flac} in the second phase:
the \textit{validate} phase.

\begin{figure}[htbp]
  \centering
  \includegraphics[width=\linewidth]{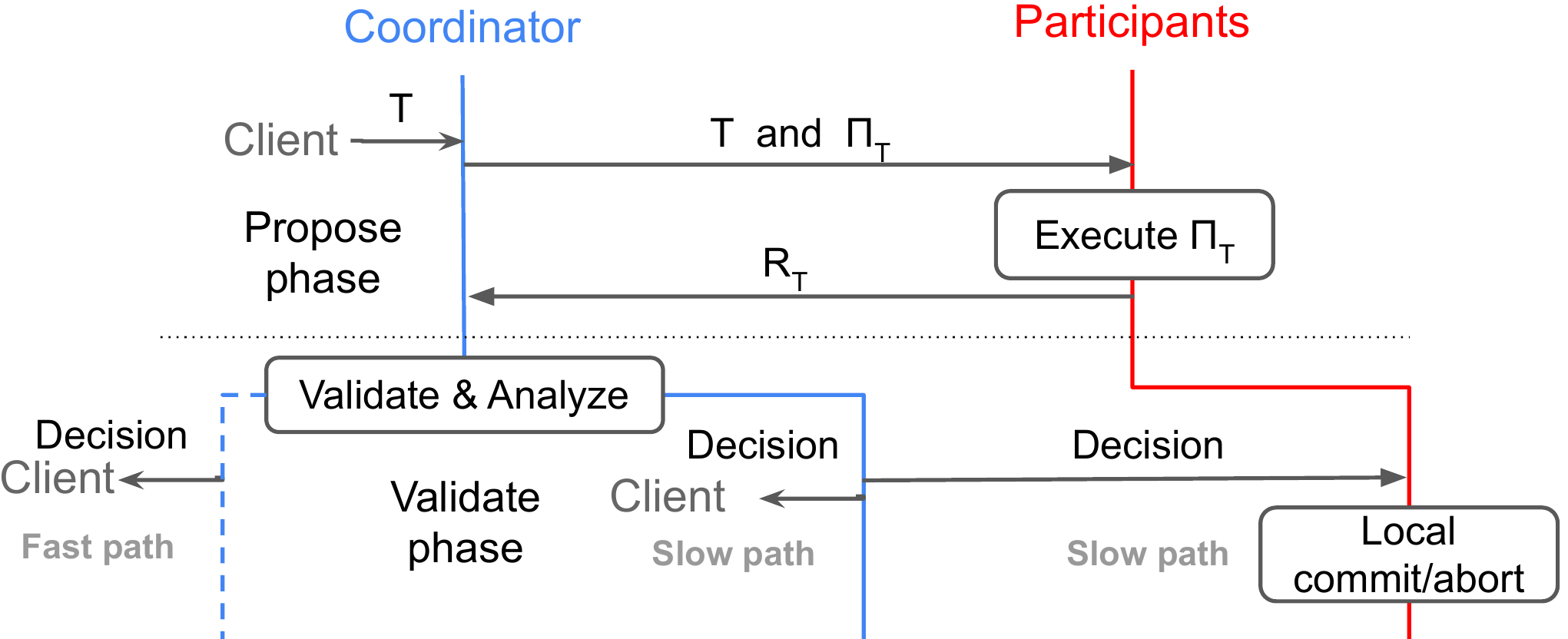}
  \caption{Overall timeline of {\flac}. \textit{Fast} path happens when all
    participants make an {\abort} or {\commit} decision.
    \textit{Slow} path occurs when some participants stay {\undecided}.
  }
  \label{fig:SAB}
\end{figure}

\subsubsection{The validate phase.}
\label{validate}
In this phase, the coordinator $C^*$ performs two tasks concurrently to process
the result set $\resultsT$.
Firstly, $C^*$ analyzes the results in $\resultsT$ to detect if any failure
occurs and adjusts the robustness level of nodes in $\nodesT$ accordingly.
We will explain in detail this adjustment in Section~\ref{rlsm}.
Secondly, it commits or aborts the transaction $T$ in a \textit{fast}
  or a \textit{slow execution path}.


In particular, a \textit{fast path} is resulted when all nodes in $\nodesT$ have
decided to abort or commit the transaction $T$.
Since $T$ is already rolled back or made persist in the
  \textit{propose} phase at each participant node (Section~\ref{propose}), the participants have finished their executions of $T$, and thus the
  coordinator $C^*$ only needs to send the decision to clients and terminate
  the transaction $T$.
This path is illustrated by the blue dashed line in Figure~\ref{fig:SAB}.

In contrast, a \textit{slow path} is resulted when not all
participant nodes $\nodesT$ have decided to abort/commit the transaction $T$.
In this case, the coordinator $C^*$ dictates a decision according to $\resultsT$
and $\Pi_T$, and then notifies all participants.
Each participant, upon receiving the decision, will locally commit or abort $T$
on its storage.
After broadcasting all messages to notify the decision, the coordinator then
reports the decision to clients.
This path is illustrated by the red and blue solid lines; both labeled by
\textsf{``Slow path''} in Figure~\ref{fig:SAB}.

{\flac} executes most transactions in \textit{fast path} for efficiency. The reason is two-fold. On the coordinator side, the \textit{slow path} requires it to broadcast
  decisions to all {\undecided} participants, which slightly increases the
  latency.
  On the participants' side, the \textit{slow path} requires each participant to
  block and wait for decisions from other nodes, which introduces significant
  contention between concurrent transactions and reduces the throughput.
Therefore, to achieve high performance, {\flac} tries to execute most
transactions in \textit{fast path}.
%


\section{Dedicated Protocols for Different Operating Environments}

\label{sec:subProtocols}

This section describes the three dedicated protocols {\flacFF}, {\flacCF},
{\flacNF} that our {\flac} protocol uses to operate in three different
environments: failure-free, crash-failure, and network-failure.
%
We first present each protocol in detail and then prove their correctness.
Finally, we compare these dedicated protocols with existing ACPs and discuss their
strengths and weaknesses.
\subsection{{\flacFF} for Failure-Free Environment}
\label{LP1}


Distributed systems suffer from two major bottlenecks related to ACP~\cite{DBLP:journals/corr/abs-2206-00623, DBLP:journals/pvldb/YuXPSRD18}.
%
First, the long coordinator-side delay prolongs the latency of distributed transactions.
%
%
Second, the long participant-side delay of ACP could block the execution of
concurrent transactions and limit system parallelism.
%
In the failure-free environment, {\flacFF} assumes no network failures or
crash failures among all participant and coordinator nodes.
This gives us opportunities to address both aforementioned challenges.
By adopting optimizations discussed in Section~\ref{failures}, we design {\flacFF} to take fewer message delays on
both the coordinator and the participants than existing ACPs like 2PC~\cite{2PC} in failure-free environments.

\begin{figure}[ht]
  \centering
  \includegraphics[width=\linewidth]{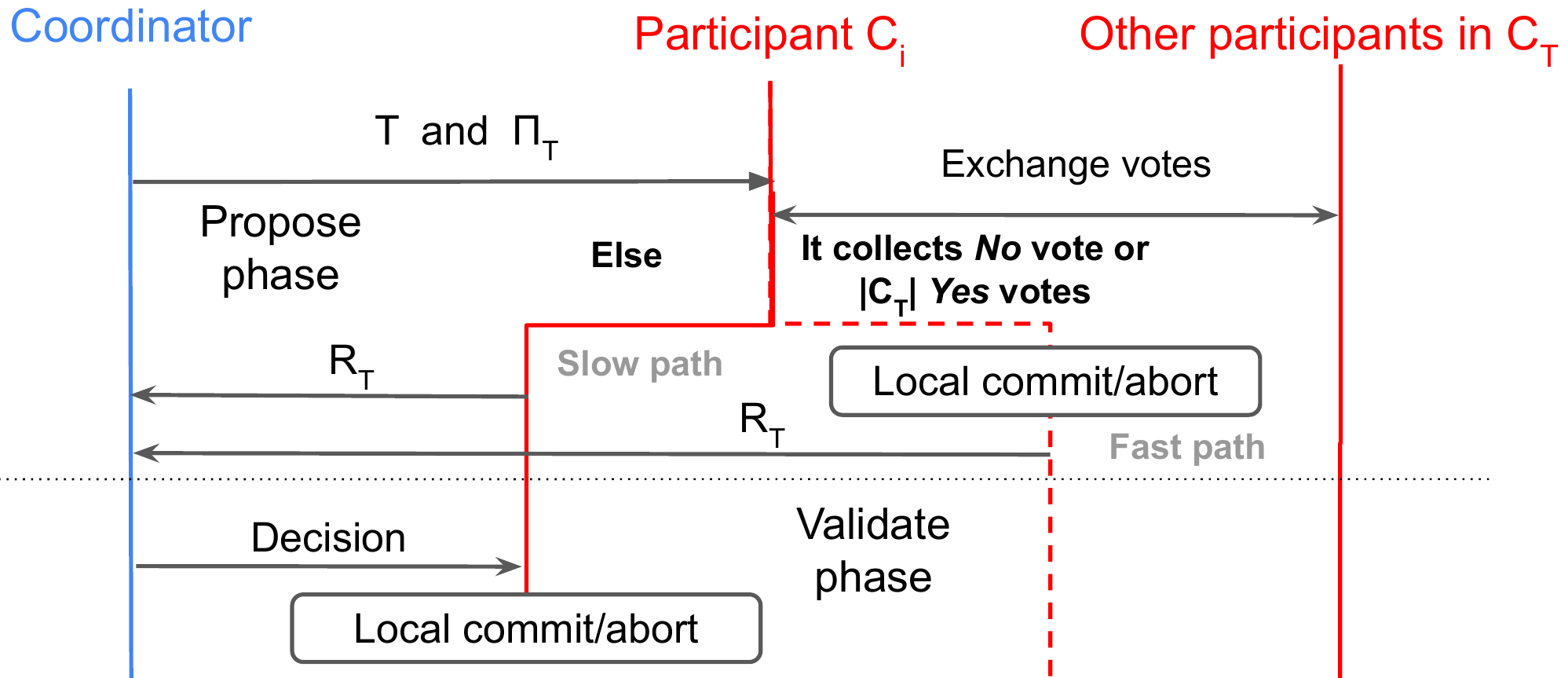}
  \caption{Timeline of {\flacFF}.}
  \label{fig:prac-ff}
\end{figure}

We describe the {\flacFF} protocol for the coordinator and the participants
in Algorithms \ref{alg:CO1} and \ref{alg:LP1}, respectively.
Its timeline is also illustrated in Figure~\ref{fig:prac-ff}.
%
%
The coordinator initiates
an execution of the protocol on receiving a transaction $T$ from
clients.
It sends the $Propose$ message to all involved participants $\nodesT$ and then
waits for their execution results in $W_{C^*}$ time (Algorithm \ref{alg:CO1},
lines \ref{line:CO1Broadcast}--\ref{line:CO1Wait}).

\begin{algorithm}[h]
  \normalem
  \caption{{\flacFF} for the coordinator node $C^*$.}
  \SetAlgoNoLine
  \DontPrintSemicolon
  \label{alg:CO1}
  \KwData{$T$, $\nodesT$ -- the input transaction and its participant nodes.}
  \KwResult{the final global decision for $T$.}
  Broadcast $Propose$ messages with $(T, t_{\textrm{sent}}, C_T)$ to all $\nodesT$
  \label{line:CO1Broadcast}
  \\
  $W_{C^*} \gets max\{ U(C^*, C_i) + U(C_i, C_j) + U(C_j, C^*)| C_i, C_j \in C_T \}$
  \label{line:CO1Wait}
  \\
  $\resultsT \gets $ execution results for $T$ from $C_T$ that arrive in $W_{C^*}$
  time. \\
  \label{line:CO1GetResults}
  \If{$R_T$ \textup{contains} $v = {\abort}$ \textup{or} $v = {\commit}$}{
    \label{line:CO1CommitOrAbort}
    Broadcast $v$ to nodes in $\nodesT$ who have not replied decisions.
    \label{line:CO1Slow} \\
    \Return{$v$} \label{line:CO1FastPath} }
  \If{$|\resultsT| = |\nodesT|$}
    { \label{line:AlgCO1AllUndecided}
      Broadcast {\commit} to all participants in $C_T$ \\
      \Return{\commit} }
    \label{line:AlgCO1Commit}
  \lElse{run termination protocol}
  \label{line:AlgCO1Termination}
\end{algorithm}


\begin{algorithm}[h]
  \normalem
  \caption{{\flacFF} for each participant node $C_i$.}
  \SetAlgoNoLine
  \DontPrintSemicolon
  \label{alg:LP1}
  \KwData{$T$, $\nodesT$: the input transaction and its participant nodes.\\
    \hspace{2.3em} $t_{\textrm{sent}}$: the sent time of the $Propose$ message by $C^*$. \\
    \hspace{2.3em} $t_{\textrm{recv}}$: the received time of the $Propose$ message
    by $C_i$.}
  \KwResult{the final decision of $C_i$ for $T$.}
  $v \gets \textrm{DetermineLocalVote}(T)$ \rtcp{the vote of $C_i$}
  \label{line:AlgLP1DecideVote}
  

  Broadcast $v$ for $T$ to all other nodes in $\nodesT$
  \label{line:AlgLP1Broadcast}\\

  \If{$v$ \textup{is} $\no$}{
    \label{line:AlgLP1VoteNo}
    \textup{Send back} $\result{\no, \abort}$ \textup{to} $C^*$ \\
    \Return{${\abort}$}
    \label{line:AlgLP1VoteNoAndAbort}
  }
  $W_{C_i} \gets max\{t_{\textrm{sent}} + U(C^*, C_j) + U(C_j, C_i) -
  t_{\textrm{recv}} ~|~ C_j \in C_T \}$ \\
  \label{line:AlgLP1ExchangeVotes}
  \If{\textup{receive a} {\no} \textup{for} $T$
    \label{line:AlgLP1VoteYes}
    \textup{from a node in} $\nodesT$  \textup{within} $W_{C_i}$ \textup{time}
    \label{line:AlgLP1ReceiveNo}}
  {
  \textup{Send back} $\result{\yes, \abort}$ \textup{to} $C^*$ \\
    \Return{${\abort}$}
    \label{line:AlgLP1Abort}
  }
  \If{\textup{receive} {\yes} \textup{for} $T$
      \textup{from all nodes in} $\nodesT$  \textup{within} $W_{C_i}$ \textup{time}
    \label{line:AlgLP1ReceiveAllYes}}
  {
      \textup{Send back} $\result{\yes, \commit}$ \textup{to} $C^*$ \\
    \Return{${\commit}$}
    \label{line:AlgLP1Commit}
  }
  \textup{Send back} $\result{\yes, \undecided}$ \textup{to} $C^*$ \\
  \label{line:AlgLP1Sendback}
  \If{\textup{receive a decision $d$ from $C^*$ within} crash timeout \label{line:AlgLP1GetDecisionFromOthers}}{
    \Return{d}
    \label{line:AlgLP1CoordinatorDecision}
  }
  \lElse{
    \textup{run termination protocol}
    \label{line:AlgLP1TerminationProtocol}
  }
  \label{line:AlgLP1Undecided}

  \label{line:AlgLP1Finish}
\end{algorithm}




%
In Algorithm \ref{alg:CO1}, $W_{C^*}$ is the network timeout
  on the coordinator: it is used to accommodate all the results sent back from
failure-free participants.
We compute $W^*$ by utilizing a parameter called the message delay upper bound
$U$.
Given two arbitrary nodes $x$ and $y$, $U(x,y)$ indicates the time window for
all messages between these two nodes to arrive without network failures.
It is computed by multiplying the maximum delay $\sigma(x, y)$ of a ping
message\footnote{Each ping message is a 64-byte ICMP package. We measured this
  ping message delay for 100 times (when the network connection between $x$ and
  $y$ is stable) and took the maximum value as $\sigma(x, y)$.} between $x$ and
$y$\, with a network buffer parameter $r$\footnote{The network buffer parameter
  $r$ is used to adjust the message delay upper bound for different network
  conditions. We will show later in the experiments (Section \ref{effect_r}) how
  different values of $r$, ranging from 0.1 to 8, affect the performance of
{\flac}.}, that is: $U(x,y) = \sigma(x,y) \times r$.
Using this maximum message delay upper bound $U$, we can calculate the
network timeout $W_{C^*}$ by the maximum time for a failure-free
participant $C_i$ to get votes from all others participants $C_j \in \nodesT$ and
replies its message to $C^*$,
that is:  $W_{C^*} = max\{ U(C^*, C_j) + U(C_j, C_i) +
U(C_i, C^*)\}$, for any $C_i, C_j \in C_T$.

For each participant $C_i$, upon receiving the {$Propose$} message, it first
determines and broadcasts its local vote to other participants (Algorithm
\ref{alg:LP1}, lines \ref{line:AlgLP1DecideVote}, \ref{line:AlgLP1Broadcast}).
If this vote is {\no}, it will decide to abort the transaction (lines
\ref{line:AlgLP1VoteNo}--\ref{line:AlgLP1VoteNoAndAbort}).
Otherwise, it will exchange votes with other participants within the
pre-calculated network timeout $W_{C_i}$ (lines
\ref{line:AlgLP1ExchangeVotes}--\ref{line:AlgLP1Finish}).
Here, $W_{C_i}$ is computed by subtracting the latest arrival time of votes sent
from failure-free participants to $C_i$ by the received time of the
\textit{Propose} message on $C_i$.
More specifically, $W_{C_i} = max\{t_{\textrm{sent}} + U(C^*, C_j) + U(C_j, C_i)
- t_{\textrm{recv}}$, for any $C_j \in C_T \}$, where $t_{\textrm{sent}}$,
$t_{\textrm{recv}}$ are respectively the clock time of $C^*$ and $C_i$ when
sending and receiving the \textit{Propose} message\,\footnote{This computation
  might be affected by clock skew and makes $W_{C_i}$ inaccurate, as pointed out
  in \cite{yan2020domino}.
However, such clock skew will not affect the correctness of {\flacFF}.}.

Within the network timeout $W_{C_i}$, the participant $C_i$ will decide
{\abort} if it receives a {\no} vote (lines
\ref{line:AlgLP1ReceiveNo}--\ref{line:AlgLP1Abort}), or decide {\commit} if it
receives {\yes} votes from all participants (lines
\ref{line:AlgLP1ReceiveAllYes}--\ref{line:AlgLP1Commit}).
Otherwise, $C_i$ will remain {\undecided}, send back this result to the
coordinator, and wait for the coordinator's decision (lines
\ref{line:AlgLP1Sendback}--\ref{line:AlgLP1Undecided}).
If it can receive a decision before a crash timeout, it will return that
decision and finish the execution (lines \ref{line:AlgLP1GetDecisionFromOthers}--\ref{line:AlgLP1CoordinatorDecision}).
Otherwise, the coordinator might crash, and $C_i$ will run a termination
protocol to get the decision from other participant nodes in $\nodesT$, or to
reach an agreement among them on $T$'s decision with a new coordinator (line
\ref{line:AlgLP1TerminationProtocol}).

For the coordinator $C^*$, after collecting the result set $\resultsT$, it will
proceed with one of the following three scenarios.
Firstly, if $\resultsT$ includes a {\commit} or {\abort} decision, $C^*$ will
broadcast this decision to other {\undecided} participants, and then returns it
to clients (Algorithm \ref{alg:CO1}, lines
~\ref{line:CO1CommitOrAbort}--\ref{line:CO1FastPath}).
Note that when there is no {\undecided} participant, this execution path is
called the \textit{fast path}, as illustrated in Figure~\ref{fig:SAB}.
Secondly, if $\resultsT$ contains results from all participants who also stay
{\undecided}, then all these participants vote {\yes}\,\footnote{The result $\result{\no, \undecided}$ will never
  occur.}.
Thus, $C^*$ can simply notify a $\commit$ decision to them and return this
decision (lines \ref{line:AlgCO1AllUndecided}--\ref{line:AlgCO1Commit}).
Finally, if $\resultsT$ does not contain the results of all participants, $C^*$
will not have enough information to make a decision.
Then, it will run the termination protocol to continue the execution of $T$
(line \ref{line:AlgCO1Termination}).
$C^*$ will wait for the results of other participants until (i) it can receive
an {\abort} or a {\commit} decision to decide for $T$ accordingly or (ii) it can
receive the same \result{\yes, \undecided} results from all participants and
decide {\commit} for $T$.
%
For crash recovery, {\flac} adds transaction logs into non-volatile storage and recovers from crash failures like existing ACPs ~\cite{2PC, gupta2018easycommit}.
Details about {\flac}'s failure handling can be found in Appendix~\ref{sec:failureHandleApp}.



\vspace{-0.5em}
\subsection{\hspace{-0.5em}\flacCF~for~Crash-Failure~Environment}
\label{LP2}



In the crash-failure environment, \flac~ assumes no delayed message among the
participant nodes.
However, these nodes are subject to crashes.
Therefore, we consider the following two requirements when designing {\flacCF}
to improve its performance.
Firstly, {\flacCF} will always ensure liveness when crash failures happen, to
reduce the blocking of resources \cite{Skeen, 3PC,
  chrysanthis1998recovery}.
Secondly, it will quickly abort transactions that cannot commit due to crash
failures so that they do not affect the commit of other transactions.
%
We present the {\flacCF} protocol in Algorithms \ref{alg:CO2}, \ref{alg:LP2} and
Figure~\ref{fig:prac-cf}.

On the coordinator side, to meet the above requirements, {\flacCF} will quickly
abort a transaction if it cannot receive results from all participants after
the network timeout $W_{C^*}$ (Algorithm \ref{alg:CO2}, lines
\ref{line:AlgCO2NotReceiveAll}--\ref{line:AlgCO2Abort}).
%
%
%
%
It is also the only difference between the {\flacCF} and {\flacFF} protocols for
the coordinator (Algorithm \ref{alg:CO2} vs. Algorithm \ref{alg:CO1}).

\begin{algorithm}[h]
  \normalem
  \caption{{\flacCF} for the coordinator node $C^*$.}
  \setcounter{AlgoLine}{0}
  \SetAlgoNoLine
  \DontPrintSemicolon
  \label{alg:CO2}
  \KwData{$T$, $\nodesT$ -- the input transaction and its participant nodes.}
  \KwResult{the final global decision for $T$.}
  Broadcast $Propose$ messages with $(T, t_{\textrm{send}}, C_T)$ to all $\nodesT$\\
  $W_{C^*} \gets max\{ U(C^*, C_j) + U(C_j, C_i) + U(C_i, C^*)| C_i, C_j \in C_T \}$ \\
  $\resultsT \gets $ execution results for $T$ from $C_T$ that arrive in $W_{C^*}$. \\
  \If{$R_T$ \textup{contains} $v = {\abort}$}{
      \label{line:CO2CommitOrAbort}
     Broadcast $v$ to nodes in $\nodesT$ who have not replied decisions \\
     \Return{$v$}  \label{line:AlgCO2FastPath}
  }
  \If{$|\resultsT| = |\nodesT|$}{ \label{line:AlgCO2Commit}
     Broadcast {\commit} to $C_T$  \\
     \Return{\commit}
  }
  \Else{
    \label{line:AlgCO2NotReceiveAll}
    Broadcast {\abort} to $C_T$ \\
    \Return{\abort}
    \label{line:AlgCO2Abort}}
\end{algorithm}


\begin{algorithm}[htbp]
  \normalem
  \caption{{\flacCF} for each participant node $C_i$.}
  \SetAlgoNoLine
  \DontPrintSemicolon
  \label{alg:LP2}
  \KwData{$T$, $\nodesT$: the input transaction and its participant nodes.\\
    \hspace{2.3em} $t_{\textrm{sent}}$: the sent time of the $Propose$ message by $C^*$. \\
    \hspace{2.3em} $t_{\textrm{recv}}$: the received time of the $Propose$ message
    by $C_i$.}
  \KwResult{the final decision of $C_i$ for $T$.}
  $v \gets \textrm{DetermineLocalVote}(T)$ \rtcp{the vote of $C_i$}
  Broadcast $v$ vote for $T$ to all nodes in $\nodesT$\\
    \label{line:AlgLP2Broadcast}
  \If{$v$ \textup{is} $\no$
    \label{line:AlgLP2VoteNo}}
  {
  Broadcast an {\abort} decision for $T$ to all nodes in $\nodesT$\\
    \label{line:AlgLP2BroadcastAbort}
    \textup{Send back} $\result{\no, \abort}$ \textup{to} $C^*$ \\
    \Return{${\abort}$}
    \label{line:AlgLP2VoteNoAndAbort}
  }
  $W_{C_i} \gets max\{t_{\textrm{sent}} + U(C^*, C_j) + U(C_j, C_i) -
  t_{\textrm{recv}} ~|~ C_j \in C_T \}$\\
    \label{line:AlgLP2VoteYes}
    \If{\textup{receive} {\yes} \textup{from all} $\nodesT$ \textup{in} $W_{C_i}$
        \label{line:AlgLP2AllYes}
      }
    {
      \textup{Send back} $\result{\yes, \undecided}$ \textup{to} $C^*$ \\
       \label{line:AlgLP2TentativeCommit}

       \If{\textup{receive a decision} $d$ \textup{from} $C^* \cup \nodesT$
           \textup{in} crash timeout}
         {
           Broadcast $d$ for $T$ to all nodes in $\nodesT$
         \label{line:AlgLP2SlowPath}\\
           \Return{$d$}
         }
       \lElse{
         {\textup{run termination protocol} \label{line:AlgLP2Termination}}
       }
    }
    \Else{
      \label{line:AlgLP2NotReceiveAllYes}
      Broadcast an {\abort} decision for $T$ to all nodes in $\nodesT$
      \label{line:AlgLP2BroadcastAbort2}\\
      \textup{Send back} $\result{\yes, \abort}$ \textup{to} $C^*$ \\
      \Return{${\abort}$}
      \label{line:AlgLP2VoteYesAndAbort}
    }
\end{algorithm}


On the participant side, the {\flacCF} protocol for each participant (Algorithm
\ref{alg:LP2}) is also similar to the {\flacFF} protocol (Algorithm
\ref{alg:LP1}), except for the following two differences.
Firstly, to satisfy the liveness requirement, we apply the
\textit{transmission-before-decide} technique as proposed in
\cite{gupta2018easycommit}.
Each participant will broadcast their decision to other nodes before rolling back
or persisting the transaction (Algorithm \ref{alg:LP2}, lines
\ref{line:AlgLP2BroadcastAbort}, \ref{line:AlgLP2SlowPath}, and
\ref{line:AlgLP2BroadcastAbort2}).
With reliable connections, the broadcasted decision will eventually reach
every non-crashed node.
Therefore, when a participant is blocked for not receiving any decision, no
participant has decided for $T$.
It can then communicate with a new coordinator to commit or abort the
transaction (line \ref{line:AlgLP2TentativeCommit}).
%
%
%
Secondly, each participant can quickly abort a transaction if it can detect a
crash failure in other participants.
In particular, if the participant $C_i$ does not receive all {\yes} votes from
other participants within the network timeout $W_{C_i}$, it can directly decide
{\abort} (lines
\ref{line:AlgLP2NotReceiveAllYes}--\ref{line:AlgLP2VoteYesAndAbort}).
This processing corresponds to the \textit{fast execution path} of the
participant, as illustrated in Figure \ref{fig:prac-cf}.

In addition to the above requirements, termination protocol will also be
triggered if a participant receives all {\yes} votes from other participants,
but does not receive any {\abort} or {\commit} decision from other participants
or the coordinator within the crash timeout.
In this case, the transmission-before-decide technique ensures that none of these nodes has made a decision for the transaction.
Thus, all alive participants can run the termination protocol to reach an agreement
for the transaction.

\begin{figure}[ht]
  \centering
  \includegraphics[width=\linewidth]{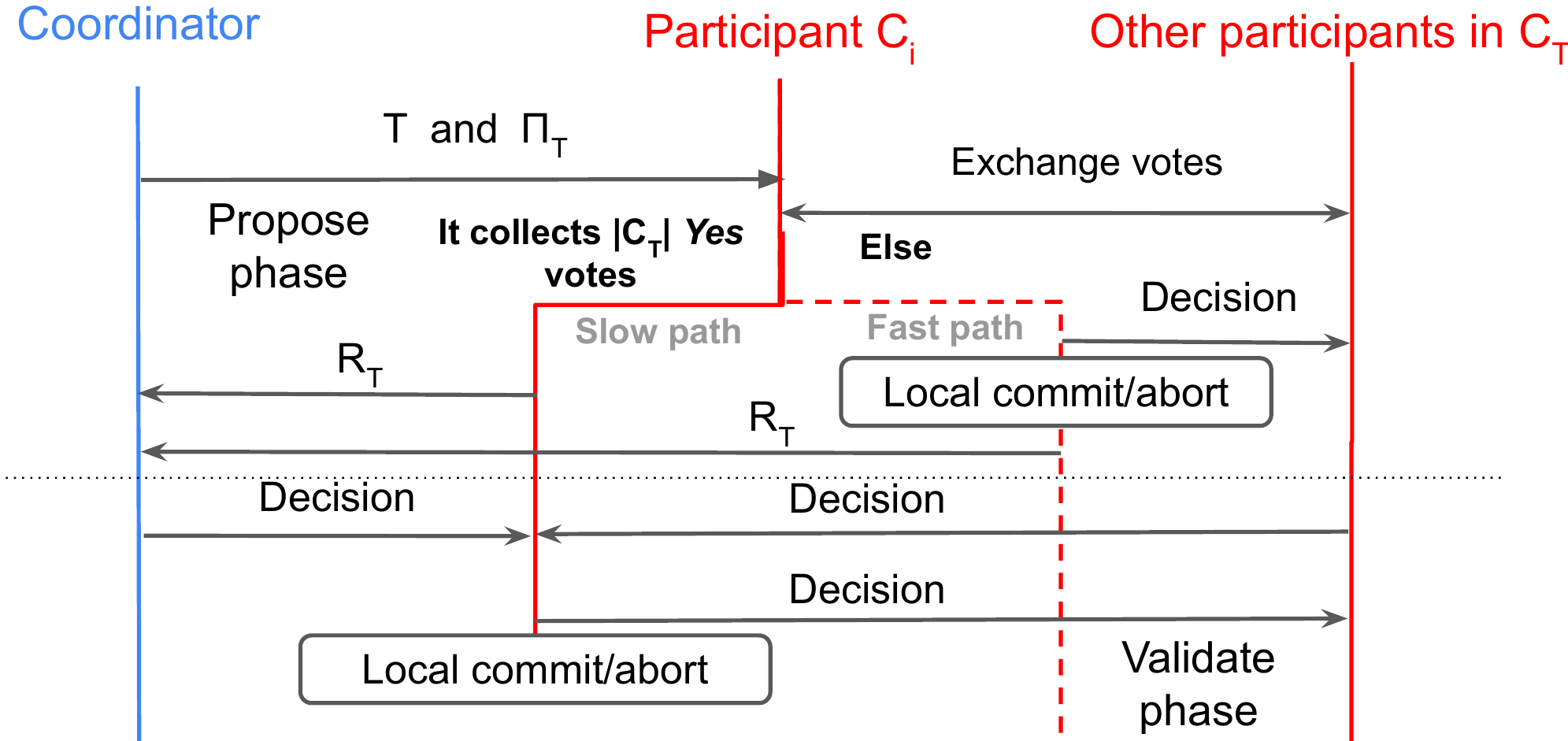}
  \caption{Timeline of {\flacCF}. Each participant forwards its decisions to
    others before committing/aborting the transaction.}
  \label{fig:prac-cf}
\end{figure}

%

\vspace{-1.5em}
\subsection{\hspace{-0.5em}{\fontsize{10.2}{12}\selectfont
    {\flacNF} for Network-Failure Environment}}
\label{LP3}

In this environment, both crash failures and network failures may occur.
There have been several ACPs proposed to handle both these two failures, such as 3PC
\cite{3PC}, PAC \cite{maiyya2019unifying}, and Easy Commit (EC)
\cite{gupta2018easycommit}.
Among them, EC is the state-of-the-art non-blocking protocol, which is well
optimized for the network-failure environment.
%
For this environment, it is challenging to design a new ACP
  that is better than EC.
  Moreover, our goal is not to design a new protocol for the specific
  network-failure environment, but rather to design the {\flac} protocol which
  can adapt well to different operating environments.
  Therefore, we directly adopt the EC protocol as the dedicated protocol
  {\flacNF} for the network-failure environment.
  By using EC, {\flacNF} can directly commit transactions without having to enter
  the \textit{validate} or \textit{propose} phases like {\flacFF} and
  {\flacCF}.
%
Finally, since {\flac}'s design is modular, in the future, if there is a new ACP
that can perform better than EC, we can also adopt it as {\flacNF}.

\begin{table*}[httb]
  \renewcommand{\arraystretch}{0.88}
  \centering
  \caption{Comparison of protocols w.r.t. a transaction $T$, where $N$ is the number of nodes
    contributing to the commit of $T$ ($N = |\nodesT|$).}
  \vspace{-1em}
  \label{tab:cp}
  \begin{tabular}{lccccccc}
    \toprule
    & \flacFF & \flacCF & \flacNF & 2PC & 3PC & PAC & EasyCommit \\
    \midrule
    Number of message delays on coordinator & 3 & 3 & 2 & 4 & 6 & 4 & 2\\
    Number of message delays on participant ({Commit}) & 1 & 3 & 2 & 2 & 4 & 4 & 2\\
    Number of message delays on participant ({Abort}) & 1 & 1 & 2 & 2 & 2 & 2 & 2\\
    Number of message delays on participant ({Crash}) & - & 1 & 2 & 2 & 2 & 2 & 2\\
    Message complexity (number of exchanged messages) & $O(N^2)$ & $O(N^2)$ & $O(N^2)$ & $O(N)$ & $O(N)$ & $O(N)$ & $O(N^2)$\\
    Fault tolerance  & {$\emptyset$} & $\{CF\}$ & $\{CF, NF\}$ & $\{NF\}$ & $\{CF, NF\}$ & $\{CF, NF\}$ & $\{CF, NF\}$ \\
    \bottomrule
  \end{tabular}
  \\
\end{table*}


\vspace{-0.5em}
\subsection{Protocol Correctness}

Like existing atomic commit protocols~\cite{PAX, skeen1982quorum, gupta2018easycommit, maiyya2019unifying}, we prove the correctness of {\flacFF}, {\flacCF}, and {\flacNF} in  safety and liveness.

\begin{proposition}[Safety]
  \label{prop:Safety} All protocols in {\flac} ensure the validity and the agreement
  properties (Definition \ref{def1}).
\end{proposition}

\begin{proof}
%


  (1). The validity property of {\flacFF} is ensured by Algorithm \ref{alg:LP1}:
  this protocol decides {\abort} only if there exists a node voting {\no}
  (lines \ref{line:AlgLP1VoteNo}--\ref{line:AlgLP1VoteNoAndAbort},
  \ref{line:AlgLP1ReceiveNo}--\ref{line:AlgLP1Abort}), and decides {\commit}
  when all nodes vote {\yes} (lines
  \ref{line:AlgLP1ReceiveAllYes}--\ref{line:AlgLP1Commit}), which implies no
  node votes {\no}.

  (2). {\flacFF}'s agreement is proved by contradiction.
  Suppose that there are two nodes deciding differently, then there is a node
  deciding {\commit}, and another deciding {\abort}.
  Since there is one node deciding {\commit}, it follows from Algorithm
  \ref{alg:LP1} (lines \ref{line:AlgLP1ReceiveAllYes}--\ref{line:AlgLP1Commit})
  that all nodes must vote {\yes}.
  Since there is one node deciding {\abort}, according to Algorithm
  \ref{alg:LP1} (lines \ref{line:AlgLP1VoteNo}--\ref{line:AlgLP1VoteNoAndAbort},
  \ref{line:AlgLP1ReceiveNo}--\ref{line:AlgLP1Abort}), there must be a node
  voting {\no}.
  The two previous results contradict each other.

  (3). {\flacCF}'s validity property. According to Algorithm \ref{alg:CO2},
  {\flacCF} decides {\commit} if all nodes vote $\yes$ and no failures happen
  (lines \ref{line:AlgCO2Commit}--\ref{line:AlgCO2NotReceiveAll}). Moreover,
  Algorithm~\ref{alg:LP2} shows that {\flacCF} will decide {\abort} if there
  exists a node voting $\no$ (lines
  \ref{line:AlgLP2VoteNo}--\ref{line:AlgLP2VoteNoAndAbort},
  \ref{line:AlgLP2BroadcastAbort2}--\ref{line:AlgLP2VoteYesAndAbort}). Hence,
  the validity of {\flacCF} is confirmed.

  (4). For {\flacCF}'s agreement, by contradiction, suppose there are two nodes
  deciding differently: node $C_{i}$ decides {\commit} and node $C_{j}$ decides
  {\abort}.
  Since $C_i$ decides {\commit}, according to Algorithm~\ref{alg:CO2} (lines
  \ref{line:AlgCO2Commit}--\ref{line:AlgCO2NotReceiveAll}), no node decides
  {\abort} before the crash timeout.
  Algorithm \ref{alg:LP2} (lines
  \ref{line:AlgLP2TentativeCommit}--\ref{line:AlgLP2SlowPath}) shows that
  $C_{i}$ needs to broadcast its {\commit} decision before deciding {\commit},
  this decision will reach all nodes before the crash timeout, and make them
  decide {\commit}, too.
  This result contradicts with the assumption that $C_{j}$ decides {\abort}.

  (5). For {\flacNF}, it is identical to EC; therefore, its safety is
    guaranteed since EC's safety is proved in \cite{gupta2018easycommit}.
\end{proof}

\begin{proposition}[Liveness]
  \label{prop:Liveness}
  All protocols in {\flac} ensure the termination property when executed in expected operating conditions.
\end{proposition}

\begin{proof}
  {\flacFF} does not need to consider the transaction blocking due to crash
  failures for its failure-free operating condition assumptions.
  %
  {\flacCF} makes progress with its termination protocol in case of
  timeout due to crash failures (Section~\ref{LP2}).
  Finally, {\flacNF}, which is identical to EC, has been proven to be
  non-blocking in ~\cite{gupta2018easycommit}.
\end{proof}

\subsection{Protocol Comparison}

\label{sec:comparison}

We conclude this section by presenting a theoretical comparison among {\flacFF},
{\flacCF}, {\flacNF}, and other ACPs: 2PC \cite{2PC}, 3PC \cite{3PC}, PAC
\cite{maiyya2019unifying}, Easy Commit \cite{gupta2018easycommit} in Table
\ref{tab:cp}.
Here, we calculate the maximum number of message delays that can happen on the
coordinator and participants.
There are three cases related to a {\flac} execution: when it successfully commits a transaction, aborts a
transaction, or crashes before the transaction begins.
%
We also measure the message complexity, and describe which failures
each protocol can tolerate while ensuring liveness.
%
%
%

%
Overall, {\flacFF} and {\flacCF} require fewer message delays on both
the coordinator and the participants,
compared to 2PC, 3PC, and PAC.
As discussed in Section~\ref{LP1}, the fewer message delay enables {\flacFF} and
{\flacCF} to perform better than other protocols in distributed settings.
%
%
Table~\ref{tab:cp} also emphasizes our goal to design {\flac} to adapt well to the changes
in the operating environment.
{\flac} can perform well in different operating conditions: while
{\flacFF} is designed for failure-free environments, {\flacCF} can tolerate
crash failures, and {\flacNF} can endure both crash and network failures.
This is the major difference between {\flac} and the other protocols like 2PC, which
can only support network failures, or 3PC, EC, which can endure crash
and network failures, but might not perform well in the failure-free
environment.
%



\section{Robustness-Level State Machines}
\label{rlsm}




%
We now describe the robustness-level state machine (RLSM) manager of the
coordinator node (Figure \ref{fig:STA}).
Recall that this component is responsible for capturing the robustness level of
each participant node so that the coordinator can determine the suitable
dedicated protocol ({\flacFF}, {\flacCF}, or {\flacNF}) to execute a
transaction.
Following previous works~\cite{DBLP:phd/us/Chandra93,
  DBLP:conf/srds/GuerraouiLS95}, we design {\rlsm} as a lightweight and
practical failure detector that provides enough properties to ensure {\flac}'s
liveness.
To minimize {\rlsm}'s impact on transaction processing, {\rlsm} does not add any
extra network communications (only based on {\resultsT}), and it only adjusts
protocols when liveness is affected.
Note that we employ an RL-based parameter tuner for {\rlsm}'s downgrade
transitions.
To avoid this parameter tuning becoming a new bottleneck, we take its model
training off the critical path and cache the learned results.
In the following, we formalize the RLSM in Definition~\ref{def-rlsm} and
describe its state transitions.

\begin{definition}[RLSM]
  \label{def-rlsm}
  A robustness-level state machine (RLSM) is a state machine denoted as $\langle
  \mathcal{L}, L_0, \mathcal{E}, \delta \rangle$, which comprises of:

  \begin{itemize}[label=--, left=0.5em, nosep]
    \item A set of states $\levelSet = \{\levelFF, \levelCF, \levelNF\}$, in
    which $\levelFF, \levelCF, \levelNF$ respectively represent the
    failure-free, crash-failure, and network-failure robustness levels.

    \item An initial state $L_0 = \levelFF$.

    \item A set of input events $\eventSet = \{ CF, NF, FF(\numRunCF),
    FF(\numRunNF)\}$, in which $CF, NF$ respectively indicate a crash or a
    network failure occurs in a transaction execution, $FF(\numRunCF)$,
    $FF(\numRunNF)$ respectively represent a sequence of $\numRunCF$ and a
    sequence of $\numRunNF$ consecutive failure-free transaction executions at
    the crash- and the network-failure level.

    \item A state transition function $\delta : \levelSet \times \eventSet
    \rightarrow \levelSet$, which is defined as the illustration in
    Figure~\ref{fig:RLSM}.
  \end{itemize}
\end{definition}

\vspace{-0.5em}
\begin{figure}[htbp]
  \centering
  \includegraphics[width=\linewidth]{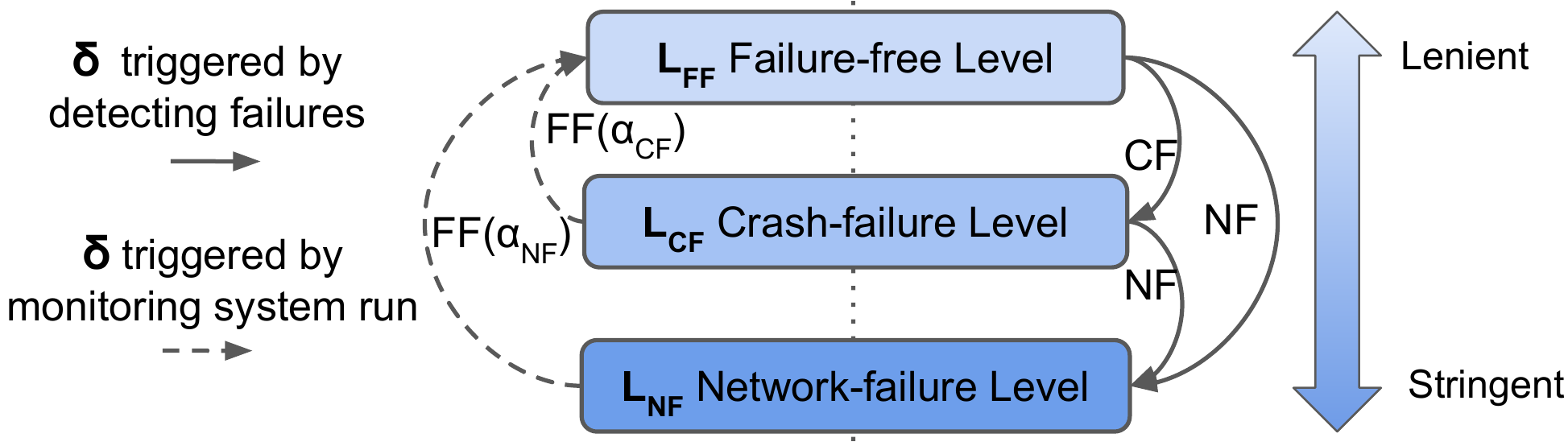}
  \caption{{\rlsm} transition diagram.}
  \label{fig:RLSM}
\end{figure}
\vspace{-0.5em}

Note that in Definition \ref{def-rlsm}, parameter $\numRunCF$ and $\numRunNF$ of
the two events $FF(\numRunCF)$ and $FF(\numRunNF)$ are determined separately for
transaction executions in the crash-failure and the network-failure levels.
Moreover, for each transition state, the RLSM manager can only detect
  events corresponding to the state's intolerable failures.
Specifically, while it can detect the two events $CF, NF$ for {\levelFF}
state, only the event $NF$ can be caught for {\levelCF} state.
Furthermore, none of the two events $CF$ and $NF$ can be tracked for {\levelNF}
state.


Algorithm \ref{alg:rlsm} formalizes how the RLSM manager detects unexpected failures or consecutive failure-free executions and sends the corresponding events (Definition~\ref{def-rlsm}) to the participants' RLSMs.
Firstly, if the RLSM manager can detect a crash or a network failure during a transaction
execution (lines \ref{line:AlgRlsmMissingFF}, \ref{line:AlgRlsmCheckNF4flacFF},
\ref{line:AlgRlsmMissingCF}, \ref{line:AlgRlsmCheckNF4flacCF}), it will trigger
a $CF$ or $NF$ event to update the current state to a more stringent robustness
level {\levelCF} or {\levelNF} (lines \ref{line:AlgRlsmCFDetect},
\ref{line:AlgRlsmAsyNF4FF}, \ref{line:AlgRlsmNF4flacFF},
\ref{line:AlgRlsmAsyNF4CF}, \ref{line:AlgRlsmNF4flacCF}).
This update is called the \textit{robustness upgrade} transition; it
  corresponds to the solid lines in Figure \ref{fig:RLSM}.
Secondly, if the RLSM manager observes a number of failure-free consecutive
executions in a high robustness-level state ($\numRunCF$ times for $\levelCF$
and $\numRunNF$ times for {\levelNF} -- lines \ref{line:AlgRlsmCheck4DownCF},
\ref{line:AlgRlsmCheck4DownNF}), it will respectively trigger the event
$FF(\numRunCF)$ or $FF(\numRunNF)$ to move the system to the lowest
robustness-level state {\levelFF} (lines \ref{line:AlgRlsmDownCF},
\ref{line:AlgRlsmDownNF}).
This update is called the \textit{robustness downgrade} transition;
it corresponds to the dotted lines in Figure \ref{fig:RLSM}.
In the following subsections, we will elaborate on these failure detections and
robustness transitions.

\vspace{-0.5em}
\begin{algorithm}[htbp]
  \normalem
  \caption{State transition algorithm by the RLSM manager.}
  \setcounter{AlgoLine}{0}
  \SetAlgoNoLine
  \DontPrintSemicolon
  \label{alg:rlsm}
  \KwData{$L$, $\nodesT$, $\resultsT$, $\numRunCF$ ,$\numRunNF$: the
    current level, the set of participants, the execution results of the
    transaction $T$ and two parameters of downgrade transitions.}
  \KwResult{Events sent to the participants' RLSMs}
  %

    \If{ $L = {\levelFF}$ \label{line:AlgRlsmFF}}{
        \If{$|\resultsT| < |\nodesT|$
            \label{line:AlgRlsmMissingFF}} {
            Send $CF$ to the RLSMs of non-responsive participants.
            \label{line:AlgRlsmCFDetect}\\
            Send $NF$ to the RLSMs of participants that reply late.
            \label{line:AlgRlsmAsyNF4FF}
        }
        \ElseIf{ $\result{\yes,\undecided} \in \resultsT$
          \label{line:AlgRlsmCheckNF4flacFF}} {
            Send $NF$ to the RLSMs of all participants in $\nodesT$.
            \label{line:AlgRlsmNF4flacFF}
        }
    }
    \ElseIf{$L = {\levelCF}$ \label{line:AlgRlsmCF}}{
        \If{$|\resultsT| < |\nodesT|$
            \label{line:AlgRlsmMissingCF}}
          {Send $NF$ to the RLSMs of participants that reply late.
            \label{line:AlgRlsmAsyNF4CF}
        }
        \ElseIf{$\result{\yes, \abort} \,{\in}\, \resultsT$, \textup{and}
      $\result{\no, \_} \,{\not\in}\, \resultsT$
      \label{line:AlgRlsmCheckNF4flacCF}}{
            Send $NF$ to the RLSMs of all participants in $\nodesT$.
            \label{line:AlgRlsmNF4flacCF}
        }
        \For{$C_i \in $ {\nodesT}}{
            \If{$\textup{NumOfConsecutiveFailureFreeRuns}(C_i) = \numRunCF$}
              {
                \label{line:AlgRlsmCheck4DownCF}
                Send $FF(\numRunCF)$ to the RLSM of $C_i$.
                \label{line:AlgRlsmDownCF}
              }
            }
    }
    \Else(\hfill\tcp*[h]{$L = {\levelNF}$}) {
      \For{$C_i \in $ {\nodesT}}{
          \If{$\textup{NumOfConsecutiveFailureFreeRuns}(C_i) = \numRunNF$}
            {
              \label{line:AlgRlsmCheck4DownNF}
              Send $FF(\numRunNF)$ to the RLSM of $C_i$.
              \label{line:AlgRlsmDownNF}
            }
          }
    }
\end{algorithm}


\vspace{-0.5em}

\subsection{Robustness-Upgrade Transitions}
\label{upward}

The robustness-upgrade transition shifts {\rlsm} of a participant
toward a more stringent robustness level so that subsequent transactions
involving it can run with a more resilient protocol.
We design {\flac} to take a lazy approach that only upgrades the robustness
level when it can detect transaction-blocking failures by analyzing past
execution results.
It is different from traditional approaches \cite{yu2011practical,
  DBLP:conf/ipps/ZhangS08} that continuously monitor the systems to predict the
periods when failure can happen proactively.
%
Such monitoring can be a new bottleneck of the system.
In contrast, {\flac} can detect failures efficiently by analyzing past execution
results as follows:

First, the {\rlsm} manager can detect
  failures for both {\flacFF} and {\flacCF} by checking if the coordinator does not receive all participants'
  results, i.e., some results are missing in {\resultsT}.
For {\flacFF}, the failure type is agnostic before a crash timeout.
In this case, the {\rlsm} manager first inputs CF events to the
RLSMs of non-responsive participants, who are likely to crash (line
\ref{line:AlgRlsmCFDetect}).
Then, it will send NF events to the RLSMs of participants who reply
late (line \ref{line:AlgRlsmAsyNF4FF}), since these late replies are likely
caused by network delays.
For {\flacCF}, it tolerates crash failures, and thus only needs to detect
network failures.
Like the previous case, the RLSM manager will input NF events to the RLSMs of
participants who reply late (line \ref{line:AlgRlsmAsyNF4CF}).

Second, the {\rlsm} manager can detect network failures for both {\flacFF} and {\flacCF} by analyzing the received results.
%
For {\flacFF}, it will report that network failures occur between
participants if the coordinator receives $\result{\yes, \undecided}$ (lines \ref{line:AlgRlsmCheckNF4flacFF}--\ref{line:AlgRlsmNF4flacFF}).
This {\undecided} result indicates that some votes sent to a participant did not arrive on time due to network delay.
For {\flacCF}, it will conclude that a network failure occurs if there
exist aborted transactions that get all {\yes} votes (lines
\ref{line:AlgRlsmCheckNF4flacCF}--\ref{line:AlgRlsmNF4flacCF}).
This is because if there were no network failures and all participants voted {\yes}, the transactions should have been committed.
For both {\flacFF} and {\flacCF}, if a network failure occurred, the {\rlsm} manager cannot identify precisely the involved participants, since the failure can happen among a set of nodes.
Therefore, the {\rlsm} manager will shift all the participants' states to the
network-failure level (lines \ref{line:AlgRlsmNF4flacFF}, \ref{line:AlgRlsmNF4flacCF}).

We prove the correctness of our failure detection rules above in Appendix~\ref{sec:failureDetectionProof}. We also prove the soundness and
completeness of the {\rlsm} state transition algorithm in Appendix~\ref{sec:rlsmCorrectness} and the
correctness of {\flac} in Appendix ~\ref{sec:flacCorrectness}.

\vspace{-0.5em}
\subsection{Robustness-Downgrade Transitions}
\label{downwd}

The robustness-downgrade transitions shift the RLSM state from a stringent level to
a lenient one.
{\flac} triggers these transitions by monitoring the protocol execution.
Specifically, if there is a participant in the crash-failure level who can consecutively
execute $\numRunCF$ transactions without any failure, then the
RLSM manager can transition its RLSM to the failure-free level (lines
\ref{line:AlgRlsmCheck4DownCF}--\ref{line:AlgRlsmDownCF}).
Likewise, $\numRunNF$ is the number of consecutive failure-free transaction
executions that a participant must perform at the network-failure level before
its RLSM can be shifted to the failure-free level (lines
\ref{line:AlgRlsmCheck4DownNF}--\ref{line:AlgRlsmDownNF}).

The two parameters $\numRunCF$, $\numRunNF$ above control the trade-off between
the protocol adaptivity and the state transitioning frequency.
Intuitively, the smaller values of $\numRunCF$ and $\numRunNF$ enable {\flac} to
react faster upon the change of operating conditions, thus getting more
opportunities to run a more lenient protocol.
However, it also lets {\flac} enter the costly \textit{slow path}
when failures recur quickly.
%

\textit{Fine-tuning $\numRunCF$ and $\numRunNF$ by reinforcement learning (RL).}
Generally, it is challenging to fine-tune $\numRunCF$ and $\numRunNF$ to
optimize the performance of {\flac}.
A naive approach is to exhaustively test their possible values, but this trial
and error method is costly.
Also, this test needs to be performed every time the system is reconfigured,
such as, when adding a new node or replacing a machine.
To mitigate this problem, we propose an RL-based optimizer to automatically
select the best $\numRunCF$ and $\numRunNF$.
We employ q-learning ~\cite{watkins1992q} as the RL learner and build our model
as depicted in Figure~\ref{fig:rl}.

\begin{figure}[htbp]
  \centering
  \includegraphics[width=0.9\linewidth]{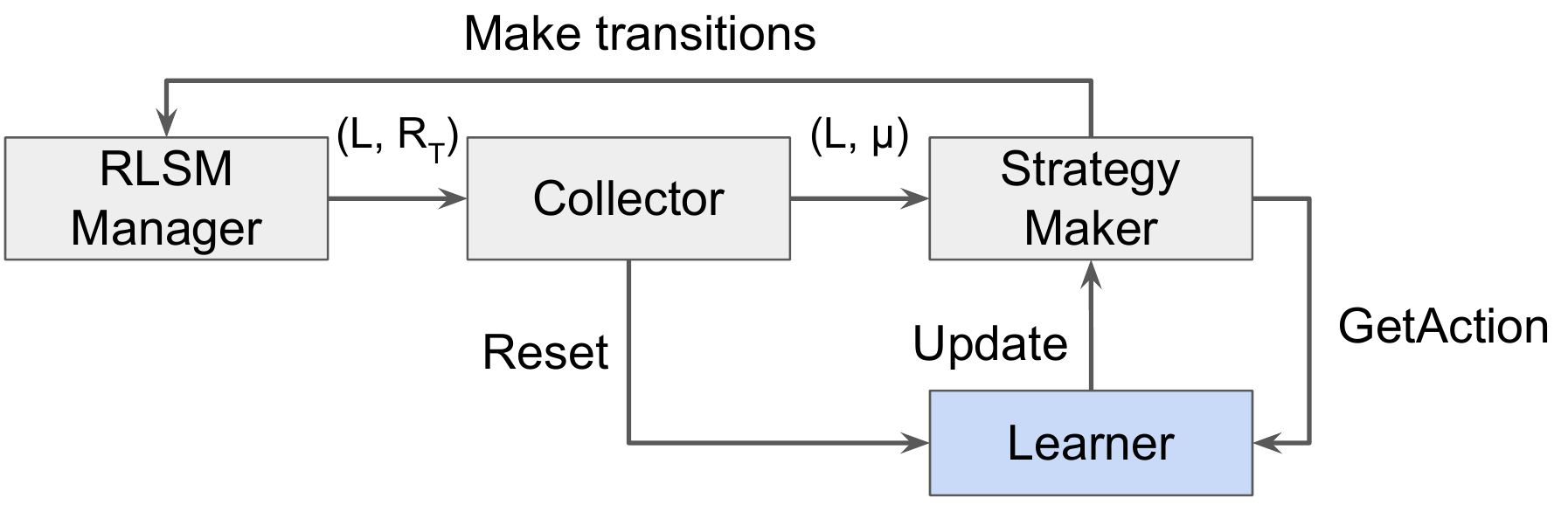}
  \vspace{-0.5em}
  \caption{Overview of the RL-based optimizer,
    where $L, \mu$ denote the current robustness
    level and the reward (average throughput).}
  \label{fig:rl}
\end{figure}

In essence, the RLSM manager asynchronously sends the execution results $\resultsT$ with
their robustness levels $L$ to the collector.
The collector resets the learner's state once it encounters a level change.
Otherwise, it buffers $(L, \resultsT)$ until enough pairs are collected.
It then calculates the average throughput (transaction per second -- tps) as
the reward $\mu$\,\footnote{Users can decide how to choose a suitable reward
  function to obtain good performance, depending on their applications.
  In this work, we use the average throughput.
  Other systems such as ~\cite{decandia2007dynamo, corbett2013spanner,
    yan2020domino} value tail latency, while ~\cite{debnath2010flashstore} value
  throughput, and \cite{li2019qtune, yoon2018mutant} value both and thus
  express their requirements using more complex functions.} and sends $(L, \mu)$
to the strategy maker, who can use this information to train the learner.
The learner learns a Markov chain on whether to continue waiting or
make the downgrade transitions regarding the current state.
{The strategy maker interacts with the learner during training and
  adjusts the participants' {\rlsm} according to the learner's replies.
%
After enough training, the strategy maker can cache the learned strategies as
$\numRunCF$ or $\numRunNF$.
%
%
%

According to our experiment that we will show later in Section
\ref{sec:TurningRLParams}, the throughput gain for parameters greater than 256
is insignificant, and thus we limit the maximum value of $\numRunCF$
and $\numRunNF$ to 256 in our RL model.
With a limited search space, the parameter tuner only takes up to 5 seconds to
tune the parameters.
The users can call the RL-based parameter tuner after a time interval
  (e.g., hourly) or when the system is reconfigured (e.g. replacing a node).

The upward and downward transitions allow a participant's {\rlsm} to
keep its data in memory without durability.
%
%
They ensure the fast correction of {\rlsm} whose robustness level may not match actual operating environments.
Thus, the {\rlsm} can be recovered to an arbitrary robustness level without
concerns about protocol liveness.
Our system sets the default {\rlsm} level to the failure-free level to execute
more transactions with lightweight protocols.
%


\section{Evaluation}
\label{exp}


We have implemented {\flac} in a prototype system using mostly Golang, except
the reinforcement learning component is written in Python.
Our system processes transactions similarly to Google's Percolator~\cite{GOO}.
Each transaction retries when it gets aborted due to conflicts (at most 10
times).
Our system does not retry transactions that cannot commit due to consistency
requirements or system failures.
The source code is published at \cite{flacCode}. 


We have also conducted extensive experiments to evaluate and compare the
performance of {\flac} with
state-of-the-art ACPs: 2PC \cite{2PC}, 3PC
\cite{3PC}, Easy Commit (EC) \cite{gupta2018easycommit}, and PAC
\cite{maiyya2019unifying}.
For the baselines 2PC, 3PC, and EC, we adopted the implementation design in
\cite{gupta2018easycommit}.
For PAC \cite{maiyya2019unifying}, we implemented its centralized variant,
denoted as C-PAC, where the coordinator is fixed as the leader.
%
%
%
Note that we do not compare {\flac} with protocols designed for systems which
replicate data across different nodes like TAPIR~\cite{zhang2018building},
G-PAC~\cite{maiyya2019unifying}
since their settings are different from {\flac}'s
 system model.

In the following, we will present our detailed experimental results regarding to
three main aspects: (1) performance comparison of {\flac} versus other protocols: 2PC, 3PC, EC, C-PAC, (2) impact of {\flac}'s
parameters (network buffer parameter $r$, numbers of consecutive failure-free
runs $\numRunCF$, $\numRunNF$) on its performance, and (3) effectiveness of
reinforcement learning in enhancing the performance of {\flac} for unstable
operating conditions.
Appendix~\ref{sec:experimentApp} also includes experiments about the effect of network delays,
a comparison between {\flac} and replicated commit protocols, {\flac}'s message delay side effects, and a sensitivity analysis showing {\flac}'s performance under varying cross-shared transaction percentages.

\subsection{Experiment Setup}
\label{setup}

We conducted experiments on 11 server nodes deployed in a data cluster with
Ubuntu 20.04 system.
Each server consists of 10x2 3.7GHz Intel Xeon CPU W-1290P processor, 128GB of
DRAM, and connects to the cluster by 1Gbps network with 0.2ms network delay.
For our experiments, we configure one node as the coordinator while the others
are participants.
Each node maintains a single server process that handles requests from others
using lightweight execution threads.
We simulate a cross-datacenter environment by adding 10ms latency
    to packages sent among these nodes.

\textbf{Benchmarks.} For micro benchmarking, we run the experiment on a
YCSB-like micro-benchmark proposed in ~\cite{cooper2010benchmarking,
  maiyya2019unifying, qadah2021highly}.
Clients continuously generate read and write multi-record transactions with
closed loops, whose data access distribution fits Zipfian
distribution~\cite{cooper2010benchmarking} with a skew factor that controls
contention.
We make all transactions in this benchmark cross-shard by default to evaluate
the performance of all protocols in handling cross-shard transactions.
%
%
%
%
%
Tail latency (99\% latency or P99 latency) is used to measure transaction
latency since it can reflect precisely user experiences.
This latency is also widely used by existing systems \cite{chen2021achieving,
  decandia2007dynamo, suresh2015c3, prekas2017zygos}.
%
%
%

For macro benchmarking, we use TPC-C, a standard benchmark for OLTP
systems that include three types of read-write and two types of read-only
transactions~\cite{council2010tpc}.
We utilize three warehouses located separately in three server nodes.
Each warehouse contains 10 different districts, and each district maintains the
information of 3000 customers.
We use closed loops to simulate clients that keep sending transactions to the
coordinator and control the contention by adjusting the client thread number.
For this TPC-C benchmark, we only use ACPs for cross-shard transactions and
commit all single-shard transactions in a single message round trip as
in~\cite{gupta2018easycommit, lu2021epoch}.

\textbf{Parameters.}
%
%
We conduct the experiments in two operating environments: (i) the failure-free
environment, and (ii) the failure environment where crash and network failures
frequently occurs.
A parameter $\tau$ is used to control how failures are generated: in every cycle
of $2 \tau$ seconds, the system exhibits $\tau$ seconds of crash or network
failures, and then resumes to normal for the remaining $\tau$ seconds.
%
%
%
By default, we set the number of clients to 512 (YCSB-like) and 3000 (TPC-C), skew factor to 0.5, and network
buffer parameter $r$ to 1.
We also follow the existing works \cite{maiyya2019unifying, 1PCD} to set the
number of participant nodes to 3 when assessing the protocols' detailed
performance.
An experiment with varying number of participants, ranging from 3 to 10, is also
conducted to evaluate their scalability.

\subsection{Performance in Failure-Free Environment}
\label{failure-free}

In this environment, no failure will occur.
We start to record the experimental data 5 seconds after all the machines have
warmed up.
All tests are run for 10 times and we record their average results.

\begin{figure}[ht]
  \centering
    \subfigure[Throughput (skew factor = 0.5).]{%
    \includegraphics[width=0.48\linewidth]{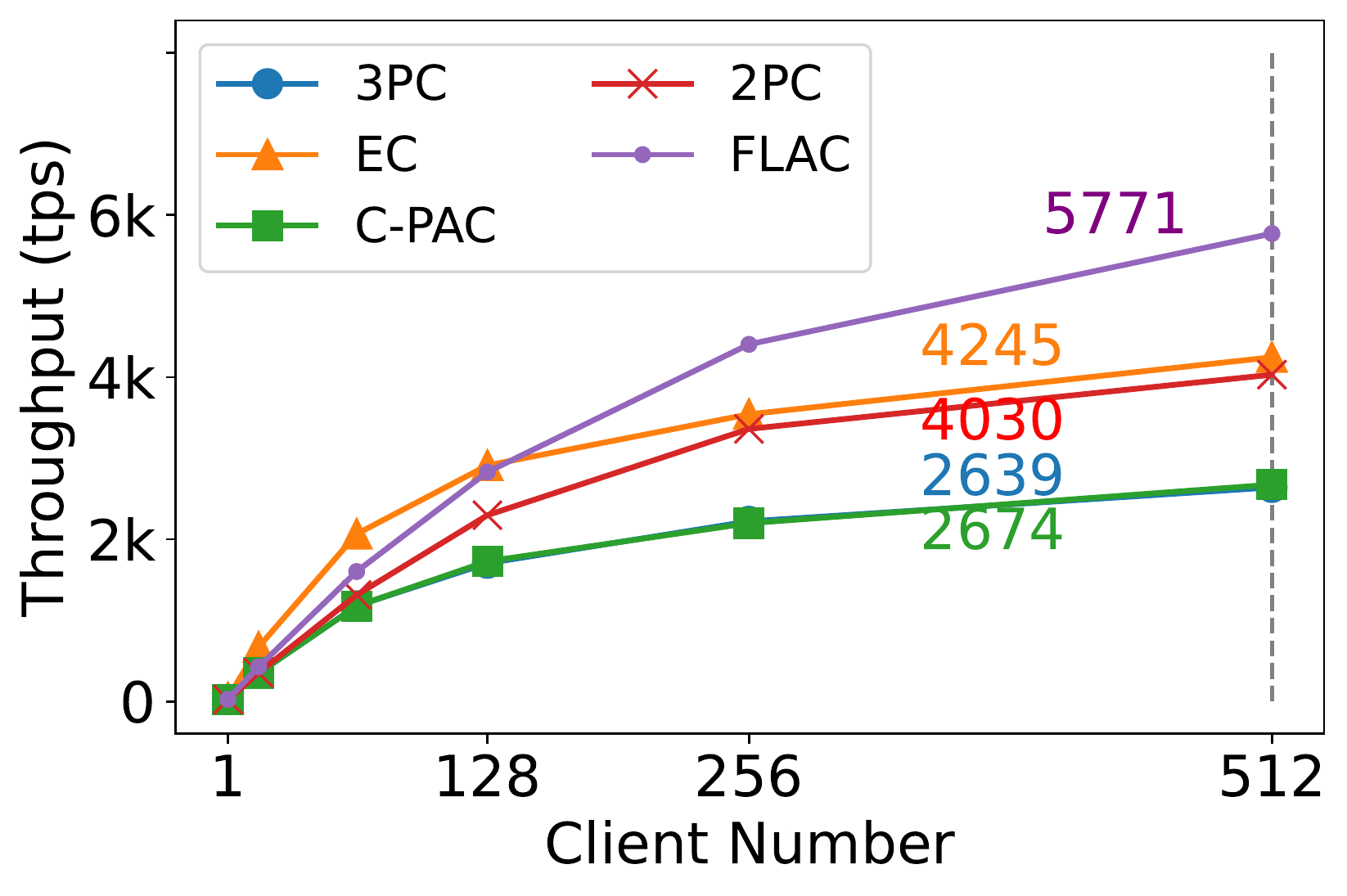}
    \label{fig:ycsb-th}}
  \subfigure[Tail latency (skew factor = 0.5).]{%
    \includegraphics[width=0.48\linewidth]{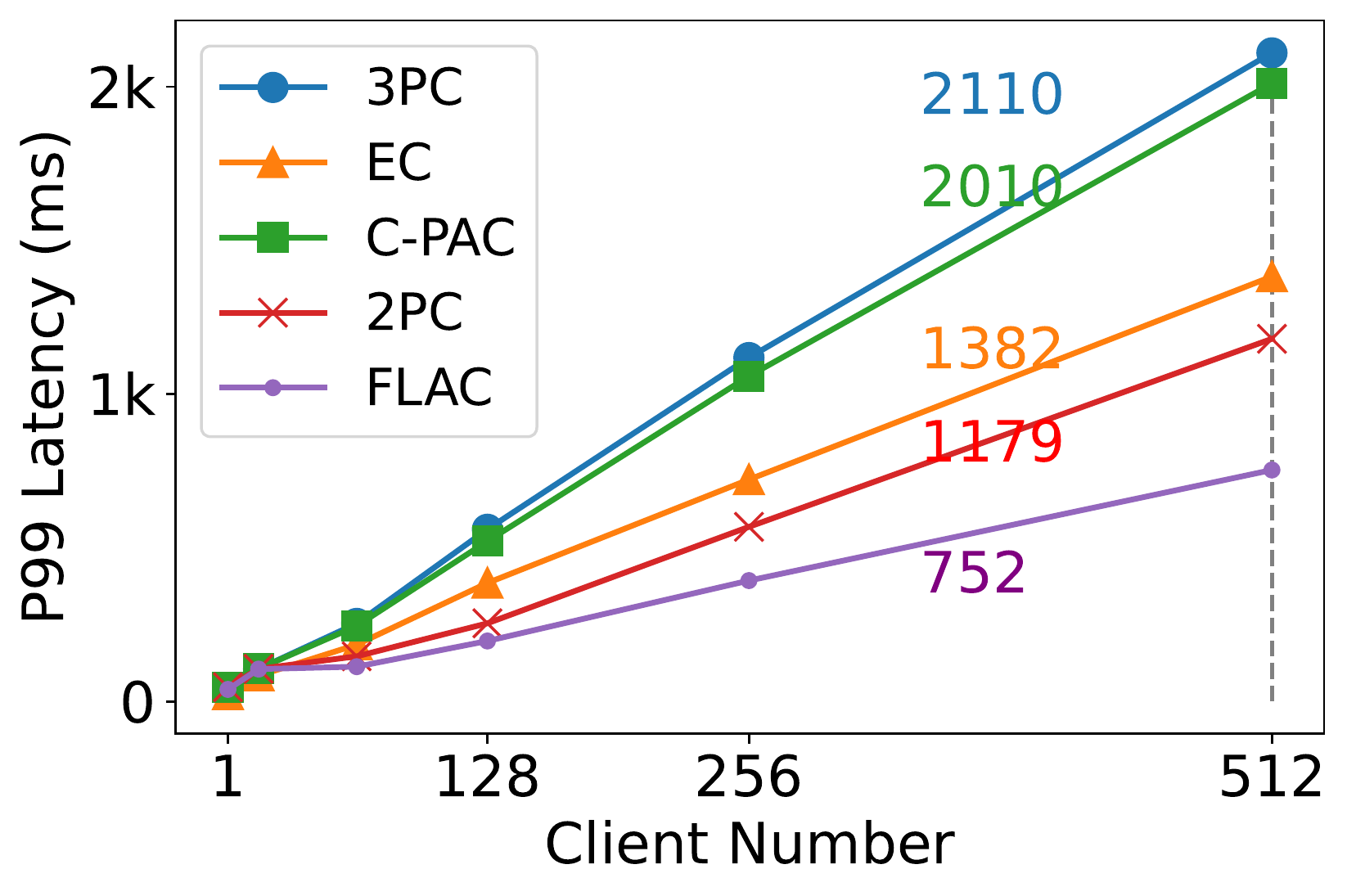}
    \label{fig:ycsb-la}}
   \subfigure[Throughput (512 clients).]{%
    \includegraphics[width=0.48\linewidth]{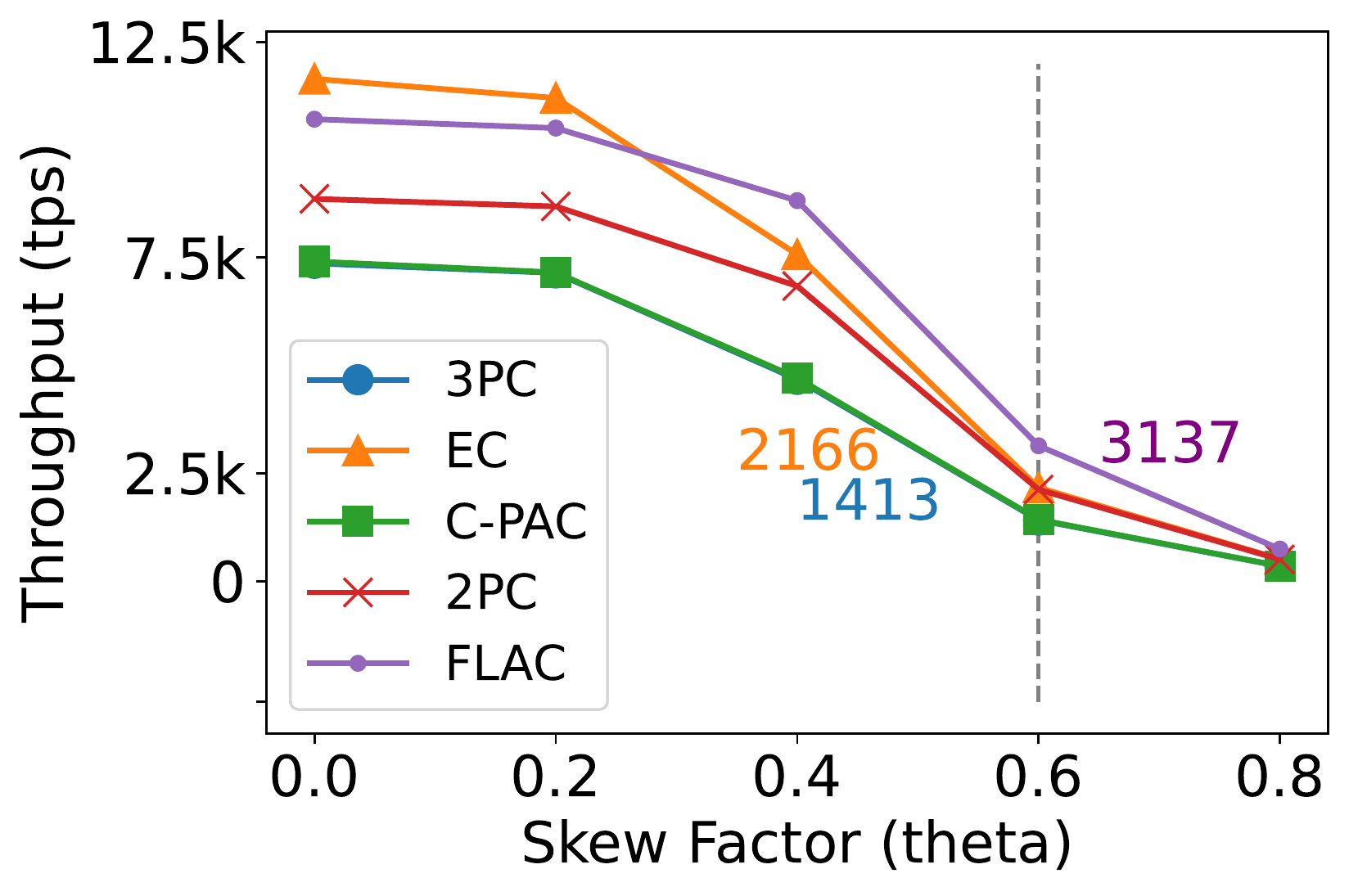}
  \label{fig:skew-th}}
  \subfigure[Tail latency (512 clients).]{%
    \includegraphics[width=0.48\linewidth]{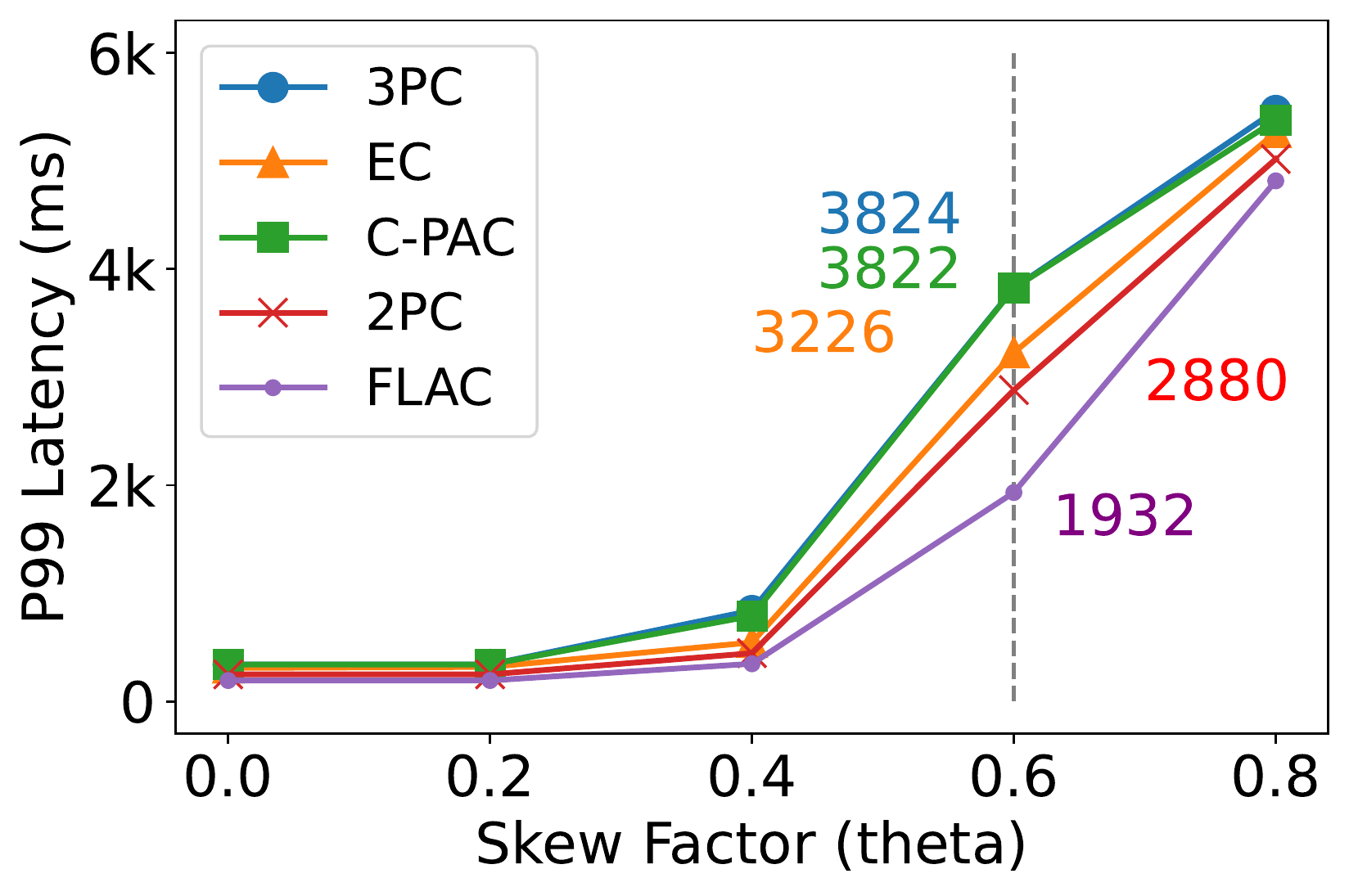}
  \label{fig:skew-la}}
  \vspace{-1em}
  \caption{Performance of all protocols under the YCSB-like micro-benchmark
    with varying client numbers, skew factors.}
  \label{fig:ycsb}
\end{figure}

\subsubsection{Experiment with the YCSB-like micro-benchmark, under varying
  numbers of clients and different skew factors.}
\label{sec:YCSBVaryContention}
Figure \ref{fig:ycsb} presents our experimental results, in which, {\flac}
demonstrates a significant gain in throughput, with a substantial reduction in
latency in all these settings.
For example, when there are 512 clients (Figures \ref{fig:ycsb}a and
\ref{fig:ycsb}b), {\flac} achieves 1.36x, 1.43x, 2.19x, and 2.16x throughput
speedup compared to EC, 2PC, 3PC, and C-PAC, while {\flac}'s latency is about
54.4\%, 63.8\%, 35.6\%, and 37.4\% of the others respectively.
When the skew factor is 0.6 (Figures \ref{fig:ycsb}c and \ref{fig:ycsb}d), it
achieves from 1.45x to 2.22x throughput speedup and from 32.9\% to 49.5\%
decrease in latency compared to other protocols EC, 2PC, 3PC, and C-PAC.
Again, these better performances of {\flac} are contributed by the protocol
{\flacFF} in this failure-free environment.
More specifically, {\flacFF} requires only 1 message delay on each participant
to commit or abort transactions, while all the protocols EC, 2PC, 3PC, and C-PAC
require at least 2 message delays (Table~\ref{tab:cp}).
%
%

We also observe in Figure \ref{fig:ycsb} that {\flac} performs slightly worse
than EC in low contention settings with small number of clients
(1{\textasciitilde}128) and low value of skew factor (0.0{\textasciitilde}0.3).
In the low contention setting, the protocols' performance is dependent mainly on
the latency of the coordinator side.
EC can perform better in such setting since it requires only 2 message delays on
the coordinator side, which is less than 3 of {\flacFF} (Table~\ref{tab:cp}).

\begin{figure}[ht]
  \centering
  \subfigure[16 clients.]{%
    \includegraphics[width=0.48\linewidth]{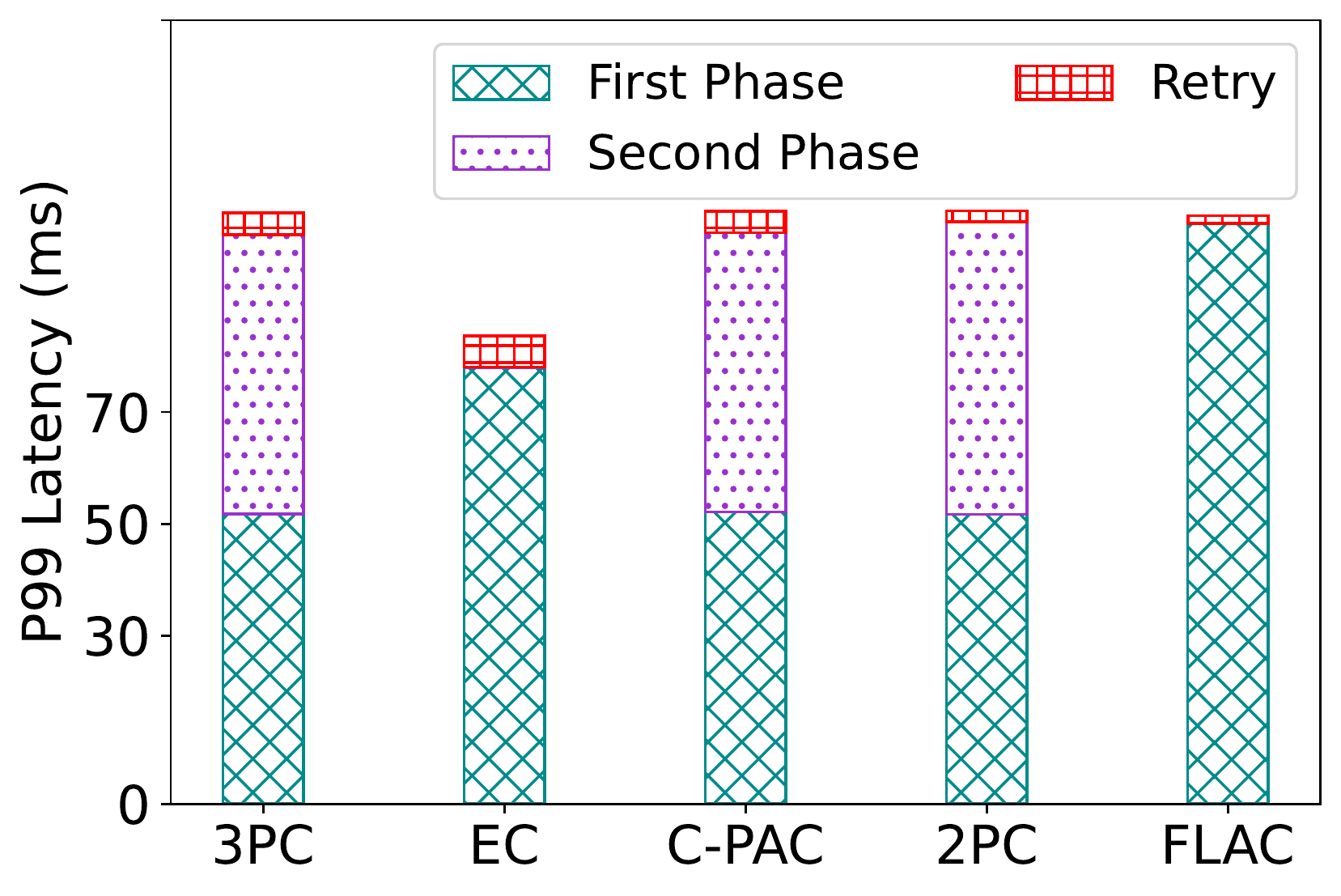}
    \label{fig:breakdown-low}}
  \subfigure[512 clients.]{%
    \includegraphics[width=0.48\linewidth]{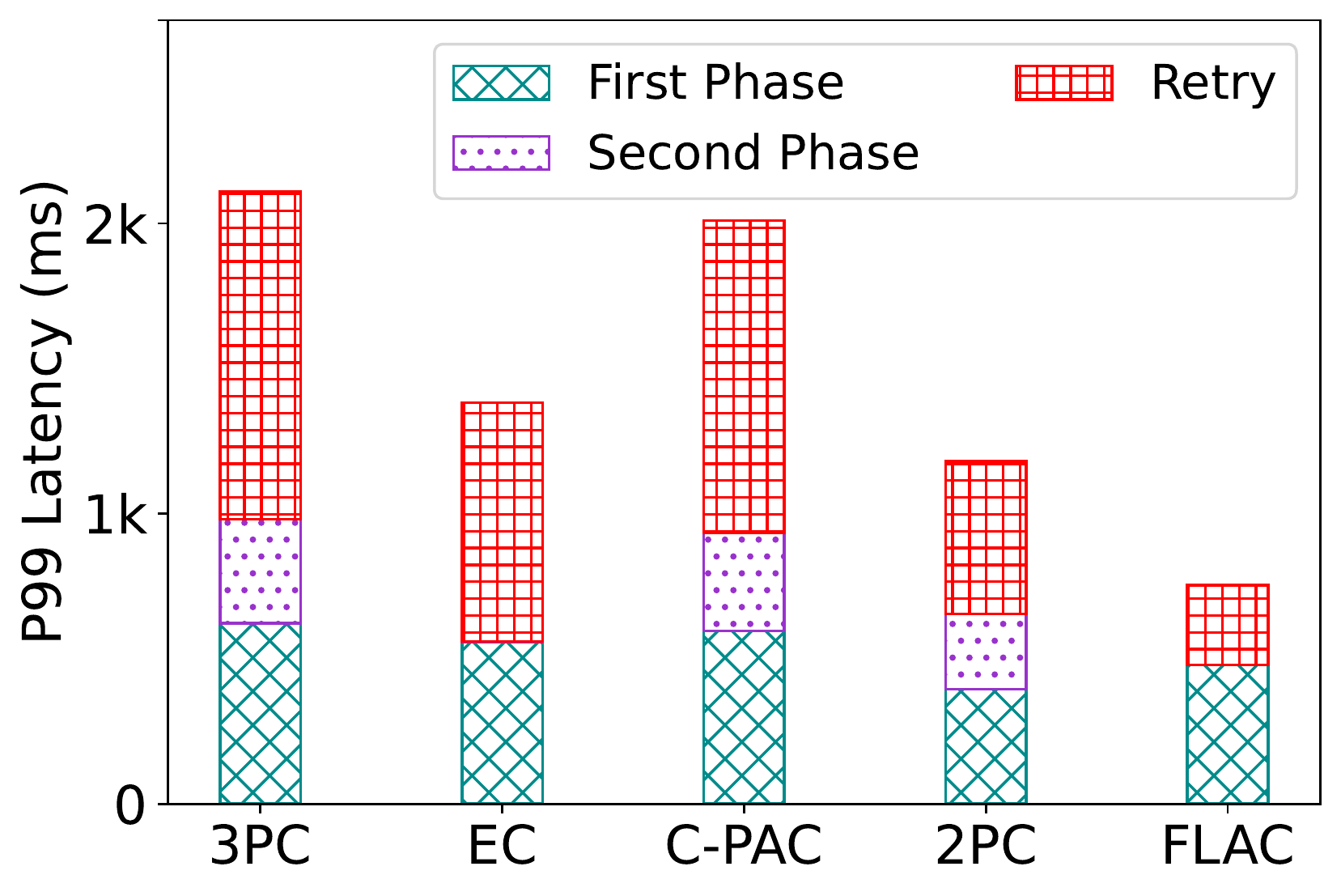}
    \label{fig:breakdown-high}}
  \vspace{-1em}
  \caption{Latency breakdown under the YCSB-like micro-benchmark when
    there are 16 and 512 clients.
  }
  \label{fig:breakdown}
\end{figure}

The above observation is also consistent with our experimental results on
latency breakdown of all protocols in Figure~\ref{fig:breakdown}.
In low contention, EC performs the best for its fewer coordinator side delays
(Figure~\ref{fig:breakdown-low}), while {\flac} outperforms others for the
reduce in retry cost thanks to its less blocking between transactions
(Figure~\ref{fig:breakdown-high}).

\subsubsection{Scalability evaluation of all protocols with the YCSB-like
  micro-benchmark, under a varying number of participants.}
As mentioned previously, we follow existing works \cite{maiyya2019unifying,
  1PCD} to set the number of participants to 3.
However, in practice, users can choose to use a different number of participants
depending on applications.
Therefore, we conduct this experiment to study how {\flac} and other protocols
can scale with a larger number of participants, ranging from 3 to 10.
According to Figure~\ref{fig:scale}, when the number of participants increases,
the throughput of all the protocols reduces, while their latency increases.
Such decay in performance is due to the long occupation of transaction
resources across multiple nodes, when the number of nodes gets larger.
Nevertheless, {\flac} still performs better than others in all settings.
For example, in the settings of 3 and 10 participant nodes, it respectively
obtains from 1.36x to 2.19x and from 1.20x to 1.62x throughput speedup
  more than other protocols.
This performance gain is also due to the similar reason we explained before: the
protocol {\flacFF}, which operates in this failure-free environment,
requires less number of participant-side message delays than other protocols
(Table~\ref{tab:cp}).

\begin{figure}[ht]
  \centering
    \subfigure[Throughput.]{\includegraphics[width=0.48\linewidth]{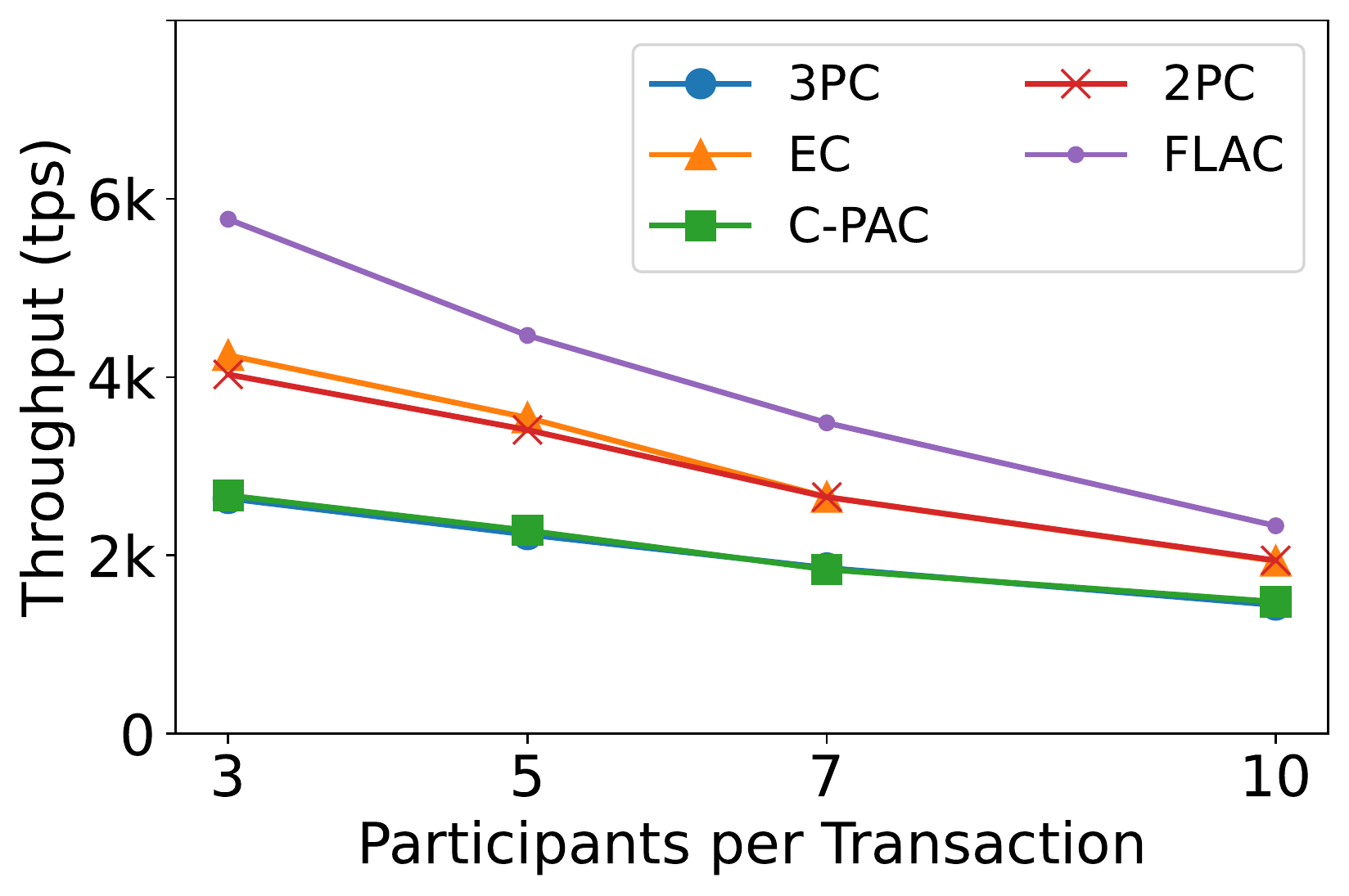}
    \label{fig:scale-th}}
  \subfigure[Tail latency.]{%
    \includegraphics[width=0.48\linewidth]{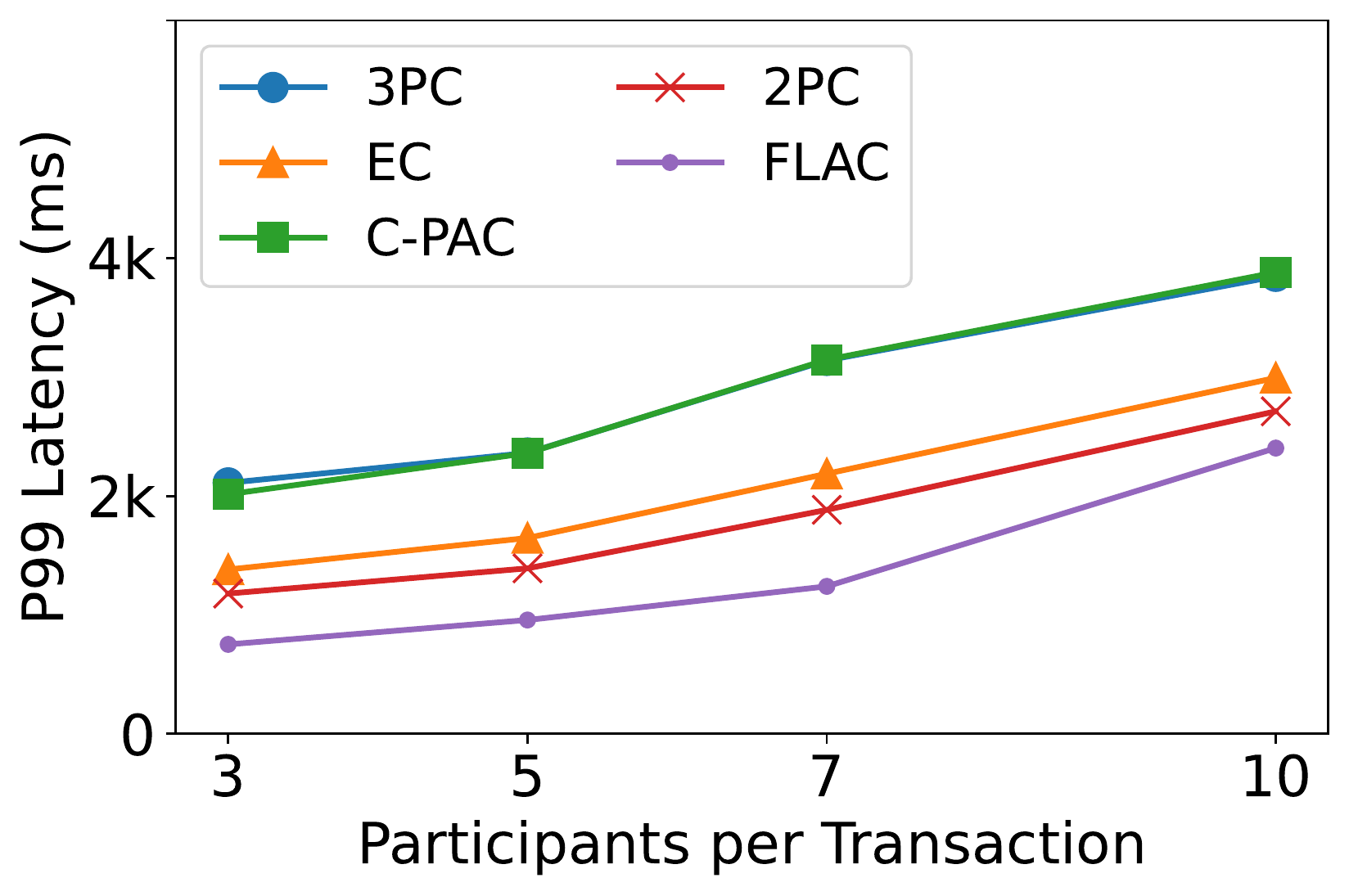}
    \label{fig:scale-la}}
  \vspace{-1em}
  \caption{Performance under the YCSB-like micro-benchmark with
    varying numbers of participants.}
  \label{fig:scale}
\end{figure}

\subsubsection{Impact of the Network Buffer Parameter $r$}
\label{effect_r}

As described in Section ~\ref{LP1}, this parameter is used to adjust the message
delay upper bound for different network conditions.
To study its impact on protocol performance, we evaluate {\flac} with the TPC-C
benchmark under different contentions when the number of clients varies.
The results are presented in Figure~\ref{fig:r}.
Firstly, when $r$ changes from 0.5 to 1, the throughput increases from 13.0K tps
to 21.7K tps and the latency decreases from 684ms to 321ms.
With smaller $r$, the system assumes a shorter network timeout to collect
messages.
Consequently, more messages fail to arrive within the network timeout, and
{\flac} regards them as network failures, causing it to incline to the more
stringent network-failure level.
Secondly, when $r$ changes from 0.1 to 0.5 and from 1 to 8, the throughput
decreases from 14.1K to 13.0K and from 21.7K to 12.2K, while the latency
increases from 556ms to 684ms and from 321ms to 683ms, respectively.
This degrading of {\flac}'s performance while $r$ increases is because the
coordinator waits for a longer time to get all participants' results before
sending its decision back to them.

\begin{figure}[ht]
  \centering
  \subfigure[Throughput.]{%
    \includegraphics[width=0.476\linewidth]{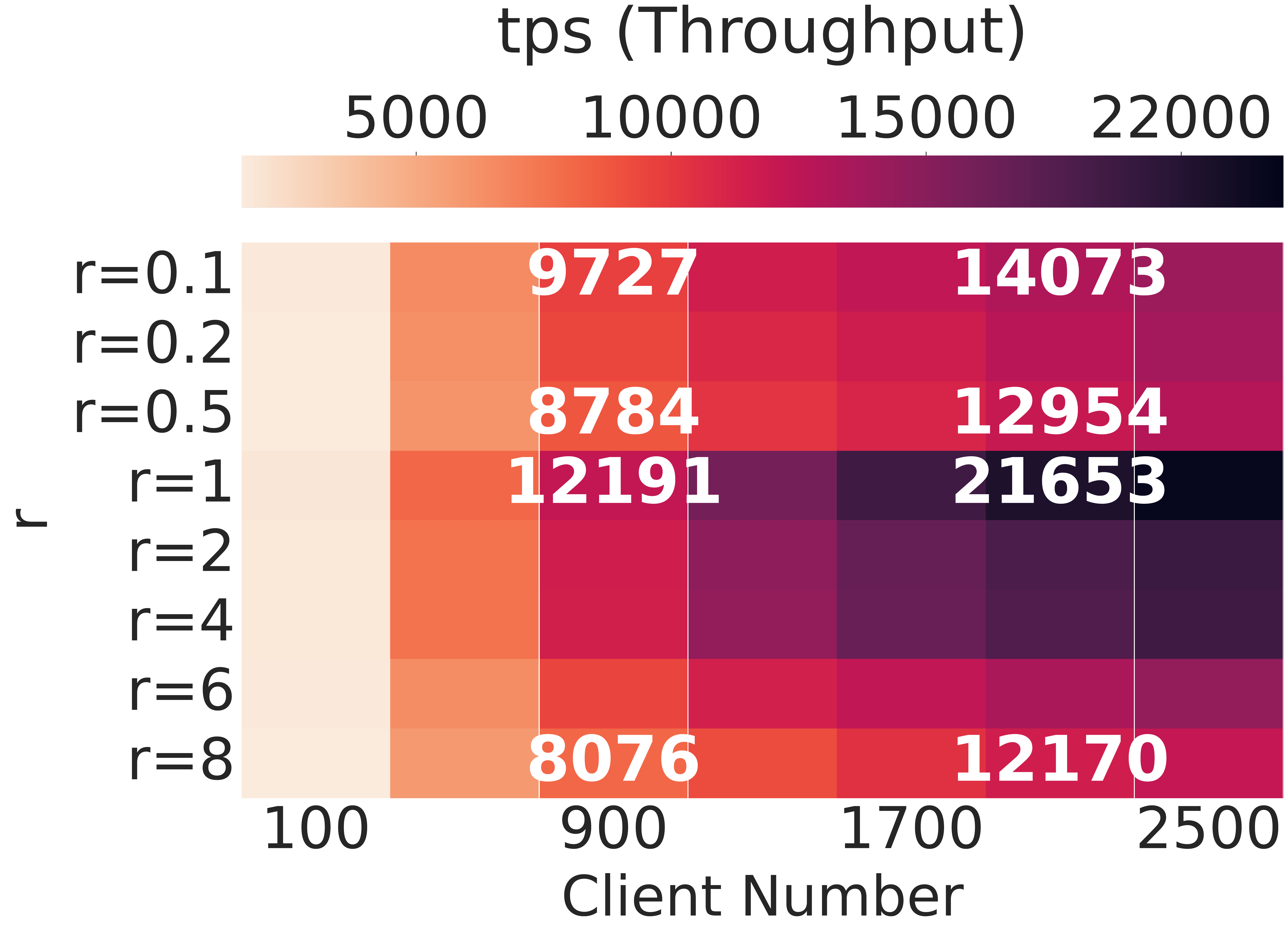}
    \label{fig:rs-th}}
  \subfigure[Tail latency.]{%
    \includegraphics[width=0.476\linewidth]{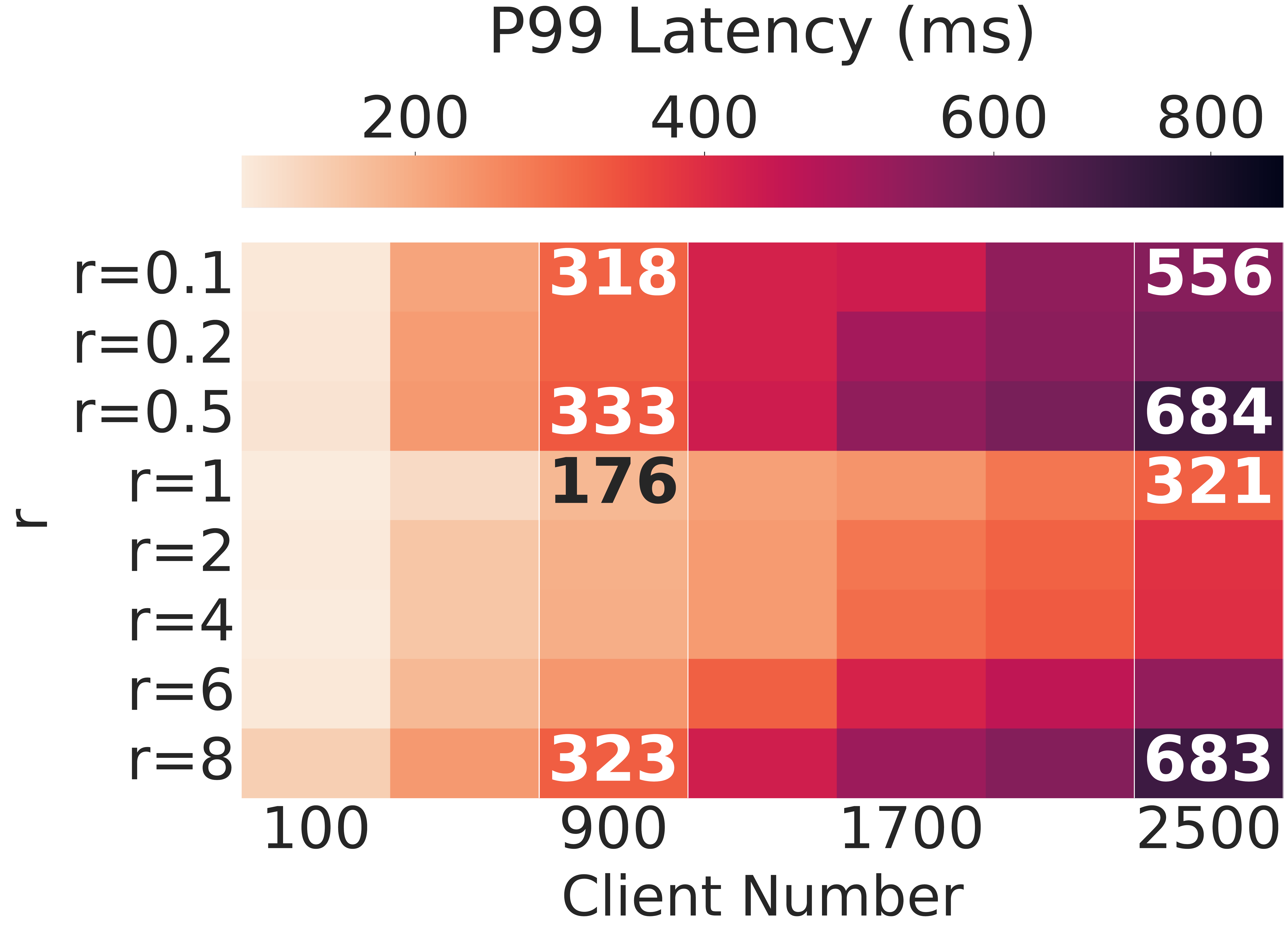}
    \label{fig:rs-la}}
  \vspace{-1em}
  \caption{Effects of network buffer
    parameter $r$ on {\flac}. Darker color indicates higher latency or
    throughput.}
    \label{fig:r}
\end{figure}

\subsection{Performance on Failure Environment}

Failure environment is created by frequently injecting crashes or delays into
nodes and network connections~\footnote{We pick one participant and let it stop or delay message handling.}.
This injection is controlled by the parameter $\tau$ (Section \ref{setup}): a
failure is triggered and last in $\tau$ seconds for a periodic cycle of $2\tau$
seconds.
We record experimental data after the machine warm up time in $5 + 2\tau$
seconds\footnote{$5$ seconds is the time needed to warm up all the machines, and
  $2 \tau$ seconds is one buffering time cycle needed for the performance of all
  the protocols to become stable.}.

\subsubsection{Experiment with the TPC-C benchmark and environments containing
  only crash failures or only network failures.}

\begin{figure}[ht]
  \centering
  \subfigure[Throughput.]{%
    \includegraphics[width=0.48\linewidth]{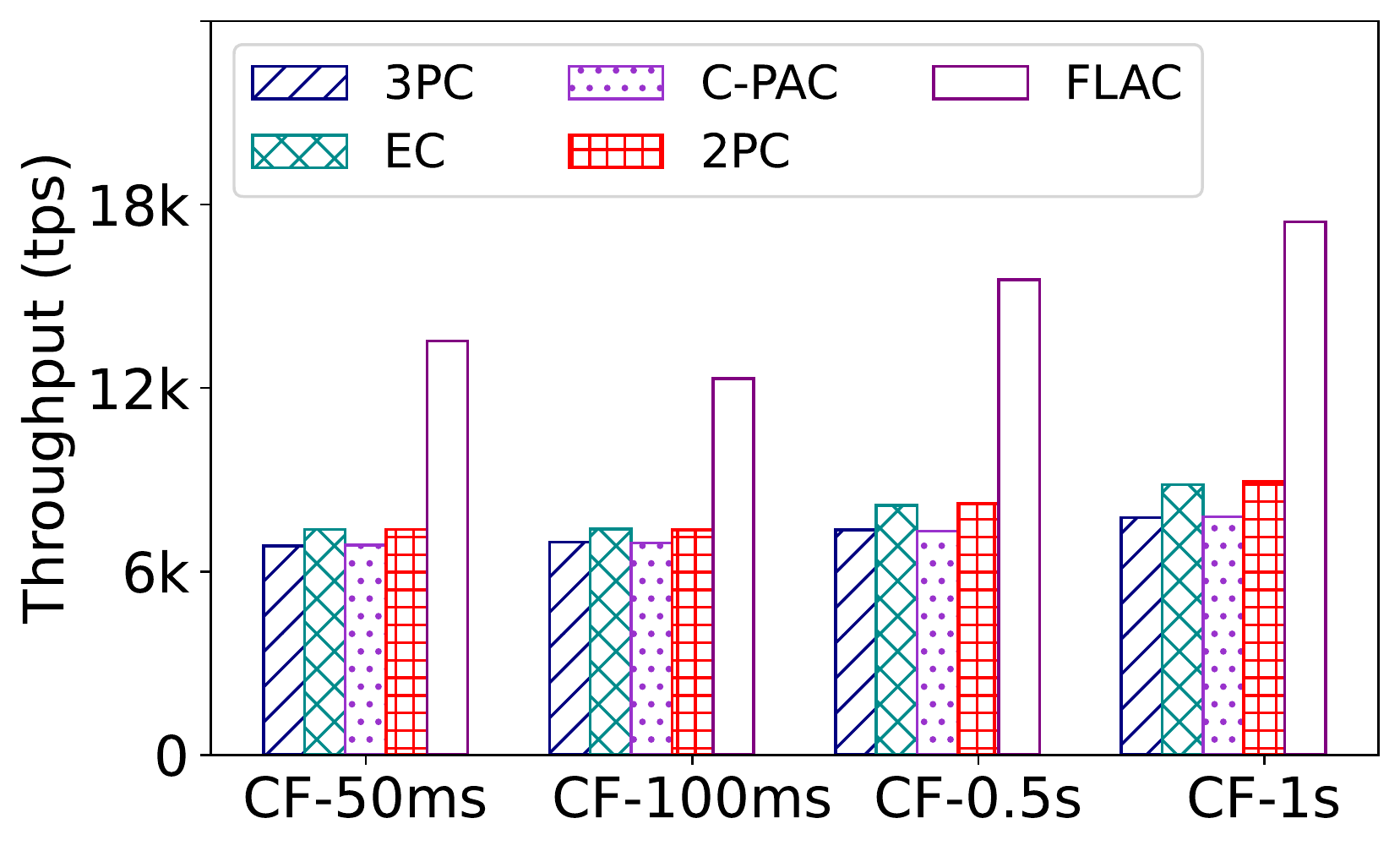}
    \label{fig:unstable-cf-th}}
  \subfigure[Tail Latency.]{%
    \includegraphics[width=0.48\linewidth]{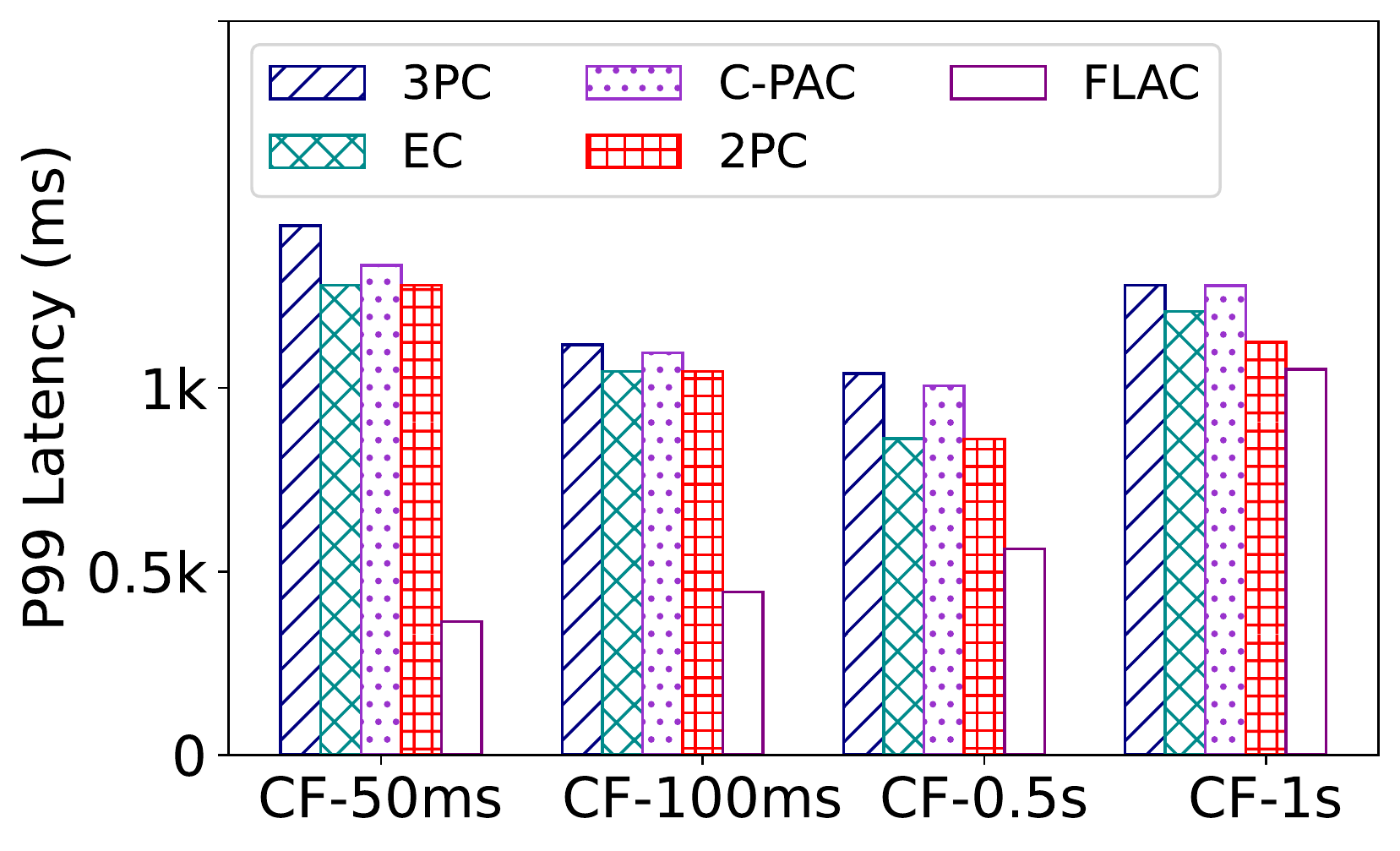}
    \label{fig:unstable-cf-la}}
      \subfigure[Throughput.]{%
    \includegraphics[width=0.48\linewidth]{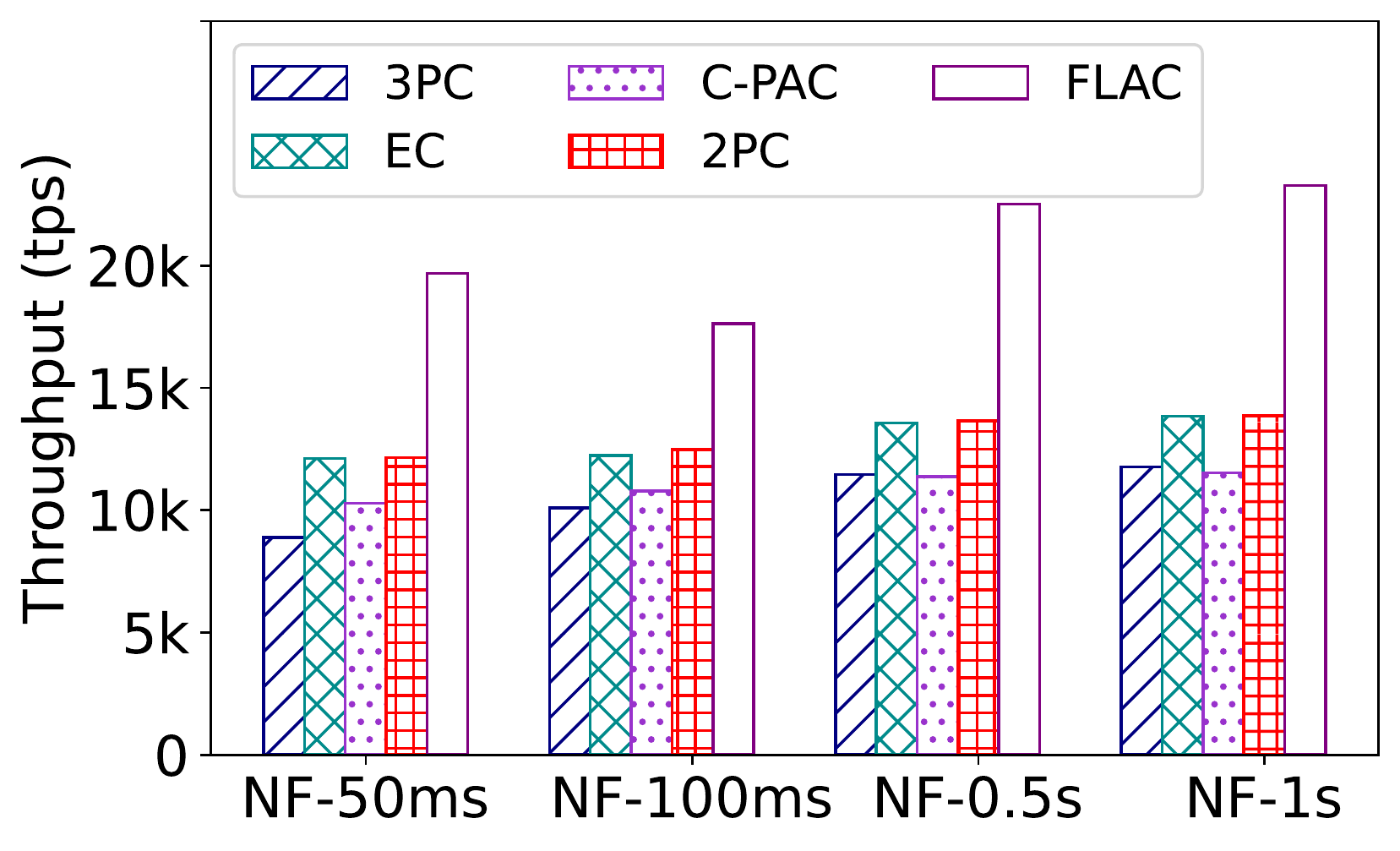}
    \label{fig:unstable-nf-th}}
  \subfigure[Tail Latency.]{%
    \includegraphics[width=0.48\linewidth]{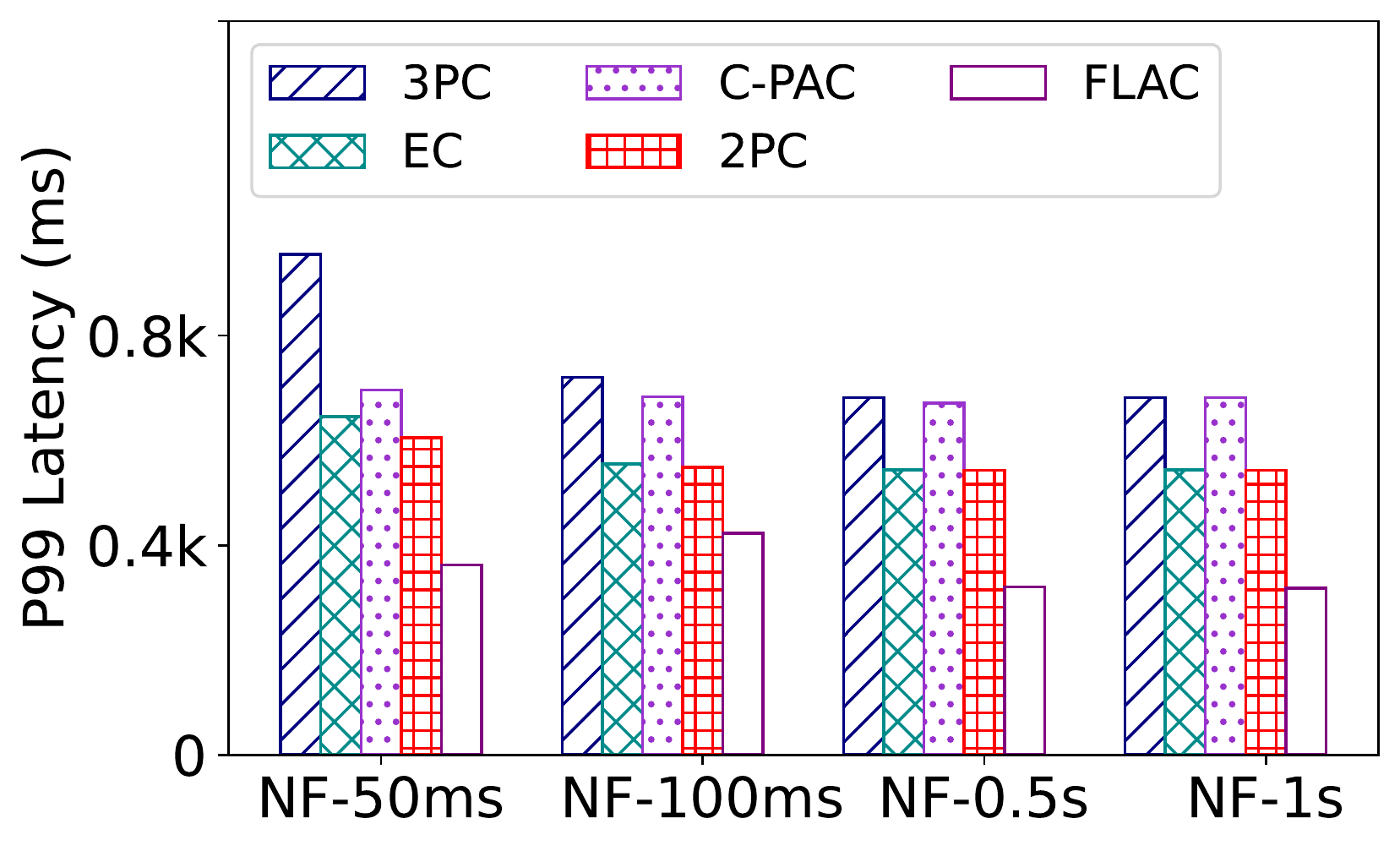}
    \label{fig:unstable-nf-la}}
  \vspace{-1em}
  \caption{Protocol performance under TPC-C benchmark with environments having
    only crash (a, b) or only network failures (c, d).}
  \label{fig:unstable}
\end{figure}

%
Figure~\ref{fig:unstable} presents the experimental results in different
environments where crash failures or network failures happen on a participant
periodically.
Here, the notation {CF{-}50ms} (or {NF{-}50ms}) means that the environment will
repeatedly exhibit a crash failure (or a network failure) for 50ms and then
resume to normal for another 50ms ($\tau = 50ms$).
In all failure settings ($\tau = 50ms, 100ms, 0.5s, 1s$), {\flac} always achieves
better performance than other protocols.
Figures \ref{fig:unstable-cf-th} and \ref{fig:unstable-cf-la} show that with
periodic crash failures, {\flac} gains from 1.67x to 2.25x throughput speedup,
while its latency drops from 25.2\% to 93.3\%, compared to the next two best
protocols EC and 2PC.
Such improvement is derived from {\flac}'s capability to select and execute the
most suitable protocol between {\flacFF} and {\flacCF}.
%
%
When network failures happen (Figures \ref{fig:unstable-nf-th} and
\ref{fig:unstable-nf-la}) {\flac} can still achieve better throughput speedup
(from 1.41x to 2.22x), and lower latency (from 38.0\% to 77.1\%) compared to EC
and 2PC.
This performance improvement, again, is due to {\flac}'s adaptivity to switch
between its protocols ({\flacFF} and {\flacNF}).
%

\subsubsection{Experiment with the TPC-C benchmark and environments containing
  both crash and network failures.}
\label{sec:UnstableEnv1}

\begin{figure}[ht]
  \centering
  \subfigure[Throughput.]{%
    \includegraphics[width=0.48\linewidth]{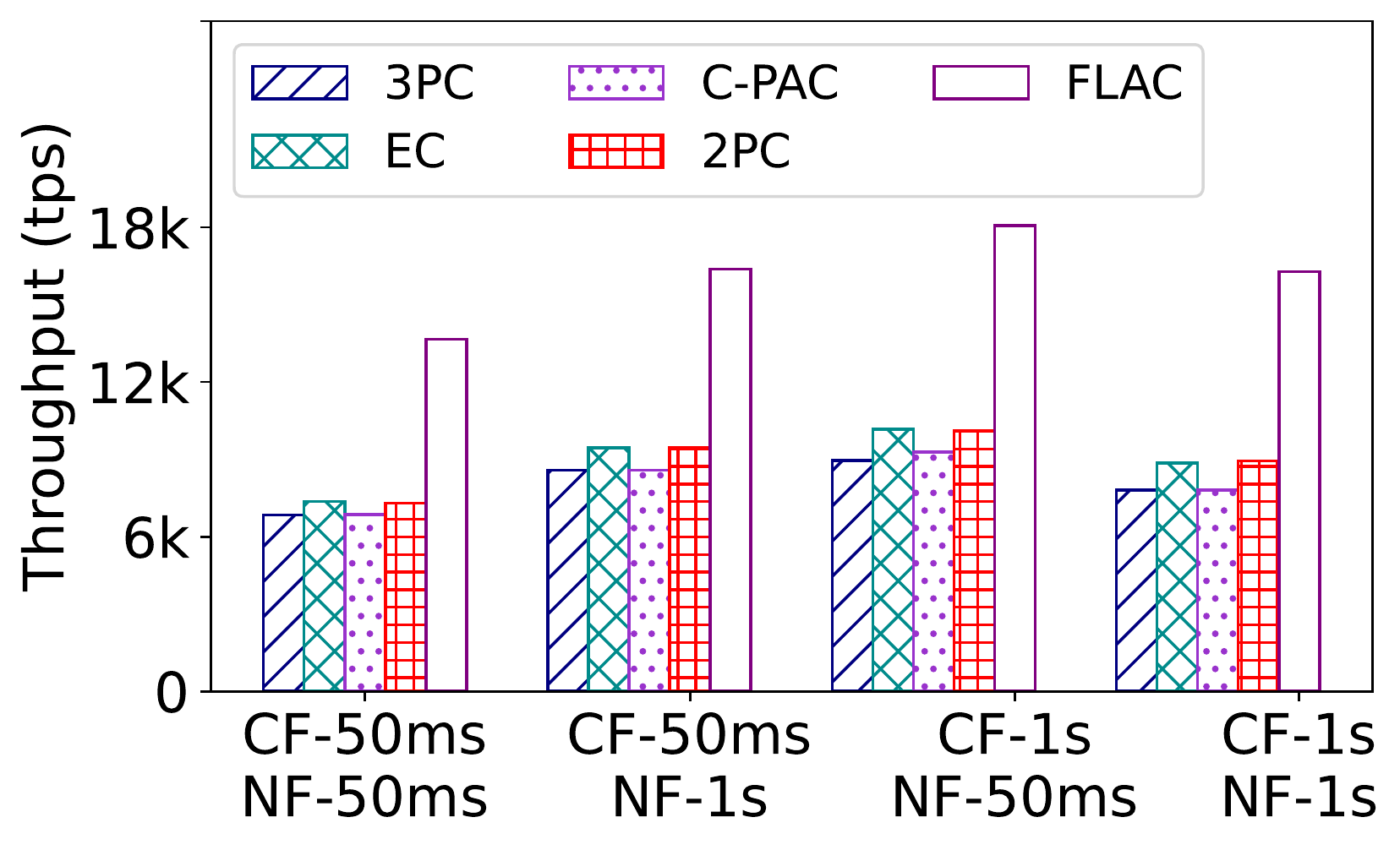}
    \label{fig:addition-th}}
  \subfigure[Tail Latency.]{%
    \includegraphics[width=0.48\linewidth]{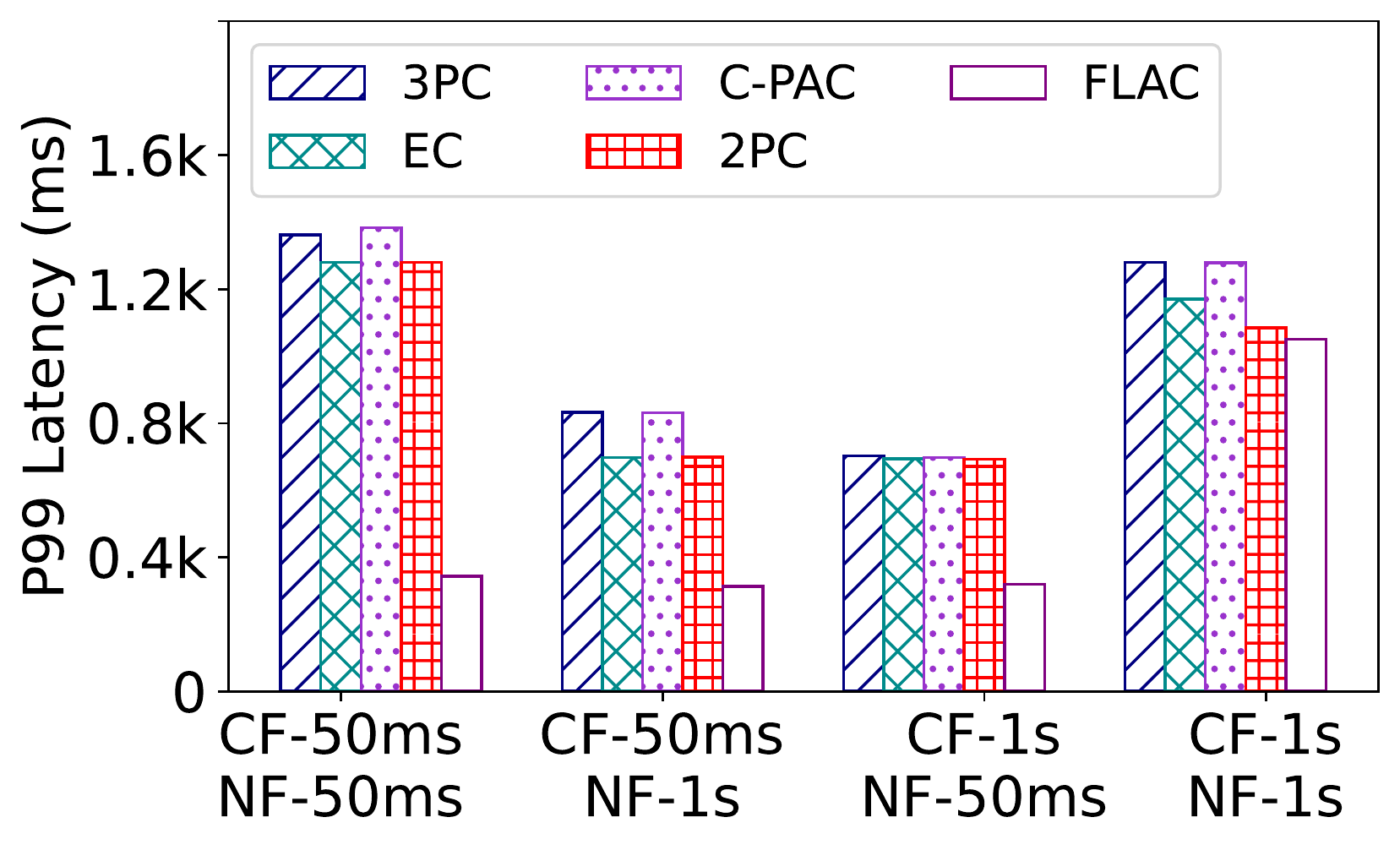}
    \label{fig:addition-la}}
  \vspace{-1em}
  \caption{Protocol performance under the TPC-C benchmark
    with environments containing both crash and network failures.}
  \label{fig:addition}
\end{figure}

%
Figure~\ref{fig:addition} presents the experimental results in different
environments where crash failures and network failures both happen.
For example, the combined setting CF{-}50ms NF{-}1s indicate that the
environment will repeatedly exhibit a crash failure for 50ms, resume to normal
for the next 50ms, and then continue to exhibit a network failure for 1s, and
return to normal for the next 1s.
In all experimented settings (CF{-}$\tau_1$ NF{-}$\tau_2$, where $\tau_1,
\tau_2$ can be $50ms$ or $1s$), {\flac} still performs the best among all
protocols.
In particular, it achieves from 1.73x to 2.09x more throughput speedup than EC
and 2PC (the next two best protocols), while its latency reduced from 24.8\% to
96.9\% of them.
This result is also due to {\flac}'s ability to switch to a suitable protocol
({\flacFF}, {\flacCF}, or {\flacNF}) when the operating environment changes.

\subsubsection{Impact of reinforcement learning (RL) on fine-tuning
  robustness-downgrade parameters ($\numRunCF$, $\numRunNF$) to {\flac}'s
  performance}
\label{sec:TurningRLParams}
Figure~\ref{fig:unstable-rl} reports the performance comparison of {\flac} in
failure environments when the parameters $\numRunCF, \numRunNF$ are fine-tuned
by RL compared to when they are fixed to 4 settings:
(1) $\numRunCF{=}1$, $\numRunNF{=}1$, which adapts {\flac} to environments where
failures do not recur quickly (stable environments),
(2) $\numRunCF{=}1, \numRunNF{=}256$, where network failures recur quickly,
(3) $\numRunCF{=}256$, $ \numRunNF{=}1$, where crash failures recur quickly,
(4) $\numRunCF{=}256$, $\numRunNF{=}256$, where all failures recur quickly.

The result shows that {\flac} obtains better or equivalent performances when
$\numRunCF, \numRunNF$ are fined-tuned by RL, compared to when they are fixed to
certain values in all environments.
In particular, {\flac} with RL-tuned parameters outperforms other fix-valued
settings in the environments where network failures recur fast (Figure
\ref{fig:unstable-rl-th}, NF-50ms) or where crashes recur fast but network
failures recur slowly (Figure \ref{fig:unstable-rl-mixed}, CF-50ms NF-1s).

We also notice that a fixed-value setting can enable {\flac} to perform well in
some environments, such as the settings (i) $\numRunCF{=}256$, $\numRunNF{=}1$ for
environments where crash failures recur fast (Figure \ref{fig:unstable-rl-th},
CF-50ms), (ii) $\numRunCF{=}1$, $\numRunNF{=}1$ for stable environments (Figure
\ref{fig:unstable-rl-th}, CF-1s, Figure \ref{fig:unstable-rl-mixed}, CF-1s NF-1s),
and (iii) $\numRunCF{=}256$, $\numRunNF{=}256$ for environments
where network failures recur fast (Figure \ref{fig:unstable-rl-th},
NF-50s).
%
However, there is no fixed-value setting that can make {\flac} perform well in
all experimented failure environments.
In contrast, the setting that fine-tunes $\numRunCF$, $\numRunNF$ by RL
can boost {\flac} to gain better performance in all these environments,
as described earlier.
Furthermore, {\flac} only spent up to 5 seconds for such parameter tuning.
This experiment confirms our proposal that RL can help to fine-tune the
robustness-downgrade transition parameters so that {\flac} can perform well in different failure environments.

 \begin{figure}[ht]
   \centering
   \subfigure[Only crash or only network failures.]{%
     \includegraphics[width=0.48\linewidth]{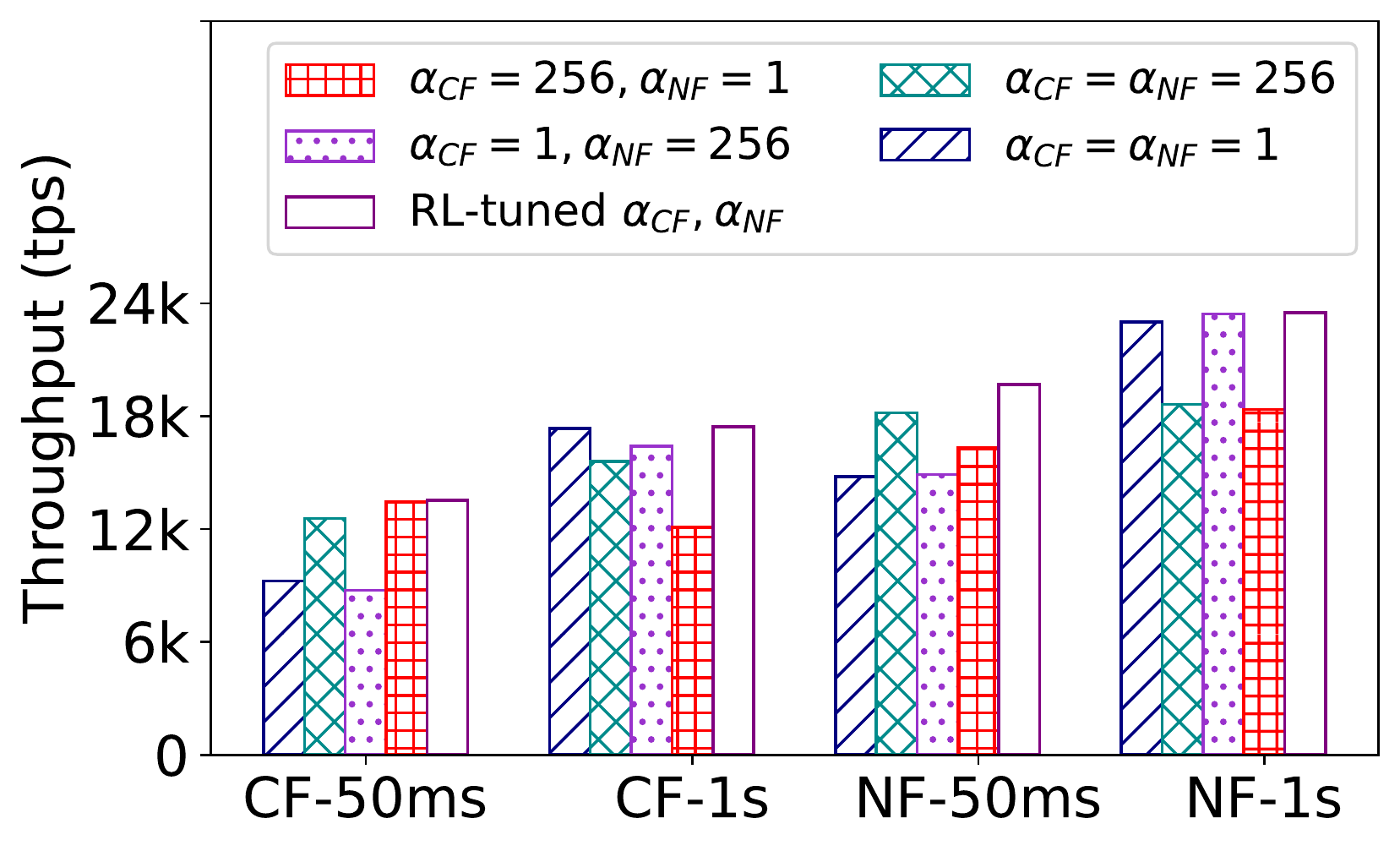}
     \label{fig:unstable-rl-th}}
   \subfigure[Both crash and network failures.]{%
     \includegraphics[width=0.48\linewidth]{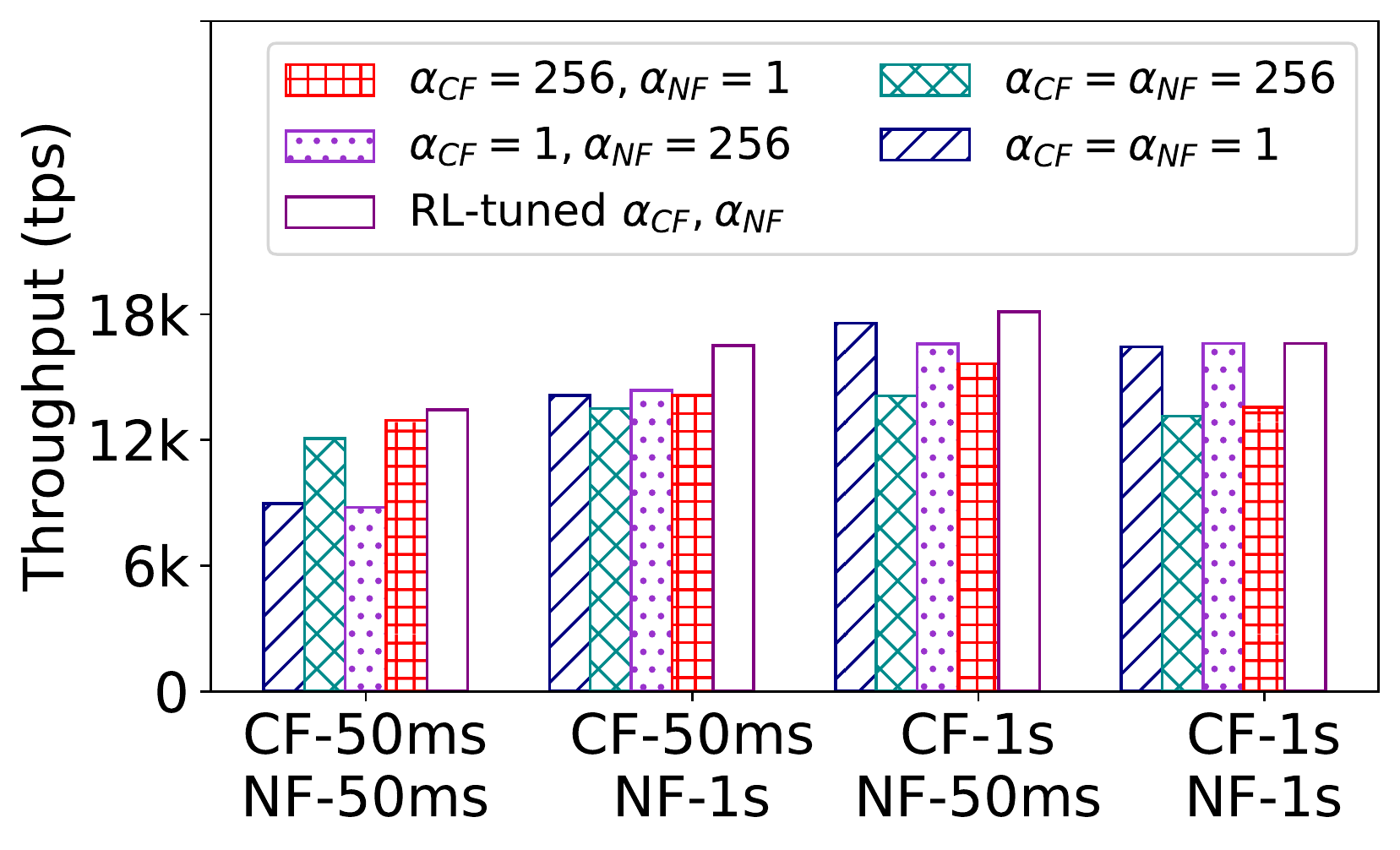}
     \label{fig:unstable-rl-mixed}}
   \vspace{-1em}
   \caption{Comparison of {\flac}'s performance under TPC-C benchmark
     with fixed-value and RL-tuned parameters {\numRunCF, \numRunNF}.}
   \label{fig:unstable-rl}
 \end{figure}


\vspace{-0.5em}
\section{Related works}
\label{related}
\vspace{-0.1em}

Atomic Commit Protocol (ACP) for distributed transactions is an important topic in
database systems.
Initially proposed by Skeen in~\cite{Skeen}, ACP sees a series of improvements
in ~\cite{HFC, NBAF, WFD}.
Among them, 2PC~\cite{2PC} is widely recognized as the golden standard in the
industry but is also well-known to be the bottleneck of distributed systems for
its low efficiency.
Existing ACPs optimize 2PC by (i) enhancing the performance during normal
processing, or (ii) improving availability by ensuring liveness property, or
(iii) making the protocol adaptive for better overall performance.

To enhance the performance during normal processing, PrA (Presumed Abort)
and PrC (Presumed Commit) protocols were introduced in
\cite{mohan1986transactionprcpra}, which use additional logs to reduce the cost
of aborting or committing messages.
PrC was further optimized in \cite{lampson1993new} by giving up knowledge about
the transactions prior to coordinator failures to enhance its performance.
Nonetheless, these protocols, and their variants in \cite{al1997enhancing,
chrysanthis1998recovery}, cannot reduce the participant-side message delay.
To improve such participant-side message delay, 1PC \cite{1PCY} and other
variant protocols \cite{stamos1990low, stonebraker1979UV, stamos1993EP}
sacrifice the generalizability by adding additional strong assumptions, while
O2PC \cite{levy1991optimistic} compromises the atomicity by optimistically
assuming all global transactions will not fail to ensure serializability.
Compared with them, {\flac} can execute transactions more
efficiently without compromising the generalizability or the ACID properties.

Several works have also been proposed to improve the availability of ACPs by
ensuring the liveness property.
3PC and variant ACPs ~\cite{skeen1982quorum, 3PC, gupta2018easycommit} address
2PC's lack of termination when nodes fail by adding additional messages or
communication steps.
G-PAC and its variants \cite{hammer1980reliability, PAX, maiyya2019unifying,
  zhang2018building} ensure the liveness of 2PC under node failures by
replication.
Easy Commit (EC) reduces the communication steps of 3PC by the
transmission-before-decide technique \cite{gupta2018easycommit}.
Paxos Commit is introduced in \cite{PAX} to ensure liveness by replicating
coordinator logs.
Unlike them, {\flac} does not require replicas; it is equipped with an
adaptivity ability to handle failures and achieve high performance.
Moreover, while other ACPs assume a fixed running condition, {\flac} monitors
the condition during execution and dynamically switches to a suitable dedicated
protocol.

There are also other ACPs that can switch to operate a suitable protocol
when system behavior or requirements change.
For example, ACOM \cite{ACOM} can shift between 2PC-PrA and 2PC-PrC according to
node behavior.
ADTP \cite{ADTP} can let service providers choose between PrC and PrA.
ML1-2PC \cite{12PC, ML12PC} can dynamically select either 1PC or 2PC based on
transaction behavior and system requirements.
Besides, a consensus protocol Domino \cite{yan2020domino} can adapt executions
automatically between fast or slow paths with the network timeout.
%
%
Compared to these works, {\flac} focuses on a different setting: it optimizes
ACP for operating environments where crash or network failures can occur and
change over time.
Such setting is common in cloud systems~\cite{decandia2007dynamo}, and {\flac}
can achieve good performance without compromising generalizability.



\section{Limitations and Future Work}
\label{sec:short-discussion}

We now discuss the limitations of {\flac} and our future work.
Firstly, {\flac}'s higher message complexity can affect system performance when
each transaction involves many participants.
However, we have conducted additional experiments to study this limitation and
observed that {\flac} always performs better or at least comparably to other
protocols in all settings.
Details of these experiments are described in Appendix~\ref{sec:experimentApp}.

%

Secondly, {\flac} may compromise the liveness of transactions executed with
wrong operating condition assumptions.
The blocked transactions could exacerbate the contention between transactions
and reduce system throughput.
Such blocking has already been widely accepted in existing systems that use 2PC.
Meanwhile, we have proposed {\rlsm} to alleviate this problem, enabling {\flac}
to outperform all other protocols even in unstable environments.

Thirdly, we follow state-of-the-art ACPs \cite{2PC, 3PC, gupta2018easycommit,
  gupta2020efficient} to not consider system replication in the design of
{\flac}.
Nonetheless, {\flac} can be easily customized for replicated systems by three
steps: asynchronously replicating decisions for {\flacFF} and {\flacCF},
synchronizing decisions to replicas upon transition to the network-failure
level, and adopting replicated ACPs for {\flacNF}.
We detail these steps in Appendix~\ref{app:discussion}.
Note that {\flac} does not always ensure availability when failures
happen.
Unlike replicated ACPs ~\cite{PAX, maiyya2019unifying, zhang2018building},
{\flac} is not designed for systems that require extremely high availability.
Instead, it strikes a balance between system availability and efficiency.
We discuss this issue with more details in Appendix~\ref{app:discussion}.


\section{Conclusions}
\label{conclusion}

We have presented {\flac}, a robust failure-aware atomic commit protocol for
distributed transactions.
{\flac} includes a novel state transition model RLSM to infer the operating
conditions.
Reinforcement learning is also applied to boost its adaptivity and performance in
unstable environments.
We have implemented a distributed transactional key-value storage with \flac~
and evaluated its performance.
The results show that \flac~ achieves up to 2.22x improvement in throughput and
up to 2.82x speedup in latency, compared to existing protocols 2PC, 3PC, Easy
Commit, C-PAC.
We have also evaluated the effectiveness of reinforcement learning on tuning {\flac}'s
parameters.
The results indicate that the tuning can improve the performance of {\flac} in
unstable environments.


\balance
\bibliographystyle{ACM-Reference-Format}
\normalem
\bibliography{ref}

\appendix

\newpage

\appendix

\section{Failure Handling, Correctness Proofs, and Replication Support}

\subsection{Failure Handling of {\flac}}
\label{sec:failureHandleApp}
This section elaborates on how {\flac} handles crash failures during transaction processing.
We follow existing works~\cite{skeen1983formal, gupta2018easycommit, skeen1982quorum} to discuss {\flac}'s failure handling in two aspects:
how {\flac} safely recovers a node from crash failure by continuing half-executed transactions (crash recovery),
and how {\flac} terminates the transaction execution when it gets blocked by crash failure of some other nodes (protocol termination).

\subsubsection{Crash Recovery}
\label{recovery}

%
{\flac} tolerates crash failures by adding transaction logs to non-volatile storages.
In recovery, we provide a set of rules that help crashed nodes resolve half-executed transactions without violating safety properties.
We design {\flacFF} to take a similar logging strategy as 2PC~\cite{2PC}.
It logs the votes and decisions of participants as $ready-yes/no$ or $commit/abort$ before sending votes to others or making local decisions, respectively.
{\flacCF} adds $ready-yes/no$ and $commit/abort$ logs alike {\flacFF}, and follow ~\cite{gupta2018easycommit} to add $transit-commit/abort$ logs before broadcasting the decision.
We exclude the recovery process of {\flacNF} since it reuses EC.
The following paragraphs show how {\flac} performs the crash recovery of the coordinator and participants.

When a coordinator node recovers from crash failure, it first finds all half-executed transactions recorded in $propose$ logs, temporarily blocks accesses to data records updated by them, and resolves these transactions following three steps:
\begin{itemize}
    \item For each half-executed transaction $T$, the coordinator first checks if $\Pi_{T} = ${\flacCF} and no $transit-commit$ log is detected locally. If this check succeeds, by Algorithm~\ref{alg:CO2}, no participant node could decide {\commit}. Thus, the coordinator can directly abort the transaction.
    \item Otherwise, the coordinator queries $T's$ logs on participants in {\nodesT}. If a decision log is found, the coordinator decides accordingly to ensure the agreement property. If there exist alive participant nodes in {\nodesT} that do not include $ready-no$ logs, the coordinator directly decides {\abort} since the transaction is not committable without all {\yes} votes (validity property).
    \item Finally, if $T$'s state is still unresolved, the coordinator waits until it gets the $ready-yes$ logs from all {\nodesT} and continues {\flac}'s execution in Algorithm~\ref{alg:CO1} line~\ref{line:AlgCO1AllUndecided} or Algorithm~\ref{alg:CO2} line ~\ref{line:AlgCO2Commit}  to commit the transaction $T$.
\end{itemize}

In recovery, the participant first takes the same process as the coordinator to detect half-executed transactions and blocks concurrent access to involved data records.
It then directly aborts half-executed transactions with $ready-no$ logs by validity property.
If the transaction is still unresolved, the participant attempts to get the final state by communicating with other nodes.
The participant decides the final state of an unresolved transaction $T$ by the following three cases:
\begin{itemize}
    \item If the participant obtains a decision for $T$ from other nodes, it decides accordingly by the agreement property.
    \item Suppose the participant detects an alive participant in {\nodesT} without $ready-yes$ logs for $T$, it directly aborts the transaction by the validity property.
    \item If neither case occurred, the participant waits until it receives $ready-yes$ logs from all {\nodesT} and reaches an agreement with them using a new coordinator.
\end{itemize}

\subsubsection{Protocol Termination}
\label{termination}

%
In {\flac} executions, a node $C_i$ may wait for the messages from others to continue its local execution.
This node $C_i$ may hold back its execution to accommodate all messages from alive nodes before the crash timeout.
If no message arrives at $C_i$ before the crash timeout due to crash failures, $C_i$ will block its local execution and switches to termination protocol, trying to avoid transaction blocking.
Our protocol {\flac} guarantees the non-blocking termination of transactions executed with a correct dedicated protocol that matches the actual operating condition.
Otherwise, {\flac} may block the current transaction execution and wait for some crashed nodes' recovery.
In the following paragraph, we present our termination protocol by listing each coordinator's or participant's behaviors on the crash timeout.
As shown in Algorithms ~\ref{alg:CO1} \ref{alg:LP1} \ref{alg:CO2} and \ref{alg:LP2}, the coordinator can only timeout in {\flacFF} (Algorithm~\ref{alg:CO1} line~\ref{line:AlgCO1Termination}), while the participant can timeout in both {\flacFF} and {\flacCF} (Algorithm~\ref{alg:LP1} line~\ref{line:AlgLP1TerminationProtocol} and Algorithm~\ref{alg:LP2} line~\ref{line:AlgLP2Termination}).
If a coordinator timeouts, it blocks the execution of transaction $T$ locally, requests the transaction state on {\nodesT}, and waits until it (i) receives the decision for $T$ and decides accordingly (by the agreement property), (ii) finds the absence of $ready-yes$ log on one participant and then decides {\abort} (by the validity property), or (iii) receives \result{\yes, \undecided} from all participants and then decides {\commit} by continuing on  Algorithm~\ref{alg:CO1} line~\ref{line:AlgCO1AllUndecided} or Algorithm~\ref{alg:CO2} line ~\ref{line:AlgCO2Commit}.
A participant timeout in {\flacFF} first tries to find the decision logs or the absence of $ready-yes$ logs from other nodes and decides with the same process as the coordinator.
If neither is found on alive nodes, the participant waits until all $\nodesT$ become alive and then reaches an agreement among them by a new coordinator.
Otherwise, if a participant timeouts in {\flacCF}, the transmission-before-decide technique~\cite{gupta2018easycommit} used by {\flacCF} ensures that no node has made a decision for $T$.
Thus, the participant can safely decide {\abort} and terminate the current transaction execution.
%

\subsection{Failure Detection Rules of RLSM Manager}
\label{sec:failureDetectionProof}

We list detailed proofs and explanations for all failure detections in Algorithm~\ref{alg:rlsm}.

\subsubsection{Failure detection for {\flacFF} or {\flacCF} by missing results.}
A crash or a network failure can be detected by checking if the coordinator does
not receive all participants' execution results, i.e., some results are missing,
as formalized in Proposition \ref{prop:CrashFailureFlacFF}.

\begin{proposition}[Failure detection by missing results]
  \label{prop:CrashFailureFlacFF}
  In an execution of {\flacFF} or {\flacCF} for a transaction $T$, if the coordinator
  $C^*$ is alive and fails to collect the transaction result from a participant
  $C_i$, that is $|\resultsT| < |\nodesT|$, then there exist a failure in
  $C_i$'s execution.
  %

\end{proposition}

\begin{proof}
  By contradiction, assume that there is no failure happened in the execution of
  $C_i$.
  Since there is no crash failure and the protocol is executed non-blockingly in \textit{propose} phase,
  $C_i$ will be able to vote and make a decision, and send back its result to
  the coordinator.
  Since there is no network failure, $C_i$'s result will reach the coordinator
  within the network timeout.
  This result contradicts the hypothesis that $C^*$ fails to collect the
  result of $C_i$.
\end{proof}

Note that if the current protocol is {\flacFF}, the failure type detected by
Proposition~\ref{prop:CrashFailureFlacFF} is agnostic before the crash timeout.
In this case, the RLSM manager will first input CF events to the RLSMs of
non-responsive participants, who are likely to crash (line
\ref{line:AlgRlsmCFDetect}).
Then, it will input NF events to the RLSMs of participants who reply late (line
\ref{line:AlgRlsmAsyNF4FF}), since these late replies are likely caused by
network failures.
When the current protocol is {\flacCF}, which tolerates crash failures, then
failures detected by Proposition \ref{prop:CrashFailureFlacFF} can only be
network failures.
Like the previous case, the RLSM manager can just input NF events to the
RLSMs of participants who reply late (line \ref{line:AlgRlsmAsyNF4CF}).

\subsubsection{Failure detection for {\flacFF} by {\undecided} decision.}
Network failures between participants in an execution {\flacFF} can be detected by
examining if the result $\result{\yes, \undecided}$ is returned to the
coordinator (lines
\ref{line:AlgRlsmCheckNF4flacFF}--\ref{line:AlgRlsmNF4flacFF}), as formalized as
in Proposition \ref{prop:NetworkFailureFlacFF} below.

\begin{proposition}[Failure detection for {\flacFF} by {\undecided} decision]
  \label{prop:NetworkFailureFlacFF}
  In an execution of {\flacFF} for a transaction $T$, if all nodes can return
  their results ($|\resultsT| = |\nodesT|$), and a node $C_i$ returns
  $\result{\yes, \undecided}$, then a network failure occurs during the
  execution.
\end{proposition}

\begin{proof}
  Assume that there is no network failure in the transaction execution.
  Firstly, according to Algorithm~\ref{alg:LP1}, when $C_i$ returns
  {\result{\yes, \undecided}}, it did not receive all the votes of other
  participants.
  Secondly, since $|\resultsT|$ = $|\nodesT|$, then no participant crashed in
  the \textit{propose} phase.
  Since there is no network failure, votes from all other participants after
  being broadcasted will manage to reach $C_i$ (Algorithm~\ref{alg:LP1}).
  This result contradicts the earlier result that $C_i$ did not receive all
  the votes of other participants.
\end{proof}

Note that during the failure detection for {\flacFF}, the RLSM manager can
pinpoint crash failures to specific participants.
However, if a network failure occurs, it cannot identify precisely the involving
participants, since network failure can happen among a set of nodes.
In this case, the RLSM manager will shift all the participants' states to the
network-failure level (line \ref{line:AlgRlsmNF4flacFF}).

\subsubsection{Network failure detection for {\flacCF} by {\abort} decision.}
Network failures in the execution of {\flacCF} can be detected by checking if
there exist aborted transactions that should be committed (lines
\ref{line:AlgRlsmCheckNF4flacCF}--\ref{line:AlgRlsmNF4flacCF}).
We formalize this detection in Proposition
\ref{prop:NetworkFailureFlacCFFromAborts}.
%

\begin{proposition}[Network Failure Detection for {\flacCF}]
  \label{prop:NetworkFailureFlacCFFromAborts}
  In the execution of {\flacCF} for a transaction, if the coordinator can collect
  the results from all participants ($|\resultsT| = |\nodesT|$), and there
  exists a participant returns $\result{\yes, \abort}$, while no participant
  votes $\no$, then there exist network failures during the execution.
\end{proposition}

\begin{proof}
  By contradiction, suppose that there is no network failure.
  %
  Since the coordinator can collect all the results ($|\resultsT| = |\nodesT|$), no
  participant crashes in the \textit{propose} phase.
  Since no participant votes $\no$, all of them must vote $\yes$.
  Moreover, there is no network failure, thus all participants in $\nodesT$
  should receive {\yes} votes from others and then return $\result{\yes,
    \undecided}$ (Algorithm \ref{alg:LP2}).
  It contradicts the hypothesis that there exists participant return
  $\result{\yes, \abort}$.
 \end{proof}

\subsection{Correctness Proof of RLSM Manager}
\label{sec:rlsmCorrectness}

This section shows the correctness of {\rlsm} in two aspects: soundness (Proposition~\ref{thmSoundness}) and completeness (Proposition~\ref{thmCompleteness}), i.e.
%
%
{\rlsm} does not report false failures, and it detects all failures that cause the abort of committable transactions.

\begin{proposition}
[Soundness of RLSM]
\label{thmSoundness}
In execution of {\flac} for a transition, if the RLSM manager triggers a $CF$
or an $NF$ event, then there is indeed an unexpected failure that occurs during the transaction execution.
\end{proposition}

\begin{proof}
Algorithm~\ref{alg:rlsm} has listed all possible ways to detect failures (for example, $CF$ in line~\ref{line:AlgRlsmCFDetect}).
The following Proposition~\ref{prop:CrashFailureFlacFF}, \ref{prop:NetworkFailureFlacFF}, and ~\ref{prop:NetworkFailureFlacCFFromAborts} have proven the correctness of each detected failure, so that $CF$ can only be triggered by the execution of {\flacFF} when crash or network failures happen, while $NF$ is triggered  by the execution of {\flacCF} only when network failures occur.
\end{proof}

\begin{proposition}[Completeness of RLSM]
  \label{thmCompleteness}
  In an execution of {\flac} for transaction $T$, if there is a crash or
  a network failure occurring that aborts committable transactions, then
  the RLSM manager will always be able to detect it and trigger the event $CF$
  or $NF$.
\end{proposition}

\begin{proof}
  We prove this by contra position.
  Firstly, for {\flacFF}, if it does
  not trigger any $CF$ or $NF$ event (Algorithm \ref{alg:rlsm}), then (i) $\resultsT$ must contain only the
  $\commit$ or only the $\abort$ decision, and (ii) the coordinator can collect
  all the participants' results ($|\resultsT| = |\nodesT|$).
  It follows that no
  crash failure occurs.
  From (i), all participants have decided {\commit} or {\abort}.
  Then, according to Algorithm~\ref{alg:LP1}, each node has received $|\nodesT|$ votes, indicating no network failure affecting the executions of committable transactions.

  Secondly, for {\flacCF}, if it does not
  trigger any $NF$ event (Algorithm \ref{alg:rlsm}),
  then $\resultsT$ must either contain all $\result{\yes, \undecided }$ or contains some $\result{\no, \abort}$. In particular:
  When all nodes replied $\result{\yes, \undecided }$, the coordinator will decide to commit the transaction (Algorithm \ref{alg:LP2}, line
  \ref{line:AlgCO2Commit}), implying that the committable transactions are not affected.
  When some nodes replied $\result{\no, \abort}$, the transaction is not committable.
  Thus, the committable transactions in {\flacCF} are not affected by non-detected failures either.
\qedhere
\end{proof}

\subsection{Correctness Proof of \flac}
\label{sec:flacCorrectness}

We conclude our {\flac} design with a proof of correctness for the overall system formed by all components introduced in this paper, including the two-phase transaction processing system (Section~\ref{sec:systemModel}), dedicated protocols for each operating environment  (Section~\ref{sec:subProtocols}), and the {\rlsm} manager (Section~\ref{rlsm}).

\begin{proposition}[Correctness]
{\flac} ensures overall safety, liveness (Definition~\ref{def1}), and durability.
\end{proposition}

\begin{proof}
{\flac} ensures overall safety by executing transactions with commit protocols that have been proven to be safe (Proposition~\ref{prop:Safety}).
It sticks to {\flacFF}, {\flacCF}, or {\flacNF} in each execution and thus avoids the problems caused by interleaving executions.
For liveness, {\flac} continuously detects unexpected failures that could cause transactions to be aborted or blocked (Proposition~\ref{thmSoundness} and ~\ref{thmCompleteness}) and adjusts to a more suitable protocol for eventual liveness (Proposition~\ref{prop:Liveness} and Algorithm~\ref{alg:rlsm}).
Finally, by persisting the transaction logs to non-volatile storages and designing ways to recover nodes from failures with these logs (Section~\ref{recovery}), {\flac} ensures durability.
\end{proof}


\subsection{Extending \flac~ to Support Replicated Systems}
\label{app:discussion}

This section discusses {\flac}'s position regarding replicated systems.
Classical database systems depend on highly available hardware, like high-end servers, to ensure system availability.
However, these systems cannot guarantee both availability and data consistency following node or network failures.
Existing database systems~\cite{corbett2013spanner, DBLP:journals/pvldb/YangYHZYYCZSXYL22} avoid this problem by introducing cross-region replicas synchronized by consensus protocols.
The atomic commit protocols in these systems are either separated ~\cite{corbett2013spanner} or combined~\cite{zhang2018building, maiyya2019unifying, DBLP:journals/pvldb/YangYHZYYCZSXYL22} from underlying consensus protocols.
Both frameworks achieve high availability at the cost of increased transaction normal processing overhead due to the use of extra nodes and the extra message round trips brought by replications.
Compared to these replicated commit protocols, {\flac} is not designed for systems that require extremely high availability.
Instead, it strikes a balance between system availability and efficiency.
For availability, it immediately adjusts the protocol to a higher robustness level upon detecting unexpected failures, minimizing the protocol's blocking.
For efficiency, {\flac} only executes transactions with costly but highly available protocols when necessary.
We compare it with G-PAC, a state-of-the-art replicated commit protocol, to show their difference in Appendix~\ref{sec:experimentApp}.
In all, {\flac} can be customized into replicated systems with the following changes:
%
\begin{itemize}
    \item In {\flacFF} and {\flacCF}, the transaction commit only involves the leader node of each replica group. The leader node that serves as a participant broadcasts its decision to all replicas asynchronously before finishing its local transaction branch.
    \item {\flacNF} is changed to a replicated commit protocol that provides high availability at the cost of a longer normal processing time (as shown by our experiments in Appendix~\ref{sec:experimentApp}).
    \item When upward transitions towards {\flacNF}, each leader first synchronizes its local decisions to its replicas group. The parameter {\numRunNF} is increased to leverage the extra cost brought by such synchronizations in upward transitions.
\end{itemize}
%
%


\newpage
\section{Extra Experiments for FLAC}
\label{sec:experimentApp}

\subsection{Effect of Network Delays on FLAC}
\label{sec:YCSBVaryNetworkDelay}
Figure \ref{fig:delay} presents the throughput and latency of {\flac}, 2PC, 3PC,
C-PAC, EC under the varying settings of network delay.
{\flac} performs better than all other protocols when the network delay is
higher than 5ms, and its advantage over the others increases as the network delay
increases.
As the network delay increases from 5ms to 40ms, {\flac} achieves a 31.0\% to
49.6\% increase in throughput and 34.4\% to 59.0\% decrease in latency compared
to EC and 2PC.
There are two reasons for this achievement of {\flac}.
Firstly, {\flac} mainly run the protocol {\flacFF} in this failure-free
environment.
Secondly, {\flacFF} can commit transactions more efficiently with lesser message
delays on the participant side than all other protocols (Table~\ref{tab:cp}).
In this experiment, all the protocols exhibit performance degradation as the
network delay increases.
This phenomenon is due to the fact that when the network delay is higher, the
transaction blocking and execution time will be longer.

\begin{figure}[h]
  \centering
  \subfigure[Throughput.]{%
    \includegraphics[width=0.48\linewidth]{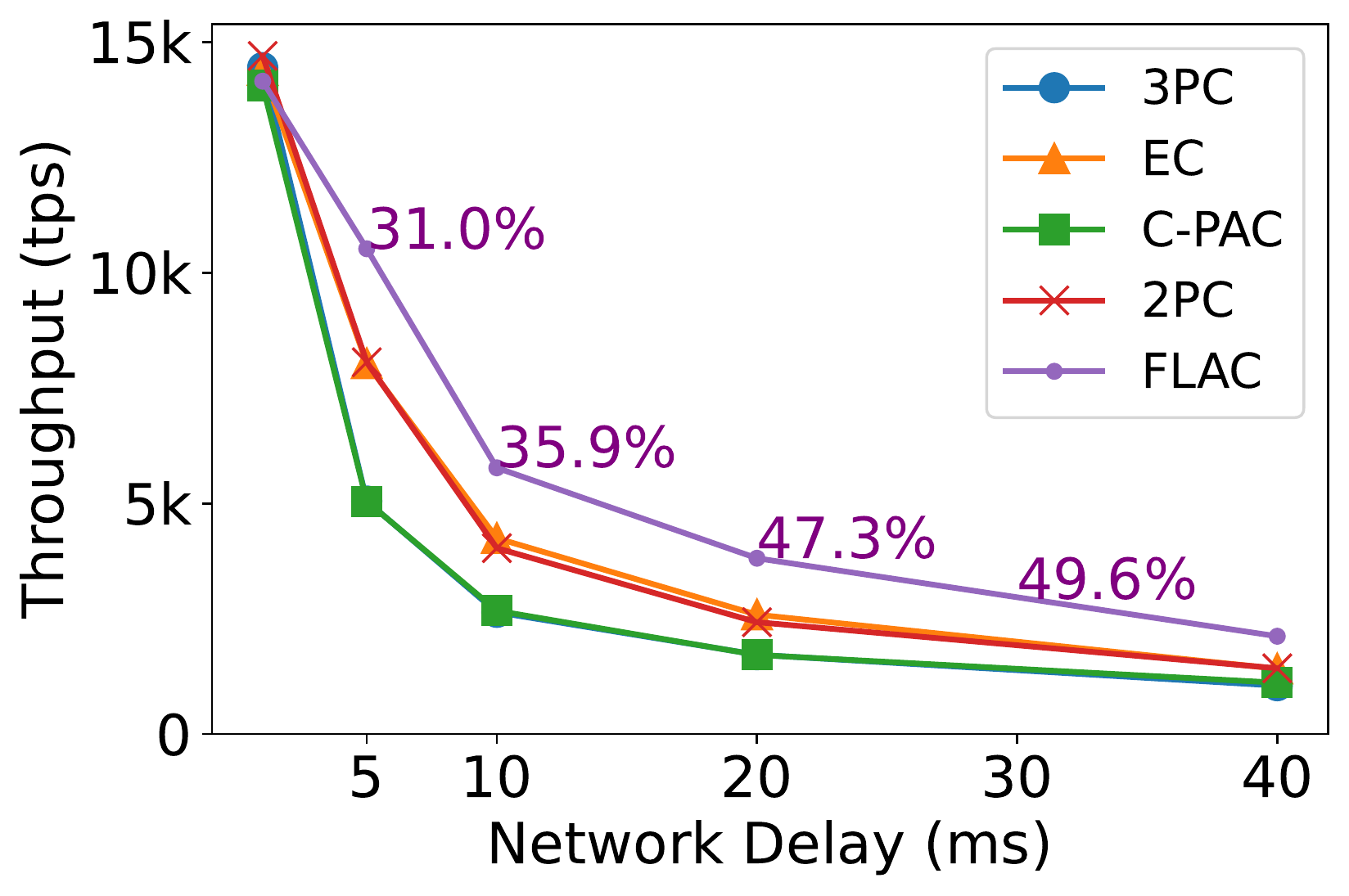}
  \label{fig:delay-th}}
  \subfigure[Tail latency.]{%
    \includegraphics[width=0.48\linewidth]{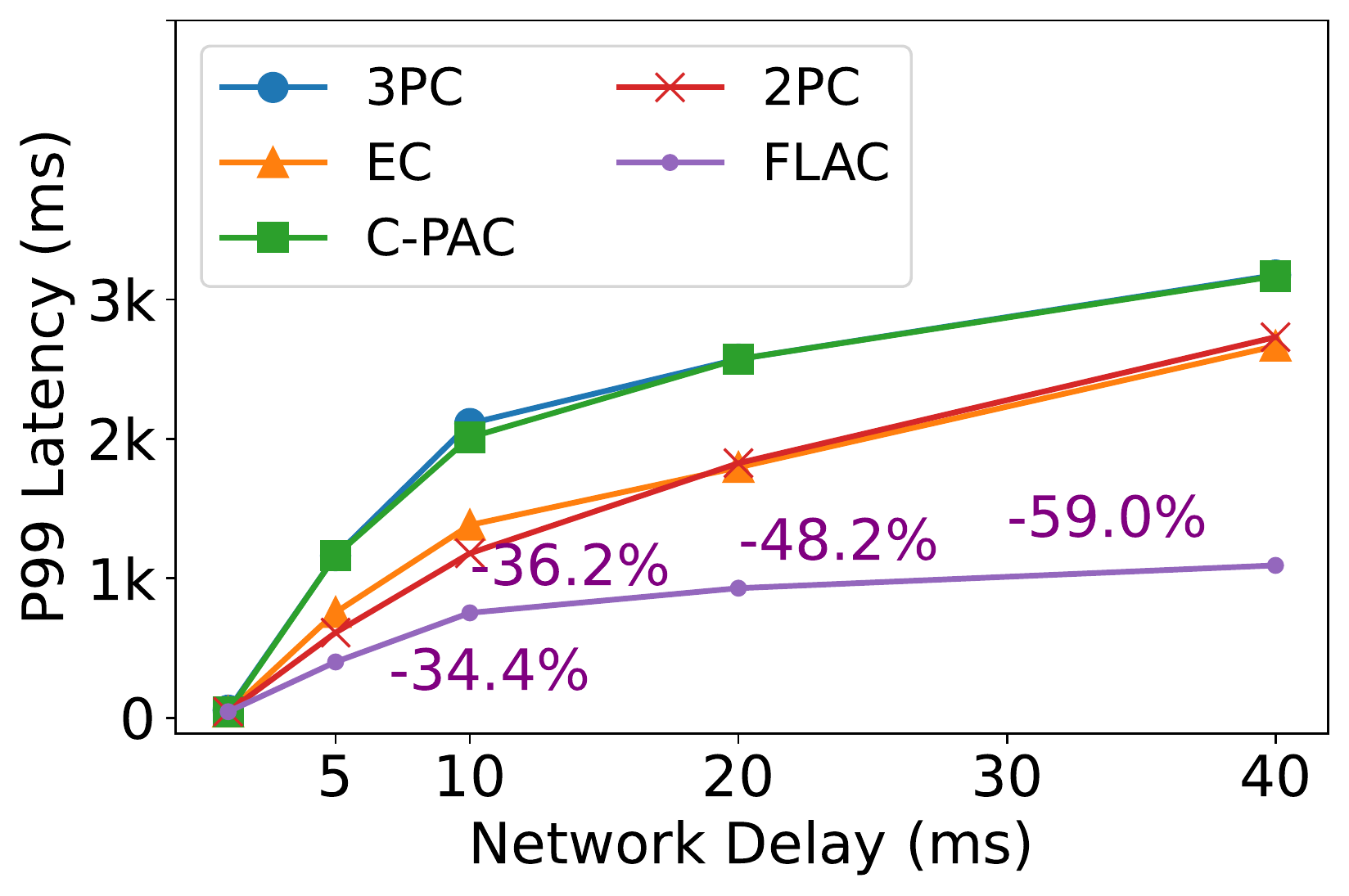}
  \label{fig:delay-la}}
  \vspace{-1em}
  \caption{Performance of all protocols under the YCSB-like
    micro-benchmark with varying network delay.}
  \label{fig:delay}
\end{figure}

\subsection{Sensitivity Analysis of All Protocols}

In the previous experiments, all protocols were evaluated with only cross-shard
transactions, which is the default setting of the YCSB-like micro-benchmark
(Section \ref{setup}).
However, it is well-known that cross-shard transactions can affect the system
throughput due to its blocking effect on single-shard transactions
\cite{gupta2018easycommit, maiyya2019unifying, chrysanthis1998recovery}.
Therefore, we conduct this experiment to assess the protocols' performance in a
mixed setting where both single-shard and cross-shard transactions can occur,
with the cross-shard transaction percentage ranging from 0\% (all transactions
are single-shard) to 100\% (all transactions are cross-shard).

As can be seen in Figure~\ref{fig:cross},
when all transactions are single-shard, all protocols achieve the same high
throughput and low latency since single-shard transactions are committed directly~\cite{gupta2018easycommit, lu2021epoch}.
When the percentage of cross-shard transactions increases, the throughput of all
the protocols decrease due to the long occupation of resources by such
transactions.
In this experiment, {\flac} can still perform better than all other protocols in
all settings when there are cross-shard transactions.
Again, this result is due to the low latency of the protocol {\flacFF} on
the participant side: it requires only 1 message delay on the participant, while
2PC, 3PC, EC, C-PAC require at least 2 message delays (Table~\ref{tab:cp}).

\subsection{Message Complexity Side-Effects of FLAC}

As mentioned in Section~\ref{sec:comparison}, the network bandwidth could become the bottleneck as the number of involved participants increases.
Our experiment in Section~\ref{exp} has shown that {\flac}'s throughput gets close to others as 10 participants are involved in transactions.
To further demonstrate the side effect of increased message complexity,  we let transactions only access one record on each participant and added the number of participants to 15, to amplify the impact of message complexity.
It can be observed that {\flac} and EC both perform better than others before 10 participants per transaction but act slightly worse than 2PC in throughput afterwise (0.97x and 0.90x for {\flac} and EC, respectively).
It presents the side-effect of higher message complexity of two protocols.
Such a problem does not affect {\flac} under most OLTP workloads since the number of participants per transaction is less than 10 in practice~\cite{council2010tpc, gupta2018easycommit}.

\begin{figure}[ht]
  \centering
    \subfigure[All transactions.]{%
    \includegraphics[width=0.48\linewidth]{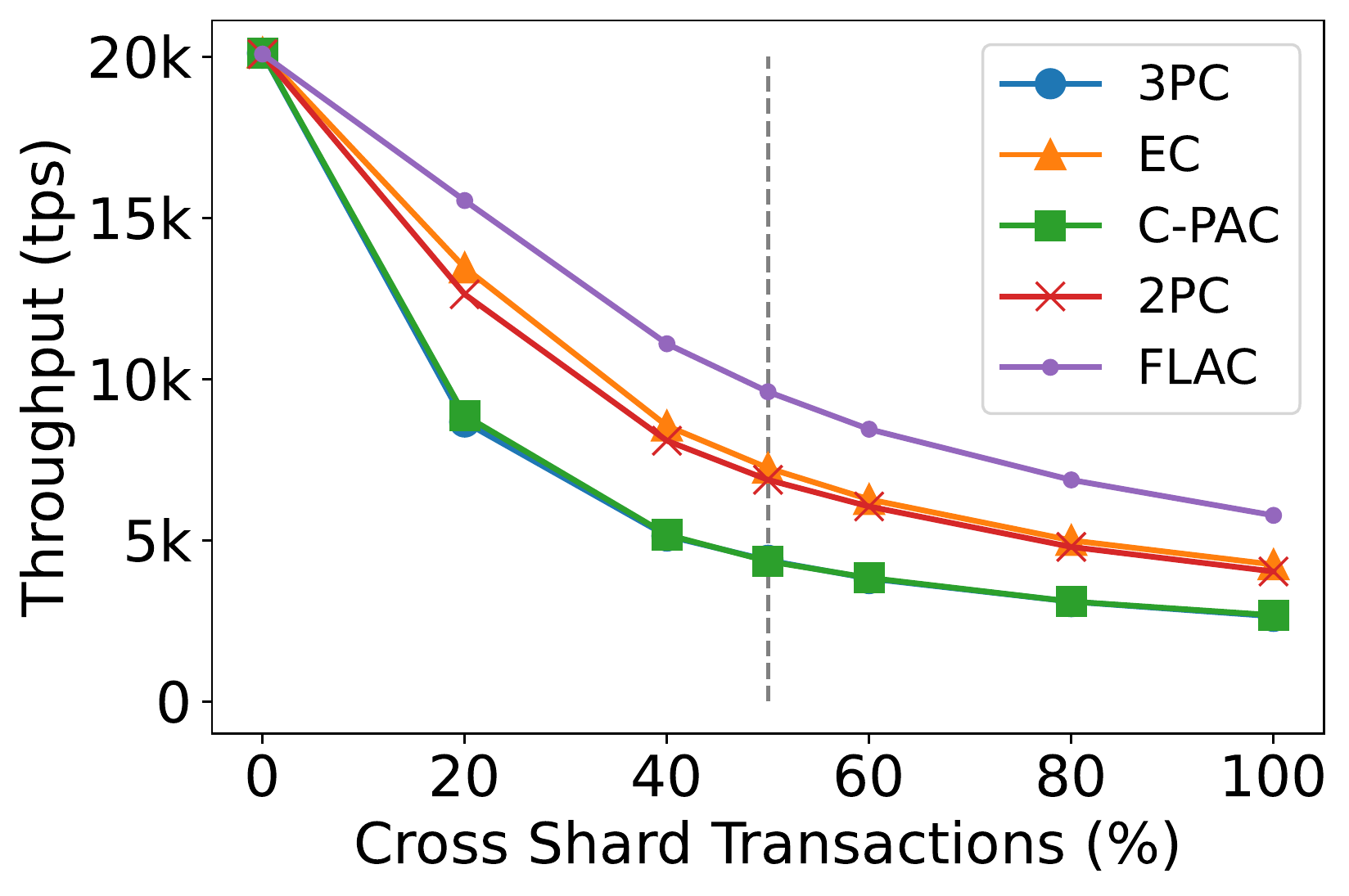}
    \label{fig:cross-th}}
  \subfigure[Only single-shard transactions.]{%
    \includegraphics[width=0.48\linewidth]{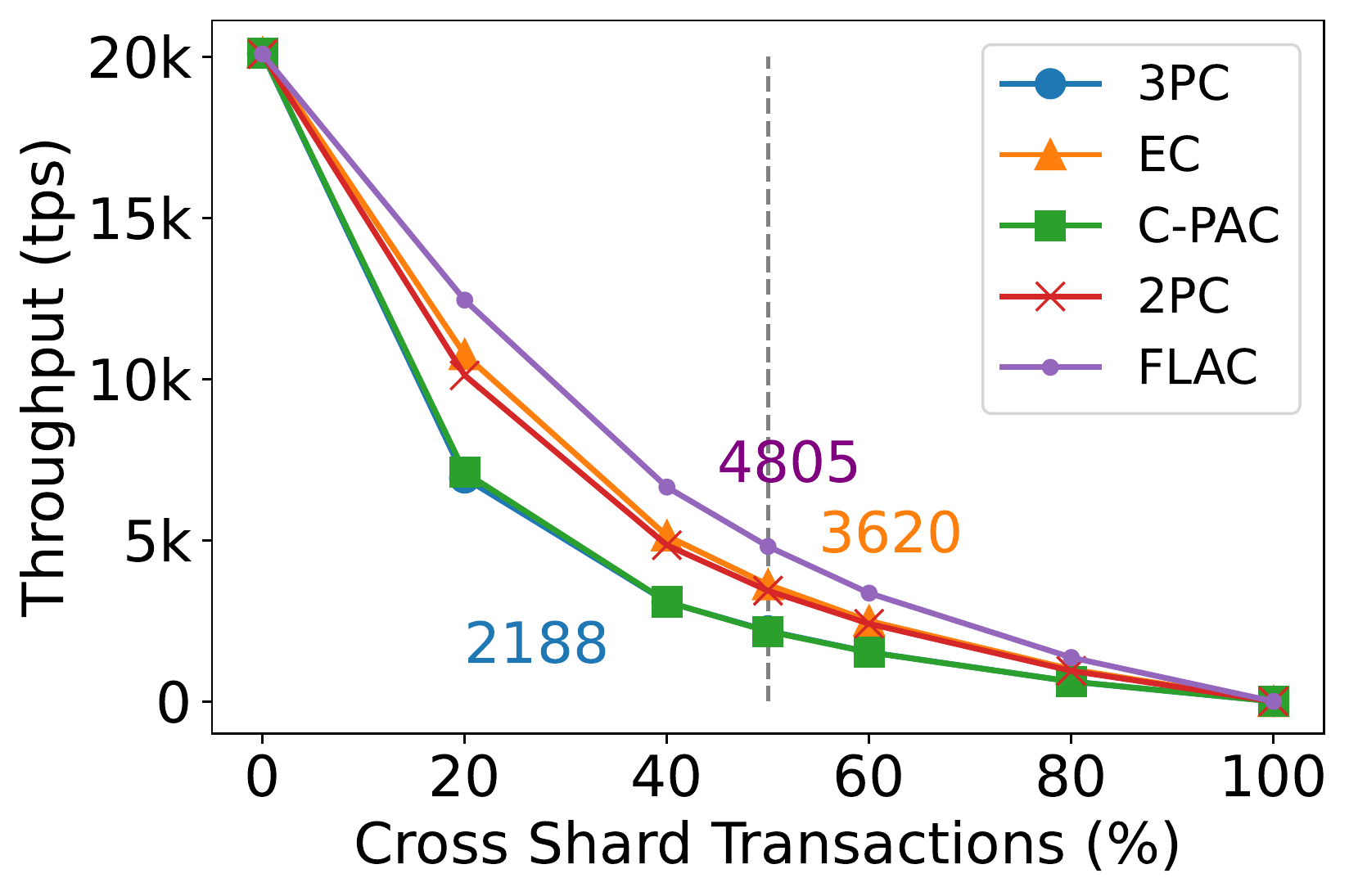}
    \label{fig:cross-only}}
  \caption{Performance under the YCSB-like micro-benchmark, with
    varying percentage of cross-shard transactions.}
  \label{fig:cross}
\end{figure}

\begin{figure}[htbp]
  \centering
    \subfigure[Throughput.]{\includegraphics[width=0.48\linewidth]{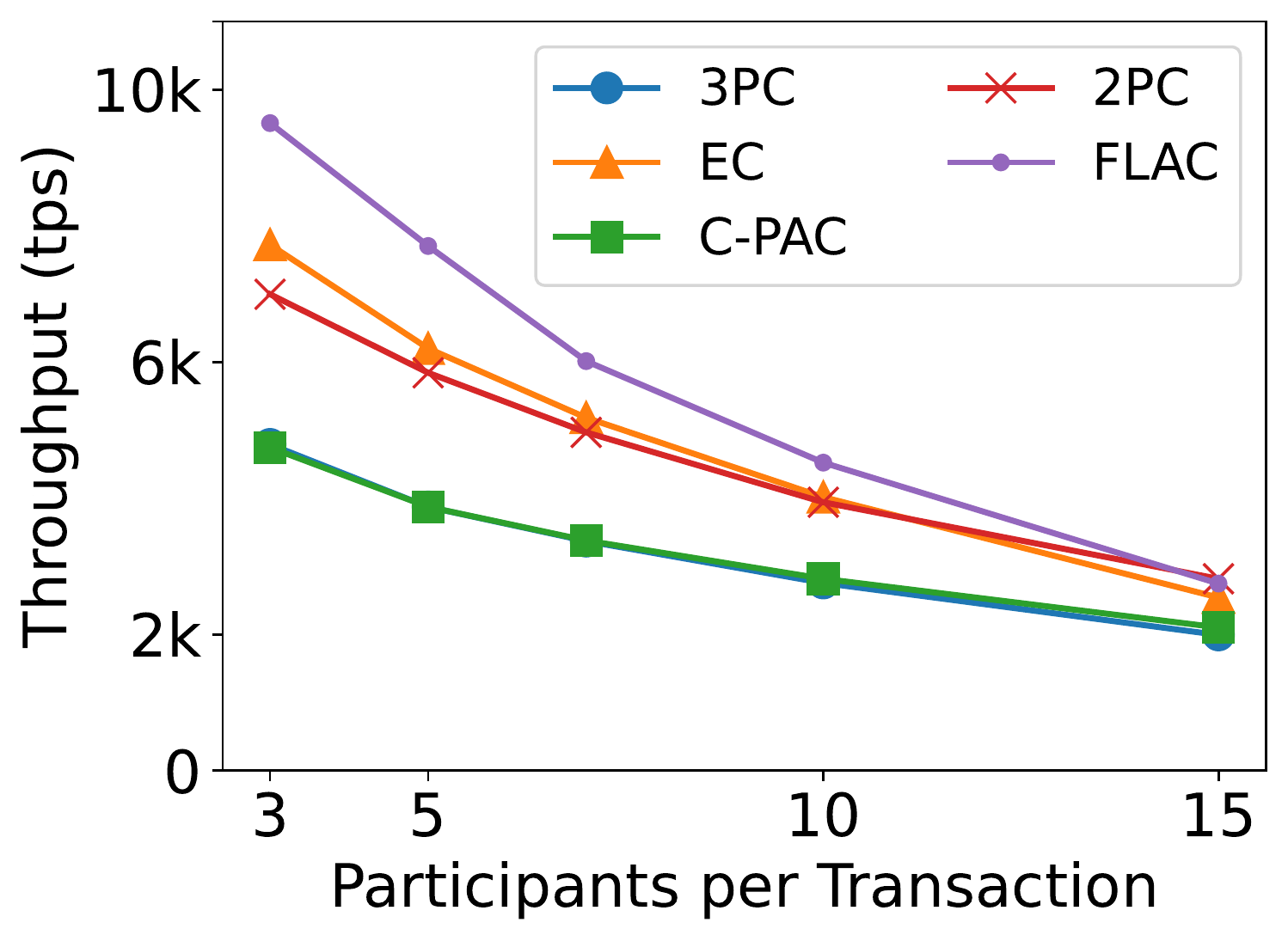}
    \label{fig:scale-more-th}}
  \subfigure[Tail latency.]{%
    \includegraphics[width=0.48\linewidth]{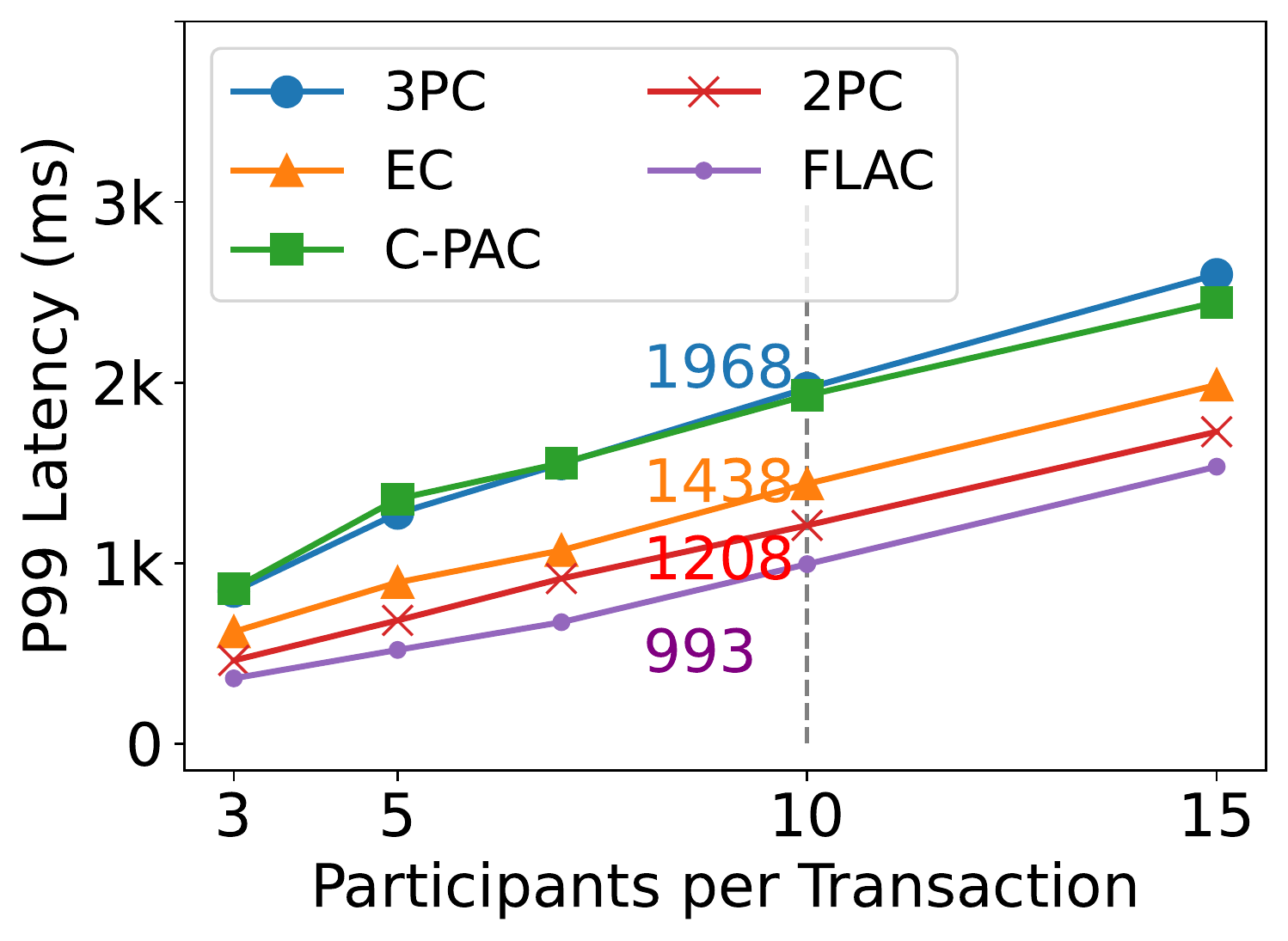}
    \label{fig:scale-more-la}}
  \caption{Performance under the YCSB-like micro-benchmark with
    varying numbers of participants.}
  \label{fig:scale-more}
\end{figure}

\subsection{Comparison between Non-Replicated and Replicated Protocols}
This experiment shows a comparison between {\flac} and a state-of-the-art replicated commit protocol G-PAC~\cite{maiyya2019unifying}.
We reuse the replicated setting in G-PAC's original paper~\cite{maiyya2019unifying} to replicate all three shards across three participant nodes.
As Figure~\ref{fig:ycsb-more} shows, G-PAC has 65.3\% lower throughput with 2.68x latency compared to {\flac}.
It reveals that replicated commit protocols like G-PAC achieves higher availability at the cost of normal transaction processing efficiency.
Thus, we design {\flacFF} and {\flacCF} for non-replicated settings and let it only involve one node per replica group when customized for replicated systems (Section~\ref{app:discussion}).

\begin{figure}[htbp]
  \centering
    \subfigure[Throughput.]{\includegraphics[width=0.48\linewidth]{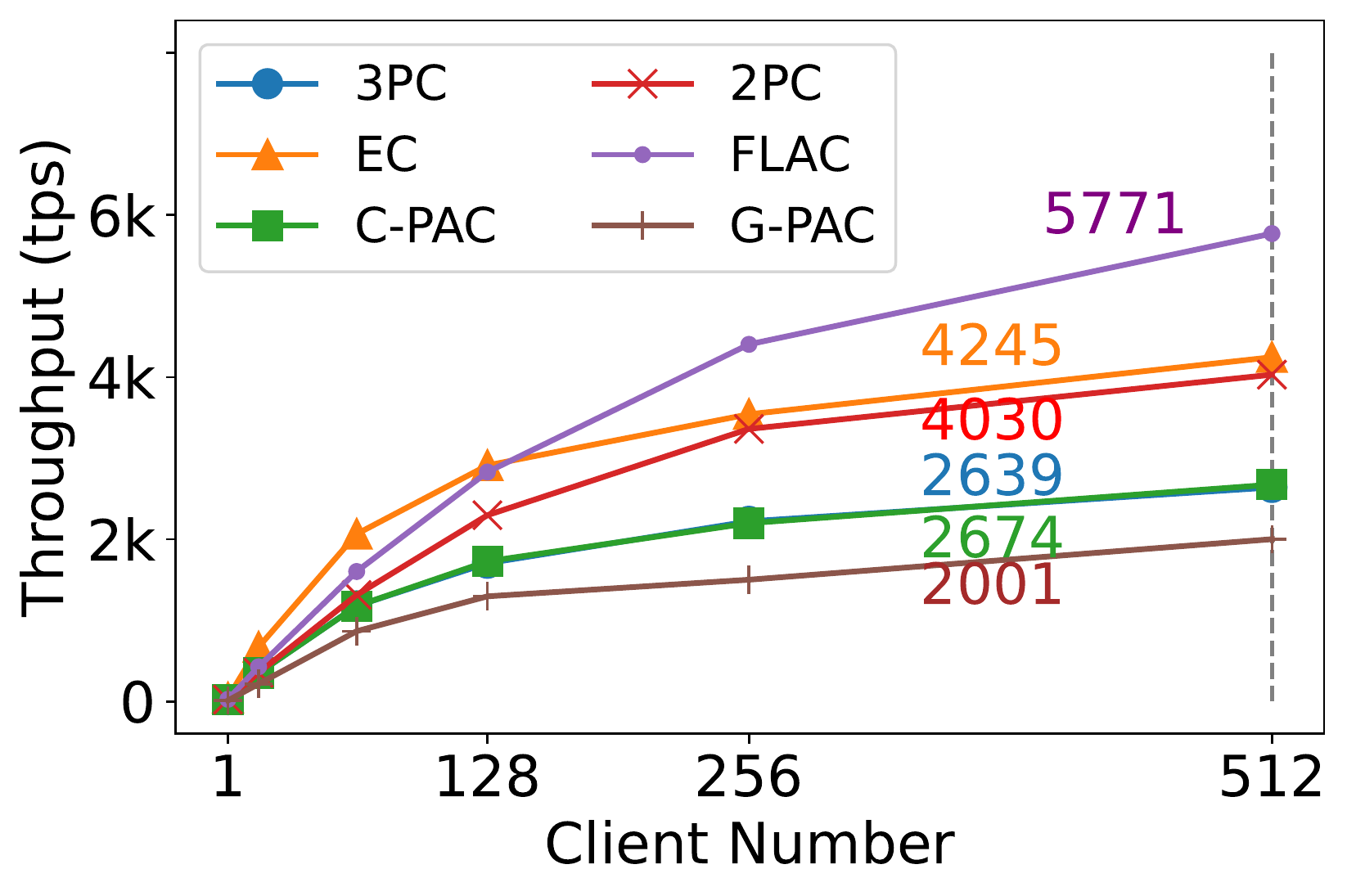}
    \label{fig:ycsb-more-th}}
  \subfigure[Tail latency.]{%
    \includegraphics[width=0.48\linewidth]{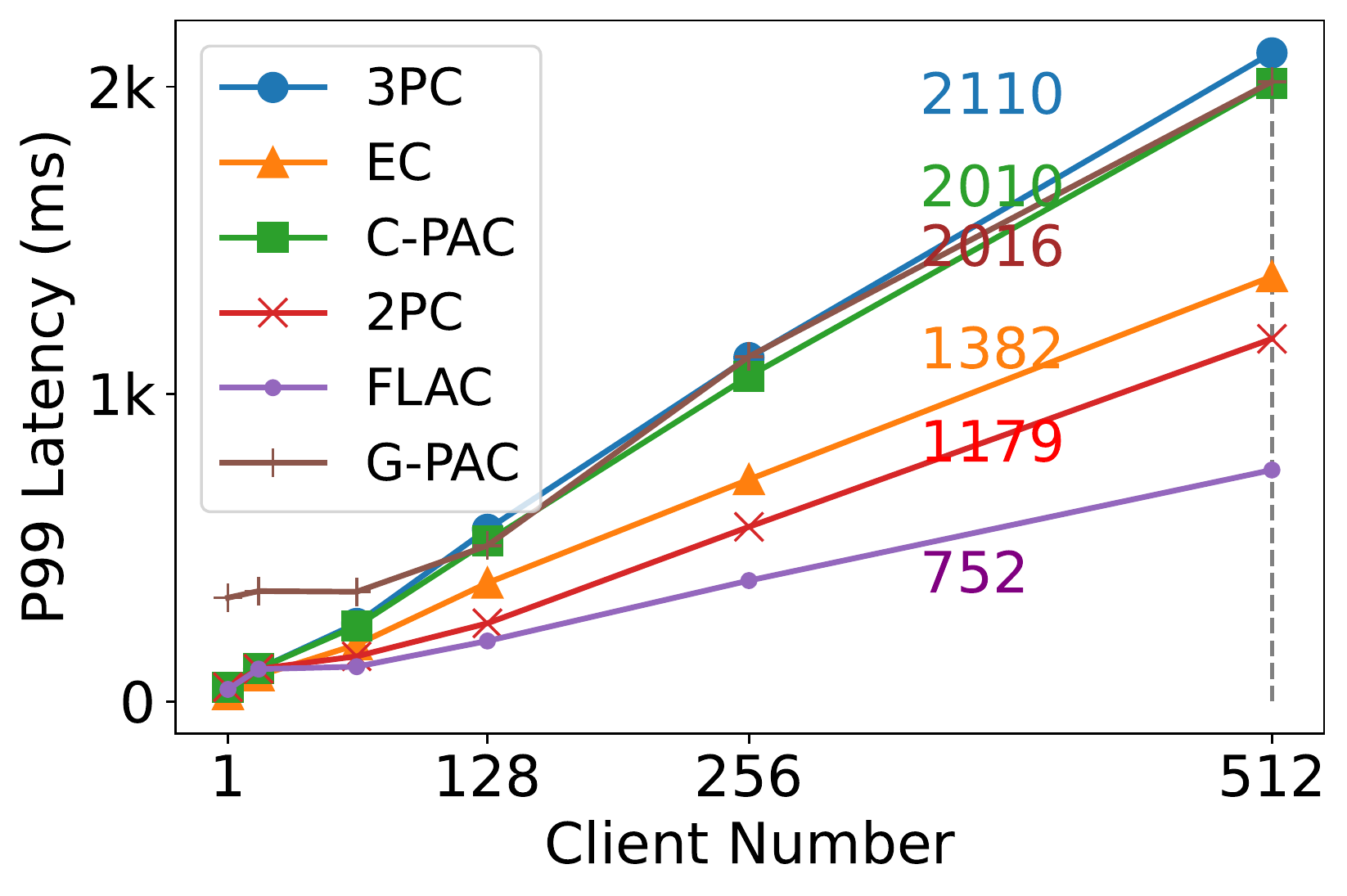}
    \label{fig:ycsb-more-la}}
  \caption{Performance under the YCSB-like micro-benchmark with
    varying numbers of clients.}
  \label{fig:ycsb-more}
\end{figure}

\subsection{Performance of All Protocols on TPC-C Benchmark of Real Workloads}

Figure~\ref{fig:tpc} presents the performance of all protocols on the TPC-C
benchmark of real workload.
In this experiment, {\flac} still achieves the best performance.
In particular, it achieves 1.35x to 2.22x throughput speedup compared with the
other protocols, while its tail latency is only 35.4\% to 65.0\% of them.
Compared to the previous experimental results on the YCSB-like micro-benchmark
(Figure~\ref{fig:ycsb}), all the protocols achieve higher throughput and lower
latency for the TPC-C benchmark.
An explanation for this result is that the TPC-C benchmark contains many
single-shard transactions, which can be quickly processed by all protocols,
unlike the YCSB-like micro-benchmark which contains only cross-shard
transactions, as described earlier in Section \ref{setup}.

\begin{figure}[ht]
  \centering
    \subfigure[Throughput.]{%
    \includegraphics[width=0.48\linewidth]{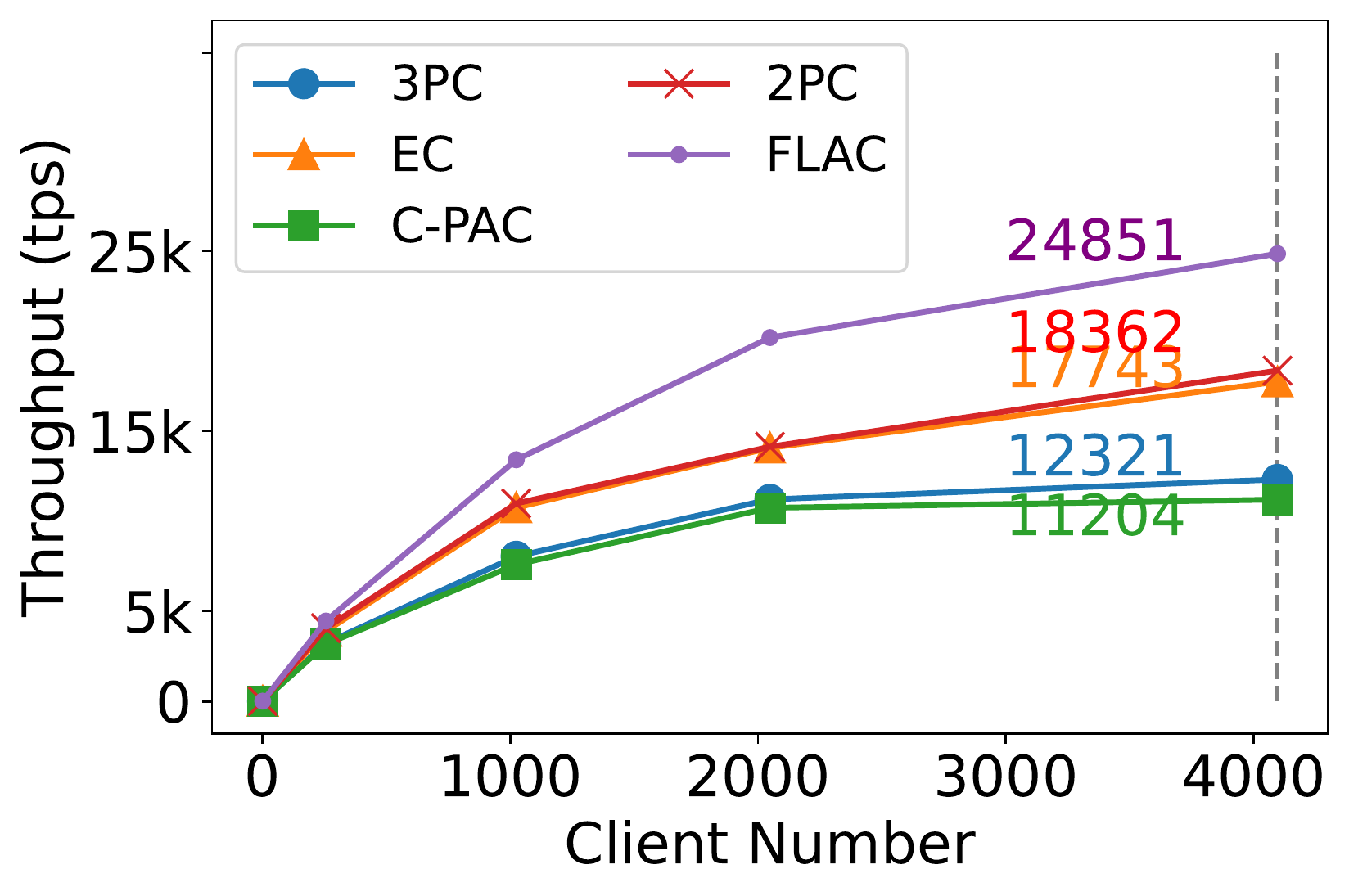}
    \label{fig:tpc-th}}
  \subfigure[Tail latency.]{%
    \includegraphics[width=0.48\linewidth]{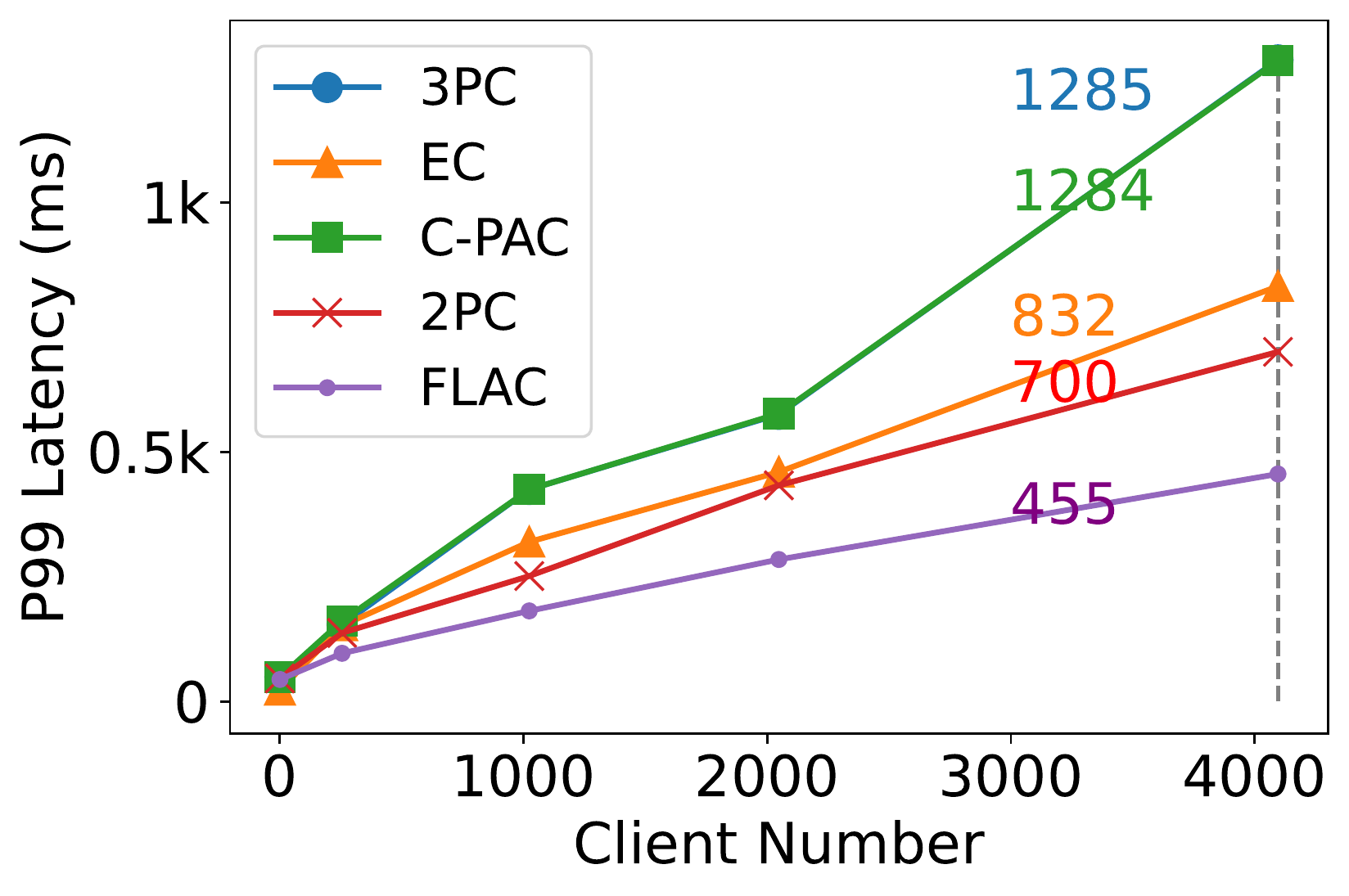}
    \label{fig:tpc-la}}
  \vspace{-1em}
  \caption{Performance of all protocols under TPC-C macro-benchmark,
    with varying numbers of clients.}
  \label{fig:tpc}
\end{figure}

\end{document}